\documentclass[english]{article}
\usepackage[T1]{fontenc}
\usepackage[latin9]{inputenc}
\usepackage[letterpaper]{geometry}
\usepackage{amsmath}
\usepackage{amssymb}
\usepackage{amsthm}
\usepackage{esint}
\usepackage[english]{babel}
\usepackage[pdftex]{graphicx} 
\usepackage{verbatim}
\numberwithin{equation}{section}

\usepackage{hyperref} 

\newtheorem{Theorem}{Theorem}[section]
\newtheorem{Definition}[Theorem]{Definition}
\newtheorem{Proposition}[Theorem]{Proposition}
\newtheorem{Lemma}[Theorem]{Lemma}

\newtheorem{Remark}[Theorem]{Remark}

\newcommand{\RR}{\mathbb{R}}
\newcommand{\NN}{\mathbb{N}}
\newcommand{\ZZ}{\mathbb{Z}}
\newcommand{\nr}{{\text{near}}}
\newcommand{\fr}{{\text{far}}}

\newcommand{\nn}{\nonumber}
\newcommand{\wt}{\widetilde}
\DeclareMathOperator{\sinc}{sinc}
\DeclareMathOperator{\spec}{spec}
\newcommand{\eff}{{\rm eff}}
\def\cprime{$'$}

\begin{document}

\title{Oscillatory and localized perturbations of periodic structures and the bifurcation of defect modes}
\author{V. Duch\^ene$^1$, I. Vuki\'cevi\'c$^2$ and M. I. Weinstein$^{2,3}$ \\
{\small $^1$ Institut de Recherche Math\'ematique de Rennes, France} \\
{\small $^2$ Department of Applied Physics and Applied Mathematics, Columbia University}\\
{\small $^3$ Department of Mathematics, Columbia University}}

\maketitle

\abstract{ Let $Q(x)$ denote a periodic function on the real line. The Schr\"odinger operator, $H_Q=-\partial_x^2+Q(x)$, has $L^2(\RR)-$ spectrum equal to the union of closed real intervals
 separated by open spectral gaps. 
 In this article we study the bifurcation of discrete eigenvalues (point spectrum) into the spectral gaps for the operator $H_{Q+q_\epsilon}$, where $q_\epsilon$ is spatially localized and highly oscillatory in the sense that its Fourier transform, $\widehat{q}_\epsilon$ 
  is concentrated at high frequencies. Our assumptions imply that $q_\epsilon$ may be pointwise large but $q_\epsilon$ is small in an average sense.  For the special case where
 $q_\epsilon(x)=q(x,x/\epsilon)$ with $q(x,y)$ smooth, real-valued, localized in $x$, and periodic or almost periodic in $y$, the bifurcating eigenvalues are at a distance of order $\epsilon^4$ from the lower edge of the spectral gap. We obtain the leading order asymptotics of the bifurcating eigenvalues and eigenfunctions. Consider the $(b_*)^{th}$ spectral band  ($b_*\ge1$) of $H_Q$. Underlying this bifurcation is an effective Hamiltonian associated with the lower spectral band edge: 
 $H^\epsilon_{\eff}=-\partial_x A_{b_*,\eff}\partial_x - \epsilon^2 B_{b_*,\eff} \times \delta(x)$ where $\delta(x)$ is the Dirac distribution, and effective-medium parameters $A_{b_*,\eff},B_{b_*,\eff}>0$ are explicit and independent of $\epsilon$.
  The potentials we consider are a natural model for wave propagation in a medium with localized, high-contrast  and rapid fluctuations in material parameters about a background periodic medium.
 }

\section{Introduction}\label{sec:introduction}

Let $Q(x)$ denote a one-periodic function on the real line:
\begin{equation}
Q(x+1)\ =\ Q(x),\ x\in\RR.
\label{Qperiodic}
\end{equation}
  The Schr\"odinger operator, 
\begin{equation}
H_Q\ =\ -\partial_x^2 + Q(x),
\label{HQ}
\end{equation}
has $L^2(\RR)-$ spectrum equal to the union of closed real intervals (spectral bands) separated by open spectral gaps. It is known that a spatially localized and small perturbation of $H_Q$, say $H_{Q}+\epsilon V$, where $V\in L^1$, induces the bifurcation of discrete eigenvalues (point spectrum) from the edge of the continuous spectrum (zero energy) into the spectral gaps at a distance of order $\epsilon^2$ from the edge of spectral bands; see, {\it e.g.}~\cite{simon1976bound,Gesztesy-Simon:93,DVW:14b}. In this article we study the bifurcation of discrete spectrum for the operator $H_{Q}+q_\epsilon$, where $q_\epsilon$ is localized in space and such that its Fourier transform is concentrated at high frequencies.   A special case we consider is:
 $q_\epsilon(x)=q(x,x/\epsilon)$, where $q(x,y)$ is smooth, real-valued, localized in $x$ and periodic or almost periodic in $y$. In this case, $q_\epsilon$ tends to zero weakly but not strongly. 
 
  Our motivation for considering such potentials is the wide interest in wave propagation in media
  (i)   whose material properties vary rapidly on the scale of a characteristic wavelength of propagating waves and (ii) whose material contrasts are large. We model rapid variation by assuming that the leading-order component of the perturbation $q_\epsilon$ is supported at ever higher frequencies (asymptotically as $\epsilon \downarrow0$), and we allow for high contrast media by not requiring smallness on the $L^\infty$ norm of $q_\epsilon$.
Such potentials have some of the important features of high contrast micro- and nano-structures
 (see {\it e.g.} \cite{Joannopoulos:08}, \cite{Saleh-Teich:91}) and, more generally, wave-guiding or confining media with a multiple scale structure.

 We obtain detailed leading order asymptotics of bifurcating eigenvalues and their associated eigenfunctions, with error bounds,  in the limit as $\epsilon$ tends to zero.  The present article generalizes our earlier work~\cite{DVW:14b, DVW:14a} for the case $Q\equiv0$ (homogeneous background medium) and for $H_{Q}+\epsilon V$, where  $Q$ is taken to be non-trivial and periodic and $\epsilon V$ is small and localized in space.

 Standard homogenization theory (averaging, in this case), which often applies in  situations of strong scale-separation, does not capture the key bifurcation phenomenon. This was discussed in detail in~\cite{DVW:14a}. Underlying the bifurcation is an {\it effective Dirac distribution potential well}; the bifurcation at the lower edge of the $b_*^{th}$ spectral band of $H_Q$ ($b_*\ge1$) is governed by
  an effective Hamiltonian $H^\epsilon_{\eff}=-\partial_x  A_{b_*,{\eff}} \partial_x - \epsilon^2 B_{b_*,{\eff}} \times \delta(x)$. Here, $ A_{b_*,{\eff}} ,B_{b_*,{\eff}}>0$ are independent of $\epsilon$ and are given explicitly in terms of $Q$, $q_\epsilon$. This reveals the leading-order location of the bifurcating eigenvalue at a distance $\mathcal{O}(\epsilon^4)$ from the spectral band edge.

\subsection{Discussion of results}\label{background}

  To describe our results in greater detail, we first present a short review of the spectral theory of $H_Q$;
see, for example,~\cite{eastham1973spectral,RS4}.
   The spectrum is determined by the family of self-adjoint {\it $k-$ pseudo-periodic} eigenvalue problems, 
 parametrized by the {\it quasi-momentum} $k\in (-1/2,1/2]$: 
 \begin{align}
 H_Q u(x;k)\ =\ E\ u(x;k) \ ,\label{HQuEu}\\
 u(x+1;k)\ =\ e^{2\pi i k}\ u(x;k) \ .\label{u-pseudoper}
 \end{align}
 For each $k\in (-1/2,1/2]$, ~\eqref{HQuEu}-\eqref{u-pseudoper} has discrete sequence of eigenvalues:
\begin{equation}
E_0(k) \le E_1(k)\le\dots\le E_b(k)\le\dots,
\label{eigs}
\end{equation}
listed with multiplicity, 
and corresponding $k-$ pseudo-periodic normalized eigenfunctions:
\begin{equation}
u_b(x;k)\ =\ e^{2\pi ikx}\ p_b(x;k),\ \ p_b(x+1;k)\ =\ p_b(x;k),\ \ b\ge 0.
\label{ub-def}
\end{equation}
The $b^{\rm th}$ spectral band is given by
$\mathcal{B}_b\ =\ \bigcup_{\substack{k \in (-1/2,1/2]}} E_b(k)$.
The spectrum of $H_Q$ is given by:
$\spec(H_Q) = \bigcup_{b\ge 0} \mathcal{B}_b\ = \bigcup_{b\geq 0}\ \bigcup_{\substack{k \in (-1/2,1/2]}} E_b(k)$.
Since the boundary condition~\eqref{u-pseudoper} is invariant with respect to $k\mapsto k+1$, the functions $E_b(k)$ can be extended to all $\RR$ as periodic functions of $k$. The minima and maxima of $E_b(k)$ occur at $k=k_*\in\{0,1/2\}$; see Figure~\ref{fig:SketchOfSpectrumQ}. If $k_*\in\{0,1/2\}$ and $E_b(k_*)$ is a spectral band endpoint, bordering on a spectral gap, then $E_b(k_*)$ is a simple $k_*-$ pseudo-periodic eigenvalue, $\partial_kE_b(k_*)=0$, and $\partial_k^2E_b(k_*)$ is either strictly positive or strictly negative; see Lemma~\ref{lem:band-edge}.
 
Consider now the perturbed operator
$H_{Q+V}$, where $V(x)$ is sufficiently localized in $x$. By Weyl's theorem on the stability of the %%
essential spectrum, one has $\spec_{\rm ess}(H_{Q+ V})=\spec_{\rm ess}(H_{Q})$~\cite{RS4}. 
Therefore, the effect of a localized perturbation is to possibly introduce discrete eigenvalues into the spectral gaps. Note that $H_{Q+V}$ does not have discrete eigenvalues embedded in its continuous spectrum; see~\cite{Rofe-Beketov:64},~\cite{Gesztesy-Simon:93}.\bigskip

 Theorem~\ref{thm:Qzero} ($Q\equiv0$) and Theorem~\ref{thm:per_result} ($Q$ non-trivial periodic)  are our main results 
 on bifurcation of discrete eigenvalues of $H_{Q+q_\epsilon}$ from the left (lower) band edge into spectral gaps of $H_Q$. They apply to $q_\epsilon$ spatially localized and spectrally supported at ever higher frequencies as $\epsilon\downarrow0$ (hence weakly convergent as $\epsilon\downarrow0$).
  In this introduction, we state for simplicity the results for the particular case of $Q$ periodic and $x\mapsto q_\epsilon(x)$ a two-scale function (spatially localized on $\RR$ on the slow scale and almost periodic on the fast scale) of the form: 
 \begin{align}\label{eq:q-eps-def}
 q_{\epsilon}(x) = q\left( x, \frac{x}{\epsilon} \right) = \sum_{j \neq 0} q_j(x) e^{2\pi i \lambda_j \frac{x}{\epsilon}}, 
 \end{align}
where the frequencies satisfy the nonclustering assumptions:
\[ \inf_{j \neq l} |\lambda_j - \lambda_l| \geq \theta > 0, \quad  \inf_{j \neq 0} |\lambda_j| \geq \theta > 0\]
for some fixed $\theta>0$. The constraint that $q$ be real-valued implies: $\lambda_{-j}=-\lambda_j$ and $q_{-j}(x)=\overline{q_j(x)}$. The particular case $\lambda_j=j$  corresponds to $y\mapsto q(x,y)$ being  $1-$periodic.

Theorem~\ref{thm:per_result} ($Q$ non-trivial periodic) for the special case~\eqref{eq:q-eps-def}
is the following (see Appendix~\ref{sec:effective-potential})
\begin{Theorem}\label{thm:qeps-specific}
Let  $E_\star = E_{b_*}(k_*)$, $k_* \in \{0,1/2\}$ denote the lower edge of the $b_*^{th}-$ spectral band and assume that this point  borders a spectral gap; see the left panel of  Figure~\ref{fig:SketchOfSpectrumQ}. Assume $q_\epsilon$ is of the form~\eqref{eq:q-eps-def} and $q_j(x)$ is sufficiently smooth and decays sufficiently rapidly as $x\to\infty$ and $j\to\infty$;
 see Lemma~\ref{lem:Qeps} and Theorem~\ref{thm:per_result}.
\medskip

\noindent Let $A_{b_*,\eff}$ and $B_{b_*,\eff}$ denote the effective-medium parameters 
\begin{align}\label{Aeff}
A_{b_*,\eff}\ &=\ \frac1{8\pi^2}\partial_k^2E_{b_*}(k_*)\qquad \textrm{(inverse effective mass)}\\
B_{b_*,\eff}\ &=\ \int_\RR\ |u_{b_*}(x;k_*)|^2 \ \sum_{j\neq0}\  \frac{1}{(2\pi \lambda_j)^2}\ |q_j(x)|^2 \ dx . \label{Beff}
\end{align}
\medskip

\noindent Then, there exist constants $\epsilon_0>0$ and $\sigma_1, \sigma_2>0$, such that for all $0<\epsilon<\epsilon_0$ the following holds:\bigskip

\noindent $H_{Q+q_\epsilon}$ has a simple discrete eigenvalue, $E^\epsilon<E_\star$
(see the right panel in Figure~\ref{fig:SketchOfSpectrumQ});
\begin{equation}
E^\epsilon = E_\star + \epsilon^4 E_2 + \mathcal{O}(\epsilon^{4+\sigma_1});
\label{eq:Qnot0-eigenvalue-bound-intro} 
\end{equation}
with corresponding localized eigenfunction, $\psi^\epsilon$: 
 \begin{equation}
\sup_{x\in\RR} \left| \psi^\epsilon(x) - u_{b_*}(x;k_*)  g_0(\epsilon^2x)  \right|  \leq C \epsilon^{\sigma_2}.
\label{eq:Qnot0-eigenfunction-bound-intro} 
\end{equation}
 Here, $E_2 \ =\ -\dfrac{B_{b_*,\eff}^2}{4A_{b_*,\eff}}<0$ is the unique eigenvalue (simple) of the effective operator 
\begin{equation}\label{ABdelta}
H_{b_*,\eff}\ =\ -\frac{d}{d y}\ A_{b_*,\eff}\ \frac{d}{dy} \ - \ B_{b_*,\eff}\times \delta(y)\ ,
\end{equation}
where $\delta(y)$ denotes the Dirac delta mass at $y=0$, 
 and $ g_0(y)=\exp \left(-\frac{B_{b_*,\eff}}{2 A_{b_*,\eff}}|y|\right)$ 
 is its corresponding eigenfunction (unique up to a multiplicative constant).
 
\end{Theorem}

\begin{figure}[htp]
 \begin{center}
 \includegraphics[scale=1]{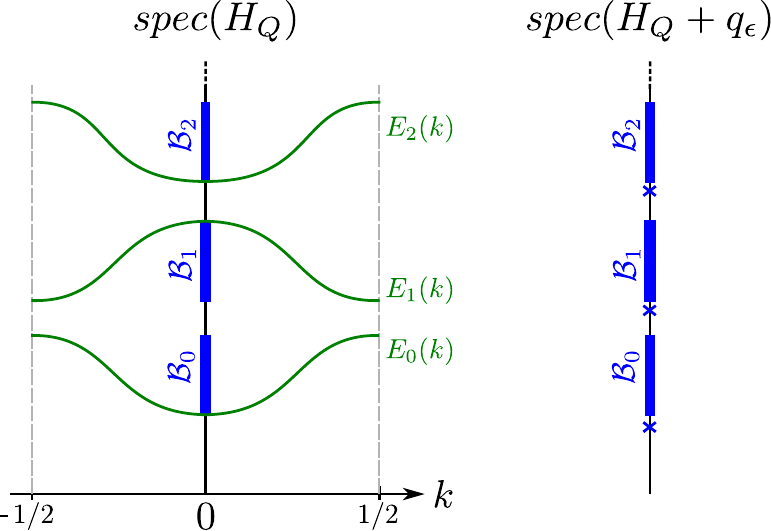}
 \end{center}
 \caption{\footnotesize{$\spec(H_Q)$ (left panel)
  and  $\spec(H_{Q+q_\epsilon})$ (right panel).Continuous spectra - thick vertical lines (blue) and discrete eigenvalues -cross marks (blue).  Dispersion curves (thin, green),  $k\in(-1/2,1/2]\mapsto E_b(k), b\ge0$.  }}
  \label{fig:SketchOfSpectrumQ}
\end{figure}

\begin{Remark}\label{Qeq0}
Theorem~\ref{thm:qeps-specific} applies to the special case: $Q\equiv0$. Indeed, the spectrum of $H_0=-\partial_x^2$ consists of a semi-infinite interval, $\spec(H_0)=[0,\infty)$, the union of intersecting bands with no positive length gaps. The only band-edge is located at $E_\star=E_0(0)=0$, where we have: $b_*=0$, $k_*=0$, $u_0(x;0)\equiv1$ for all $x\in\RR$ and $E_0(k)=4\pi^2 k^2$, and therefore
\[ A_{\eff}\ =\ 1,\qquad \
B_{\eff}\ =\ \int_\RR\ \sum_{j\neq0}\  \frac{1}{(2\pi \lambda_j)^2}\ |q_j(x)|^2 \ dx .\]
Thus we recover the result of~\cite{DVW:14a}, where it was shown that the bifurcation at the lower edge of the continuous spectrum of $H_0$ is governed by the Hamiltonian corresponding to a small {\em effective potential well} on the slow length-scale:
\[ H_{-\epsilon^2 \Lambda_{\rm eff}}=-\partial_x^2 - \epsilon^2 \Lambda_{\rm eff}(x) , \qquad \Lambda_{\rm eff}(x) =\sum_{j\neq0}\ \frac{1}{(2\pi \lambda_j)^2}\ |q_j(x)|^2 . \]
Consequently, classical results of, for example,~\cite{simon1976bound,DVW:14b}
apply and yield the effective Hamiltonian with a Dirac mass~\eqref{ABdelta} in the case $Q=0$.
\end{Remark}
\begin{Remark}\label{rem:below-spectrum}
Notice that~\eqref{Beff} yields $B_{b_*\eff}>0$ and thus bifurcation of eigenvalues may occur only for $A_{b_*,\eff}>0$, that is from the lower edge of spectral bands (see Lemma~\ref{lem:band-edge}, below). The same situation holds, by hypothesis, in the more general situation of Theorems~\ref{thm:Qzero} and~\ref{thm:per_result}.
\end{Remark}

\begin{Remark}[Examples of $q_\epsilon$, not of standard two-scale type]
As mentioned earlier, our results apply in more general situations than the two-scale perturbation presented above. The assumptions of Theorems~\ref{thm:Qzero} and~\ref{thm:per_result} imply that the leading-order component of the perturbation $q_\epsilon$ is supported at ever higher frequencies, asymptotically as $\epsilon\downarrow0$. 
The main difficulty in a specific situation is to check assumption (H2) in Theorem~\ref{thm:Qzero} (resp. (H2') in Theorem~\ref{thm:per_result}) the existence of 
 effective coupling coefficient, $B_\eff$.
 
   Lemma~\ref{lem:Qeps} in Appendix~\ref{sec:effective-potential} is dedicated to the computation of  $B_\eff$  in the case where $q_\epsilon=q(x,y)$ is a two-scale function as in Theorem~\ref{thm:qeps-specific}. The computations of Appendix~\ref{sec:effective-potential}  easily extend to perturbations of the form 
 \[ q_\epsilon(x) = \sum_{j\ne0}q_j(x)e^{2\pi i \lambda_j(\epsilon)x},\]
 with, for example,  the assumptions  $\lambda_{\pm1}(\epsilon)\approx\pm\frac{1}{\epsilon}$ and 
  $|\lambda_i(\epsilon)-\lambda_j(\epsilon)|\ge\kappa/\epsilon,\ \kappa>0,\ i\ne j$. This allows for dependence of $q_\epsilon$ on two-, three- {\it etc.} scales.

One further non-standard example to which our theorems apply  is obtained by taking 
\begin{align}
 q_\epsilon(x) &=q\left( \frac{x}{\epsilon^{_{2/3}}}\right),
 \nn\end{align} 
 where $|\widehat{q}(\xi)|\lesssim |\xi|^N$ for $|\xi|$ small ($N$ sufficiently large) and $q(x)$ 
decaying sufficiently rapidly as $| x|\to\infty$. In this case,
\[ B_\eff = |u_{b_*}(0;k_*)|^2\int_{-\infty}^\infty \left|\int_{-\infty}^x  q(y)\ dy \right|^2\ dx.\]
\end{Remark}

\subsection{Motivation, method of proof and relation to previous work}
 
In \cite{Borisov-Gadylshin:06} and in \cite{DVW:14a} the case $Q \equiv 0$ where $H_{q_\epsilon} = -\partial_x^2 + q_\epsilon(x)$, with $q_\epsilon(x)=q(x,x/\epsilon)$ is considered under different hypotheses. Our analysis in \cite{DVW:14a}  allows for almost periodic dependence in the fast-scale variable, {\em i.e.} potentials of the type displayed in~\eqref{eq:q-eps-def}. In this work we obtain details about eigenvalue asymptotics, and far more, by deriving asymptotics of the transmission coefficient, $k\mapsto t^{q_\epsilon}(k)$, that are valid uniformly for $ k\in\RR$ and in a complex neighborhood of zero energy.
This enables us to control the spectral measure of $H_{q_\epsilon}$, $|t^{q_\epsilon}(k)|^2\ dk$,
  leading to detailed dispersive energy transport information (time-decay estimates) in addition to results on eigenvalue-bifurcation.
 
The subtlety in this analysis stems from the behavior in a  neighborhood of $k=0$. Indeed, bounded away from $k=0$,  $t^{q_\epsilon}(k)\to1$ uniformly; see~\cite{DW:11}. The heart of the matter is a proof that
 \begin{equation}
 \frac{k}{t^{q_\epsilon}(k)} - \frac{k}{t^{\sigma}(k)}
 \end{equation}
 can be made to converge to zero as $\epsilon\to0$ uniformly on $\RR$ (and in a complex neighborhood of $k=0$) for the specific choice
 $\sigma(x)\equiv -\epsilon^2 \Lambda_{\eff}(x)$; see Remark~\ref{Qeq0}. Since $-\epsilon^2 \Lambda_{\eff}$ is a small potential well, classical results~\cite{simon1976bound} for the operator $H_{q_\epsilon} = -\partial_x^2 +\sigma(x)$ apply, and we conclude that $t^\sigma(k)$ and consequently $t^{q_\epsilon}(k)$ have a simple pole of order $\mathcal{O}(\epsilon^2)$ on the positive imaginary axis, from which the existence of a negative discrete eigenvalue, $E^\epsilon$, of order $\mathcal{O}(\epsilon^4)$ is an immediate consequence. More precisely, the asymptotic behavior of the eigenvalue corresponding to the small potential well, and therefore to the original oscillatory potential, is predicted by the Schr\"odinger operator with Dirac distribution potential with negative mass (see~\cite{DVW:14b}, consistently with \cite{simon1976bound,Borisov-Gadylshin:08}):
 \[ H^\epsilon_{\eff} =-\partial_x^2 -\left(\epsilon^2 \int_\RR \Lambda_{\eff}(x)\ dx\right)\times \delta(x).\]
 \medskip
 
 Since perturbations of the periodic Hamiltonian by weak potentials are also known to generate discrete eigenvalues, seeking an extension of the results in \cite{DVW:14a} to the case of a non-trivial and  periodic background was a natural motivation for the current article.

Indeed, it was proved in~\cite{Borisov-Gadylshin:08,DVW:14b}, for the Hamiltonian $H_{Q+\lambda V}=-\partial_x^2+Q+\lambda V$, where $Q$ is 1-periodic and $0<\lambda\ll1$, that if
\[
\partial_k^2 E_{b_*}(k_*) \times \int_\RR |u_{b_*}(x;k_*)|^2 V(x) dx < 0,
\]
then an eigenvalue of order $\lambda^2$ bifurcates from the edge of the $(b_*)^{\rm th}$ spectral band of the unperturbed operator $H_Q$. If $\partial_k^2 E_{b_*}(k_*) > 0$ and $\int_\RR |u_{b_*}(x;k_*)|^2 V(x) dx < 0$, this bifurcation is from the lower edge of the $(b_*)^{\rm th}$ band, while if   $\partial_k^2 E_{b_*}(k_*) < 0$ and $\int_\RR |u_{b_*}(x;k_*)|^2 V(x) dx > 0$ the bifurcation is from the upper edge of the $(b_*)^{\rm th}$ band.

Consistent with the case $Q=0$,  in this work we prove that the spectral properties of the Hamiltonian $H_{Q+q_\epsilon}$ localized near the $(b_*)^{\rm th}$ band edge are related to those of an effective Hamiltonian 
  \[
  H^{\epsilon}_{b_*,\eff}  = -\frac{d}{dx}A_{b_*,\eff}\frac{d}{dx} -\epsilon^2B_{b_*,\eff} \times \delta(x).
  \]
  Upon rescaling by $y=\epsilon^2x$ gives the operator $H_{b_*,\eff}$, displayed in~\eqref{ABdelta}.

{\em In contrast to the case of a \underline{multiplicatively} small perturbation, the eigenvalue bifurcations of $H_{Q+q_\epsilon}$ are shown in the present work  to occur \underline{only} from the lower band edge into the spectral gap below it. }
 The mathematical reason for this is that the bifurcation phenomena we study is an effect that occurs at second order in $\epsilon$. Making this effect explicit  requires iteration of our formulation of the eigenvalue problem,  leading to 
  terms which are quadratic in  $q_\epsilon$. As in the case $Q=0$, the dominant (resonant / non-oscillatory) contribution has the distinguished sign of a potential well; see Remark~\ref{rem:below-spectrum}. This result was also observed in~\cite[Corollary 2.1]{Borisov:07}.
\medskip

Non-oscillatory perturbations of Schr\"odinger operators with periodic background have been considered in a number of other works; 
see~\cite{Deift-Hempel:86,Gesztesy-Simon:93,Hoefer-Weinstein:11,BR:11}.  For the acoustic and Maxwell operators see~\cite{Figotin-Klein:97,Figotin-Klein:98}. Finally, Borisov and Gadyl{\cprime}shin~\cite{Borisov:07,Borisov-Gadylshin:06,Borisov-Gadylshin:07} obtained results which apply to our situation provided the perturbation $q_\epsilon$ is a two-scale potential and has compact support (neither hypothesis is required in our analysis).  In \cite{Borisov-Gadylshin:07}, one-dimensional divergence-form operators are treated.

In two space dimensions, the operator $-\Delta+\lambda V$, where $0<\lambda\ll1$ and $V$ is a localized potential well, has a discrete negative eigenvalue of order $\exp(-\alpha \lambda^{-1}),\ \alpha > 0$; see, for example,~\cite{simon1976bound,parzygnat2010sufficient}. In~\cite{Borisov11}, Borisov proves that eigenvalues of the operator $-\Delta+Q+\lambda V$, where $Q$ is periodic on $\RR^2$, bifurcate from the edges of the continuous spectrum at a distance $\exp(-\alpha \lambda^{-1})$. It is natural to 
\medskip

\noindent {\bf Conjecture:} In two space dimensions $-\Delta+Q+q_\epsilon$, where $Q$ is periodic on $\RR^2$ and $q_\epsilon$ is spatially localized and concentrated at ever higher frequencies as $\epsilon\downarrow0$ as in~\eqref{eq:q-eps-def}, spawns eigenvalues from its lower spectral band edges into open gaps at a distance $\sim\exp(-\alpha \epsilon^{-2}),\ \alpha > 0$.
\medskip

Finally, we remark on our method of analysis. We transform the eigenvalue problem  using the natural basis of eigenfunctions for the unperturbed operator and study the eigenvalue problem in (quasi-) momentum space. The momentum space formulation is natural in that one can very systematically pinpoint the key resonant (non-oscillatory) terms which control the $\epsilon\to0$ limit. Using this approach one sees clearly how to treat oscillatory perturbing potentials
 which are far more general than a prescribed multiscale type (two-scale, three-scale {\it etc.}).
We explicitly, via localization to energies near the bifurcation point and rescaling, re-express the Schr\"odinger eigenvalue problem with rapidly oscillatory coefficients as an approximately equivalent eigenvalue problem for an effective Schr\"odinger operator, $H_{\rm eff}$, with coefficients which do not oscillate rapidly. This effective Schr\"odinger Hamiltonian is determined by key constants 
 $A_{\rm eff}$ and $B_{\rm eff}$, which have natural physical meanings (inverse effective mass and effective potential well couple parameter, respectively). 
 
 The main tool for re-expressing the eigenvalue problem is careful integration by parts, which exploits oscillations of non-resonant (``irrelevant'') terms to show that they are small in norm. Resonant  (non-oscillatory) terms cannot be transformed to terms of high order in the small parameter and it is these terms that contribute to the effective operator, $H_{\rm eff}^\epsilon$. Thus our approach is somewhat akin to that taken in  Hamiltonian normal form theory and  the method of averaging. See also \cite{DVW:14a}.

\subsection{Outline of the paper}\label{sec:outline}

In Section~\ref{sec:background} we present background material concerning spectral properties of Schr\"odinger operators with periodic potentials defined on $\RR$. In Section~\ref{sec:main-results} we give precise technical statements of our main results: Theorem~\ref{thm:Qzero} and Theorem~\ref{thm:per_result}. Section~\ref{sec:technical-lemma} reviews general technical results on a class of band-limited Schr\"odinger operators, derived in~\cite{DVW:14b}, and applied in Sections~\ref{sec:Q=0} and~\ref{sec:Qnot0}. The strategy of the proof is explained in Section~\ref{sec:strategy}. Appendix~\ref{sec:bounds-I} gives detailed proofs of  bounds used in Section~\ref{sec:Qnot0}. Appendix~\ref{sec:bounds-FB} summarizes and proves bounds relating to the Floquet-Bloch states used in Section~\ref{sec:Qnot0}.  Finally,  Appendix~\ref{sec:effective-potential} has a detailed analysis and calculation of the effective potential for the particular case of the localized and oscillatory (almost periodic) potential $q_\epsilon(x)$,  defined in~\eqref{eq:q-eps-def}. 

\subsection{Definitions and notation}

We denote by $C$ a constant, which does not depend on the small parameter, $\epsilon$. It may depend on norms of $Q(x)$ and $q_\epsilon(x)$, which are assumed finite. $C(\zeta_1, \zeta_2,\dots)$ is a constant depending on the parameters
 $\zeta_1$, $\zeta_2$, $\dots$.
We write $A\lesssim B$ if $A\leq C \ B$, and $A\approx B$ if $A\lesssim B$ and $B\lesssim A$.

The methods of this paper employ spectral localization relative to the background operator $-\partial_x^2+Q(x)$, where $Q(x)$ is one-periodic. For the case, $Q\equiv0$, we use the classical Fourier transform and for $Q(x)$ a non-trivial periodic potential, we use the spectral decomposition of $L^2(\RR)$ in terms of {\it Floquet-Bloch} states; see Section~\ref{sec:introduction} and Section~\ref{sec:background} below. The notations and conventions we use are similar to those used in~\cite{Hoefer-Weinstein:11}.

\begin{enumerate}

\item For $f,g \in L^2(\RR)$, the Fourier transform and its inverse are given by 
\[
\mathcal{F}\{f\}(\xi)  \equiv  \widehat{f}(\xi) = \int_\RR e^{-2\pi i x\xi}f(x)dx, \qquad
\mathcal{F}^{-1}\{g\}(x)  \equiv  \check{g}(x) = \int_\RR e^{2\pi i x\xi}g(\xi)d\xi.
\]

\item $\mathcal{T}$ and $\mathcal{T}^{-1}$ denote the Gelfand-Bloch transform and its inverse, defined in~\eqref{def:T-x;k} and~\eqref{eq:bloch-inverse} respectively. We use the following notation for the Gelfand-Bloch transform of a function: $\mathcal{T}\{f\}(x;k) \equiv \wt{f}(x;s)$; see section~\ref{sec:background}. Note that we will also use the notation $\wt{f}(k)$ in Section~\ref{sec:Qnot0} to represent the projection of $\wt{f}(x;s)$ onto a particular Bloch function $p_{b}(x;k)$, for fixed $b$.
 
 \item $\chi$ and $\overline{\chi}$ are the characteristic functions defined for a parameter $\delta>0$ by
 \[
 \chi\left(|\xi| < \delta\right)  \equiv 
 \left\{ 
 \begin{array}{ll}
 1, \quad & |\xi| < \delta \\
 0, \quad & |\xi| \geq \delta
 \end{array}
 \right. , \qquad
 \overline{\chi}(|\xi| < \delta)  \equiv  1-\chi(|\xi| < \delta)\ \equiv\ 
 \left\{ 
 \begin{array}{ll}
 0, \quad & |\xi| < \delta \\
 1, \quad & |\xi| \geq \delta
 \end{array}
 \right.
 \]
 We also use the notation
 \[\chi_{_\delta}(\xi)=\chi(|\xi|<\delta) \ , \qquad \overline{\chi}_{\delta}(\xi)=\overline{\chi}(|\xi| < \delta).\]

\item $L^{2,s}(\RR)$ is the space of functions $F:\RR\to\RR$ such that $(1+|x|^2)^{s/2} F\in L^2(\RR_x)$, endowed with the norm
\begin{equation}
\big\Vert F\big\Vert_{L^{2,s}(\RR)} \equiv \big\Vert(1+|x|^2)^{s/2} F\big\Vert_{L^2(\RR_x)} < \infty.
\label{Lps-def}
\end{equation}

\item $W^{k,\infty}(\RR)$ is the space of functions $F:\RR\to\RR$ such that $\partial_x^j F\in L^\infty(\RR)$ for $0\leq j\leq k$, endowed with the norm
\[
\big\Vert F\big\Vert_{W^{k,\infty}(\RR)} \equiv \sum_{j=0}^k \big\Vert\partial_x^j F\big\Vert_{L^\infty(\RR)} < \infty.
\]
\end{enumerate}

\noindent{\bf Acknowledgements:}\  The authors thank the referees and editor for their careful reading of our article and for their suggestions.
 I.V. and M.I.W. acknowledge the partial support of U.S. National Science Foundation under U.S. NSF Grants DMS-10-08855, DMS-1412560, the Columbia Optics and Quantum Electronics IGERT NSF Grant DGE-1069420 and  NSF EMSW21- RTG: Numerical Mathematics for Scientific Computing. Part of this research was carried out while V.D. was the Chu Assistant Professor of Applied Mathematics at Columbia University.

\section{Mathematical background}\label{sec:background}

In this section we provide further mathematical background by summarizing basic results on the spectral theory of Schr\"odinger operators with periodic potentials defined on $\RR$. Specifically, in Section~\ref{subsec:FB-theory} we discuss
more detailed aspects of Floquet-Bloch theory,  the spectral theory of periodic Schr\"odinger operators, and in Section 
\ref{subsec:GB-transform} we introduce the Gelfand-Bloch transform and discuss its properties. For a detailed discussion, see for example,~\cite{eastham1973spectral,RS4,magnus1979hill}.

\subsection{Floquet-Bloch theory}\label{subsec:FB-theory}

For $Q$ continuous and one-periodic, consider the family of pseudo-periodic eigenvalue problems
\begin{equation}
(-\partial_x^2 + Q(x))u(x) = Eu(x) \ , \quad u(x+1)=e^{2\pi i k}u(x)\ ,
\label{k-evp}\end{equation}
parametrized by  $k\in(-1/2,1/2]$, the {\it Brillouin zone}.
Setting $u(x;k) = e^{2\pi ikx}p(x;k)$, this is equivalent to the family of  periodic  boundary value problems:
\begin{equation}\label{eq:e-value}
\left( -(\partial_x + 2\pi ik)^2 + Q(x) \right)p(x;k) = E(k)p(x;k), \quad p(x+1;k)=p(x;k)
\end{equation}
for each $k \in (-1/2,1/2]$. 

The solutions $p_b(x;k)$ may be chosen so that $\{p_b(x;k)\}_{b\ge0}$ is, for each fixed $k\in(-1/2,1/2]$, a complete orthonormal set in $L^2([0,1])$. It can be shown that the set of Floquet-Bloch states $\{u_b(x;k)\equiv e^{2\pi ikx}p_b(x;k), \ b\in\NN, \ -1/2< k\leq 1/2\}$ is complete in $L^2(\RR)$, {\it i.e.} for any $f\in L^2(\RR)$, 
\[
\left\Vert f(x) \ - \ \sum_{0\le b\le N} \int_{-1/2}^{1/2} \big\langle u_b(y,k),f \big\rangle_{L^2(\RR_y)} u_b(x;k)\ dk\ \right\Vert_{L^2(\RR_x)} \to0\ \  \textrm{as}\ \ N\uparrow\infty.\]
Recall that the spectrum of $H_Q=-\partial_x^2 + Q$ is the union of the spectral bands:
\[\spec(H_Q) = \bigcup_{b\ge 0} \mathcal{B}_b\ = \bigcup_{b\geq 0}\ \bigcup_{\substack{k \in (-1/2,1/2]}} E_b(k).\]

\begin{Definition} 
 We say there is a spectral gap between the $b^{th}$ and $(b+1)^{st}$ bands if
 \[ \sup_{|k|\le1/2} |E_b(k)|\ <\ \inf_{|k|\le1/2} |E_{b+1}(k)|\ .\]
\end{Definition}

Our analysis of eigenvalue-bifurcation from the band edge $E_\star \equiv E_{b_*}(k_*)$ into a spectral gap, requires  detailed properties of $E_b(k)$, {\it e.g.} regularity,  near its edges. These are summarized in the following two results; see, for example,~\cite{DVW:14b} and~\cite{eastham1973spectral}.
\begin{Lemma}\label{lem:band-edge} 
Assume $E_b(k_*)$ is an endpoint of a spectral band of $-\partial_x^2 + Q(x)$, which borders on a spectral gap.
%%%
%%%
Then $k_*\in\{0,1/2\}$ and the following results hold:
\begin{enumerate}
\item $E_b(k_*)$ is a simple eigenvalue of the eigenvalue problem~\eqref{k-evp}.
\item \subitem $b$ even: $E_b(0)$ corresponds to the {\em left (lowermost)} endpoint of the band,
\subitem \phantom{$b$ even:} $E_b(1/2)$ corresponds to the {\em right (uppermost)} endpoint.
\subitem $b$ odd: $E_b(0)$ corresponds to the {\em right (uppermost)} endpoint of the band,
\subitem \phantom{$b$ odd:} $E_b(1/2)$ corresponds to the {\em left (lowermost)}  endpoint.
\item $\partial_k E_b(k_*) = 0$;
\item \subitem $b$ even: $\partial_k^2 E_b(0) > 0$, $\partial_k^2 E_b(1/2) < 0$;
\subitem $b$ odd: $\partial_k^2 E_b(0) < 0$, $\partial_k^2 E_b(1/2) > 0$; 
\item $\partial_k^3E_b(k_*) = 0$.
\end{enumerate}
\end{Lemma}

\begin{Lemma}\label{lem:regularity-of-Eb}
 For $k$ real, consider the  Floquet-Bloch eigenpair $(E_b(k),u_b(x;k))$.
Assume $E_b(k_*),\ k_*\in\{0,1/2\}$, is a simple eigenvalue. Then, there are analytic mappings $k\mapsto E_b(k),\ k\mapsto u_b(x;k)$, with $u_b$ normalized, defined for $k$ in a sufficiently small complex neighborhood of $k_*$. 
\end{Lemma}
We conclude this section by recalling Weyl's asymptotics (see~\cite{Courant-Hilbert-I,eastham1973spectral})
\begin{Lemma}\label{lem:Weyl}
There exists $C_1,C_2>0$ such that for any $b\in\NN$ and $k\in (-1/2,1/2]$,
\begin{equation} \pi^2 b^2 -C_1 \leq E_b(k) \leq \pi^2 (b+1)^2 +C_2.\end{equation}
\end{Lemma}

\subsection{The Gelfand-Bloch transform} \label{subsec:GB-transform}

Let $f\in \mathcal S(\RR)$, the Schwartz space. We introduce the Gelfand-Bloch transform $\mathcal{T}\{f\}(x;k)$ or $\widetilde{f}(x;k) $, as follows
\begin{equation}\label{def:T-x;k}
\mathcal{T}\{f\}(x;k) =  \widetilde{f}(x;k) = \sum_{n\in\ZZ}e^{2\pi inx}\widehat{f}(k+n).
\end{equation}
  Note the following properties of $\mathcal{T}$. For any $x,k\in\RR$, one has
  \begin{align} 
\widetilde{f}(x+1;k) \ &=\ \widetilde{f}(x;k),\  \label{x-periodic}\\
\widetilde{f}(x;k+m) \ &=\ e^{-2\pi imx}\widetilde{f}(x;k),\ \ m\in\ZZ \label{k-pseudoperiodic}\\
\widetilde{f'}(x;k) \ &=\ (\partial_x+2\pi i k)\widetilde{f}(x;k). \label{T-derivative}
\end{align}
Furthermore, for any $m\in\ZZ$ we have $\mathcal{T}\{e^{2\pi i mz }f(z)\}(x;k)= e^{2\pi i mx}\mathcal{T}\{f\}(x;k)$. Therefore, 
for any sufficiently regular one-periodic function $V(z)$, 
\begin{equation}\label{prop:T-per}
\mathcal{T}\{V f\}(x;k) = V(x) \mathcal{T}\{f\}(x;k).
\end{equation}

Now, recall Poisson summation formula:
\[
\sum_{\nu \in \ZZ} f(x+\nu) = \sum_{\nu \in \ZZ} e^{2\pi i\nu x} \widehat{f}(\nu).
\]
One deduces the following identity for $f\in \mathcal S(\RR)$:
\begin{equation}\label{eq:PSF-id}
 \widetilde{f}(x;k) \ \equiv \ \sum_{n\in\ZZ} e^{2\pi i nx} \widehat{f}(k+n) \ = \  \sum_{n\in\ZZ} e^{-2\pi i k(n+x)} f(n+x) .
 \end{equation}
This yields in particular the following formula for the Bloch transform of a product of two functions.
\begin{Proposition}\label{prop:bloch-convolution}
The Bloch transform of a product of two functions can be written as a ``Bloch convolution'':
\begin{equation}\label{eq:bloch-convolution}
\wt{\left( fg \right)}(x;k) \ =\ \int_{-1/2}^{1/2} \wt{f}(x;k-l) \wt{g}(x;l) \ dl.
\end{equation}
Note that for  $|k-l|>1/2$, the integrand is evaluated using \eqref{k-pseudoperiodic}.
\end{Proposition}
\begin{proof}
We have
\begin{align*}
&\int_{-1/2}^{1/2} \wt{f}(x;k-l) \wt{g}(x;l) \ dl & & \\
 &\quad= \int_{-1/2}^{1/2} \sum_{n\in\ZZ} e^{2\pi i nx} \widehat{f}(k-l+n) \sum_{m\in\ZZ} e^{2\pi imx} \widehat{g}(l+m) \ dl 
 &\quad \text{by \eqref{def:T-x;k}} &\\
&\quad = \int_{-1/2}^{1/2} \sum_{n\in\ZZ} e^{-2\pi i(k-l)(n+x)} f(n+x) \sum_{m\in\ZZ} e^{-2\pi il(m+x)} g(m+x) \ dl
&\quad \text{by \eqref{eq:PSF-id}} &\\
&\quad = \sum_{n\in\ZZ} e^{-2\pi ik(n+x)}  f(n+x) \sum_{m\in\ZZ} g(m+x) \int_{-1/2}^{1/2} e^{-2\pi i l(m-n)} \ dl
&\quad \text{by Fubini}& \\
&\quad = \sum_{n\in\ZZ}  e^{-2\pi ik(n+x)} f(n+x) g(n+x) = \sum_{n\in\ZZ} e^{2\pi inx} \widehat{fg}(k+n) & \quad \text{by \eqref{eq:PSF-id}} &\\
&\quad = \wt{\left( fg \right)}(x;k)\ . & & 
\end{align*}
\end{proof}

Introduce the operator $\mathcal{T}^{-1}$:
\begin{equation}\label{eq:bloch-inverse}
f(x) = \mathcal{T}^{-1}\{\widetilde{f}\}(x)  =  \int_{-1/2}^{1/2} e^{2\pi ixk}\widetilde{f}(x;k)dk.
\end{equation}
One can check that $\mathcal{T}^{-1}$ is the inverse of $\mathcal{T}$,\  $\mathcal{T}^{-1}\mathcal{T} = Id$.

For any Floquet-Bloch mode,
\begin{equation}
u_b(x;k)\ =\ e^{2\pi ikx} p_b(x;k),
\label{p-def}
\end{equation}
we have, thanks to~\eqref{eq:PSF-id},
\begin{equation}
 \langle u_b(x,k),f(x) \rangle_{L^2(\RR_x)} \ = \ \langle p_b(x,k),\widetilde f(x;k) \rangle_{L^2_{\rm per}([0,1]_x)}
 \equiv \mathcal{T}_b \{f\}(k), \label{def:Tb}
\end{equation}
By completeness of the $\{p_b(x;k)\}_{b \geq 0}$, we deduce
\begin{equation}\label{eq:decomposition}
\widetilde{f}(x;k) \ = \ \sum_{b \geq 0} \mathcal{T}_b \{f\}(k) p_b(x;k) . 
\end{equation}
The above definitions and identities extend by density to $f\in L^2(\RR)$, and one has in particular for any $f\in L^2(\RR)$,
\begin{equation}
 f(x) \ = \ \sum_{b \geq 0} \int_{-1/2}^{1/2} \mathcal{T}_b \{f\}(k) u_b(x;k) dk\ =\ 
 \sum_{b \geq 0} \int_{-1/2}^{1/2}  \big\langle u_b(y,k),f(y) \big\rangle_{L^2(\RR_y)} u_b(x;k) dk.
\label{completeness}
\end{equation}

It will be natural to measure $H^s$ (Sobolev) regularity in terms of the decay properties of a function's Floquet-Bloch coefficients.
Thus we introduce the  $\mathcal{X}^s$ norm:

\begin{equation}
\big\Vert \widetilde{\phi} \big\Vert_{\mathcal{X}^s}^2 \equiv \int_{-1/2}^{1/2} \sum_{b \geq 0} \left(1+|b|^2\right)^s |\mathcal{T}_b\{\phi\}(k)|^2dk.
\label{Xs-norm}
\end{equation}

\begin{Proposition}\label{prop:norm_equivalence}
$H^s(\RR)$ is isomorphic to $\mathcal{X}^s$ for $s \geq 0$. Moreover, there exist positive constants $C_1$, $C_2$ such that for all $\phi \in H^s(\RR)$, we have
$
C_1 \big\Vert \phi \big\Vert_{H^s(\RR)} \leq \big\Vert \widetilde{\phi} \big\Vert_{\mathcal{X}^s} \leq C_2 \big\Vert \phi \big\Vert_{H^s(\RR)}.
$
\end{Proposition}
\begin{proof}
Since $E_0(0) = \inf \spec(-\partial_x^2 + Q)$, then $L_0 \equiv -\partial_x^2 + Q - E_0(0)$ is a non-negative operator which defines an equivalent norm on $H^s(\RR)$: $\big\Vert(Id + L_0)^{s/2} \phi \big\Vert_{L^2}\approx \big\Vert \phi \big\Vert_{H^s}$.
Using orthogonality, it follows that
\begin{align*}
\big\Vert \phi \big\Vert_{H^s}^2 & \approx  \big\Vert(Id + L_0)^{s/2} \phi \big\Vert_{L^2}^2\ 
 =\  \sum_{b \geq 0} \int_{-1/2}^{1/2} |\mathcal{T}_b\{\phi\}(k)|^2 |1 + E_b(k) - E_0(0)|^s dk \\
& \approx  \sum_{b \geq 0} \left( 1 + |b|^2 \right)^s \int_{-1/2}^{1/2} |\mathcal{T}_b\{\phi\}(k)|^2 dk\ \equiv \big\Vert \widetilde{\phi} \big\Vert_{\mathcal{X}^s}^2 \ ,
\end{align*}
where $\approx$ indicates norm equivalence.
The approximation in the last line follows from the Weyl asymptotics $E_b(k) \approx b^2$, stated in Lemma~\ref{lem:Weyl}. 
This completes the proof of Proposition~\ref{prop:norm_equivalence}.
\end{proof}

\section{Bifurcation of defect states into gaps; main results}\label{sec:main-results}

In this section we state our main results on the eigenvalue problem
\begin{equation}
(-\partial_x^2 + Q(x) + q_\epsilon(x) ) \psi^\epsilon(x) = E^\epsilon \psi^\epsilon(x),\  \psi\in L^2\ ,
\label{evp}\end{equation}
where $Q(x)$ is one-periodic and $q_\epsilon(x)$ a real-valued, localized at high frequencies and decreasing at infinity (precise hypotheses are specified below).
 \medskip 
  
Consider first  the case  where $Q(x)\equiv0$. The following result extends Corollary 3.7 of \cite{DVW:14a} to a larger class of localized and oscillatory potentials, $q_\epsilon$.
 \begin{Theorem}\label{thm:Qzero}
Assume that $q_\epsilon(x)$ is real-valued and satisfies the following, for $\epsilon$ sufficiently small:
\begin{itemize}
  \item[(H1a)] there exists $0<\mathcal{C}_0<\infty$, independent of $\epsilon$, such that
  \begin{equation}\label{as:Q=0-q-epsilon-bound}
  \big\Vert \widehat{q_\epsilon}\big\Vert_{L^1}+\big\Vert \widehat{q_\epsilon}\big\Vert_{L^\infty}+\big\Vert \widehat{q_\epsilon}'\big\Vert_{L^\infty}\leq\mathcal C_0 ,
  \end{equation}
  \item[(H1b)] there exists $N\geq 4$ and $0<\mathcal{C}_N<\infty$, independent of $\epsilon$,  such that 
  \begin{equation}\label{as:Q=0-q-epsilon-decay}  
  \sup_{\xi\in [\frac{-1}{2\epsilon},\frac{1}{2\epsilon}]} |\widehat{q_\epsilon}(\xi)| \leq \epsilon^N \mathcal C_N,
  \end{equation}
  \item[(H2)] there exists $0<B_{\eff},\mathcal C_{\eff},\sigma_{\eff}<\infty$, independent of $\epsilon$, such that
  \begin{equation}\label{as:Q=0-q-epsilon-effective}  
  \left| \epsilon^{-2}\int_{\RR}\frac{|\widehat{q_\epsilon}(\xi)|^2}{4\pi^2 \xi^2 + 1}\  d\xi \ - \  B_{\eff}\right| \ \leq \ \mathcal C_{\eff} \epsilon^{\sigma_{\eff}}.
  \end{equation}
  \end{itemize}

Then, there exist positive constants $\epsilon_0$, $C$, depending only on the above parameters, such that the following holds.
For all $0<\epsilon<\epsilon_0$, there exists an eigenpair $(E^{\epsilon},\psi^{\epsilon})$, for the eigenvalue problem
\begin{equation}\label{eq:Q=0-start}
\left( -\partial_x^2 + q_{\epsilon}(x) \right)\psi^{\epsilon}(x) = E^{\epsilon} \psi^\epsilon(x), \quad \psi^\epsilon\in L^2
\end{equation}
with $E^\epsilon$ strictly negative and of the order $\epsilon^4$. 
Moreover, $\psi^\epsilon\in L^\infty$ and we have
\begin{align}
\left| E^{\epsilon}+ \frac{\epsilon^4 B_{\eff}^2}{4}  \right| & \leq C \epsilon^{4+\sigma}, \label{eq:e-value-approx} \\
\sup_{x \in \RR} \left| \psi^{\epsilon}(x) - \exp \left(-\frac{\epsilon^2 B_{\eff}}{2}|x| \right) \right| & \leq C \epsilon^\sigma , \label{eq:e-fct-approx}
\end{align}
 where $\sigma = \min\{1,\sigma_\eff\}$. The eigenvalue $E^{\epsilon}$ is unique in the neighborhood defined by~\eqref{eq:e-value-approx} , and the corresponding eigenfunction, $\psi^{\epsilon}$, is unique up to a multiplicative constant. 
\end{Theorem}

We now turn to the more general case where 
$Q(x)$ may be  a non-trivial periodic background.
\begin{Theorem}\label{thm:per_result}
Assume $Q$ is real-valued, one-periodic and satisfies:
\begin{itemize}
\item[(HQ)] $Q\in W^{4,\infty}_{\rm per}$, so that one has (see Lemma~\ref{lem:estimates-uniform-in-b}) the estimate
\begin{equation}\label{as:Qn0-Q-pb}
\forall b\geq 0, \quad \forall k\in [-1/2,1/2], \quad \forall x\in [-1/2,1/2], \quad | \partial_x^\alpha p_b(x;k)| \ \leq \ (1+|b|^\alpha) C_\alpha,
\end{equation}
\end{itemize}
with $\alpha=0,\dots,6$.

Set $E_\star= E_{b_*}(k_*)$, the lower endpoint of the $(b_*)^{th}$ band, and assume that the band borders on a spectral gap. Thus $k_*=0$ or $1/2$ and $\partial_k^2 E_{b_*}(k_*) > 0$; see Lemma~\ref{lem:band-edge}. 
  
   Assume $q_\epsilon(x)$ is real-valued and localized at high frequencies in the sense that:
   \begin{itemize}
   \item[(H1'a)] there exists $0<\mathcal{C}_0<\infty$, independent of $\epsilon$, such that
   \begin{equation}\label{as:Qn0-q-epsilon-bound}
   \big\Vert \widehat{q_\epsilon}\big\Vert_{L^1}+\big\Vert \widehat{q_\epsilon}\big\Vert_{L^\infty}\leq\mathcal C_0 ;
   \end{equation}
   \item[(H1'b)] for all $\beta$ such that $0\leq \beta \leq 6 $, there exists $0<\mathcal{C}_{\beta}<\infty$, independent of $\epsilon$, such that
     \begin{equation}\label{as:Qn0-q-epsilon-decay} 
   \left( \int_{-1/(2\epsilon)}^{1/(2\epsilon)} |\widehat{q_\epsilon}(\xi)|^2 \ d\xi \right)^{1/2} \leq \mathcal{C}_{\beta}\ \epsilon^\beta .
   \end{equation}
   \end{itemize}
   Furthermore, assume $q_\epsilon(x)$ is such that
   \begin{itemize}
   \item[(H2')] there exists $0<B_{b_*,\eff},\mathcal{C}_{\eff},\sigma_\eff<\infty$, independent of $\epsilon$, such that
\begin{equation}\label{as:Qn0-q-epsilon-effective}
\left| \epsilon^{-2}\int_{\RR}|u_{b_*}(x;k_*)|^2q_\epsilon(x)Q_\epsilon(x) \ dx \ - \  B_{b_*,\eff} \right| \leq \mathcal{C}_{\eff}\ \epsilon^{\sigma_\eff},
\end{equation}
   where $Q_\epsilon(x)$ is real-valued and defined by $\widehat{Q_{\epsilon}}(\xi)=\frac{1}{1+4\pi^2 |\xi|^2}  \widehat{q_\epsilon}(\xi) $.
   \end{itemize}

Then there are positive constants $\epsilon_0, C $ and $\sigma$, depending only on the above parameters, such that the following assertions hold:
\begin{itemize}
\item[1.] For all $0<\epsilon<\epsilon_0$, there exists an eigenpair $(E^\epsilon, \psi^\epsilon(x))$ of the eigenvalue problem
\begin{align}\label{eq:Qn0-start}
\left( -\partial_x^2 + Q(x) + q_\epsilon(x) \right) \psi^\epsilon(x) = E^\epsilon \psi^\epsilon(x),\ \ \psi^\epsilon\in L^2(\RR).
\end{align}
with eigenvalue $E^\epsilon$ in the spectral gap, at a distance $\mathcal{O}(\epsilon^4)$ from the band edge, $E_\star$.

\item[2.] Specifically, for $\sigma = \min \{ 1/6,\sigma_\eff \}$ where $\sigma_\eff$ is defined in~\eqref{as:Qn0-q-epsilon-effective}: $E^\epsilon$ and $\psi^\epsilon(x)$ satisfy the following approximations:
\begin{align}
\left| E^\epsilon - \left( E_\star + \epsilon^4 E_2 \right) \right| & \leq C \epsilon^{4+\sigma} \label{eq:Qnot0-eigenvalue-bound} \\
\sup_{x\in\RR} \left| \psi^\epsilon(x) - u_{b_*}(x;k_*) \exp(\epsilon^2 \alpha_0 |x|) \right| & \leq C \epsilon^{\sigma}, \label{eq:Qnot0-eigenfunction-bound} 
\end{align}
where  $E_2<0$ and $\alpha_0<0$ are given by the expressions:
\[
E_2 = - \frac{B_{b_*,\eff}^2}{\frac{1}{2\pi^2} \partial_k^2 E_{b_*}(k_*)} < 0 \quad \text{and} \quad \alpha_0 = - \frac{ B_{b_*,\eff} }{\frac{1}{4\pi^2} \partial_k^2 E_{b_*}(k_*)} < 0.
\]

\item[3.] The eigenvalue, $E^\epsilon$, is unique in the neighborhood defined in~\eqref{eq:Qnot0-eigenvalue-bound}, and the corresponding eigenfunction, $\psi^\epsilon$, is unique up to a multiplicative constant.
\end{itemize}
\end{Theorem}
\begin{Remark}
Notice that (H2') is consistent with (H2) in the case $Q=0$. Indeed, the only band-edge of $\spec(-\partial_x^2)=[0,\infty)$ is located at $E_\star=E_0(0)=0$, where we have: $b_*=0$, $k_*=0$ and $u_0(x;0)\equiv1$ for all $x\in\RR$; and therefore
\[\int_{\RR}|u_{b_*}(x;k_*)|^2q_\epsilon(x)Q_\epsilon(x) \ dx=\int_{\RR}q_\epsilon(x)\overline{Q_\epsilon(x)} \ dx = \int_{\RR}\widehat{q_\epsilon(\xi)}\overline{\widehat{Q_\epsilon(\xi)}} \ d\xi=\int_{\RR}\frac{|\widehat{q_\epsilon}(\xi)|^2}{4\pi^2 \xi^2 + 1} d\xi.\]
\end{Remark}

\section{Key general technical results}\label{sec:technical-lemma}

In this section, we state results concerning the operator $\widehat{\mathcal{L}}_{0,\epsilon}[\theta]$, defined by:
\begin{equation}
\widehat{f}(\xi)\ \mapsto\ \widehat{\mathcal{L}}_{0,\epsilon}[\theta]\widehat{f}(\xi) \equiv \left(4\pi^2 A \xi^2 + {\theta}^2\right) \widehat{f}(\xi) - B\ \chi\left(|\xi|<\epsilon^{-\beta}\right) \int_{\RR}\chi\left(|\eta|<\epsilon^{-\beta}\right)\ \widehat{f}(\eta)\ d\eta .
\label{hatcalL0-def}
\end{equation}
Here, $A, \ B$ and $ \beta$  are fixed positive constants and $\theta^2>0$. The operator $\widehat{\mathcal{L}}_{0,\epsilon}[\theta] $ appears
in the bifurcation equations we derive via the Lyapunov-Schmidt reduction; see Section~\ref{sec:strategy}.

In $x$-space, we have that $\mathcal{L}_{0,\epsilon}[\theta]$ is a rank one perturbation of $-A\partial_y^2+\theta^2$:
\begin{equation}
\mathcal{L}_{0,\epsilon}[\theta] f \ \equiv \ (-A\partial_y^2+\theta^2)f(y) \ - \ \frac{2 B}{ \epsilon^{\beta}}\ 
\left\langle \frac{2}{ \epsilon^{\beta}}\sinc \left(\frac{2\pi }{ \epsilon^{\beta}}\ z \right),f(z)\right\rangle_{L^2(\RR_z)}\ 
\sinc \left(\frac{2\pi y}{ \epsilon^{\beta}} \right) ,
\label{calL0-def}
\end{equation}
where $\sinc(z)=\sin(z)/z$. $\mathcal{L}_{0,\epsilon}[\theta]$ is a band-limited regularization of the operator: 
\begin{equation}
 \left(\ H^{A,B}\ +\ \theta^2\right)\ f \ \equiv \ \left(-A\partial_y^2 \ - \ B \delta(y)\ +\ \theta^2\right)f\ ,
 \label{HAB-def}
 \end{equation}
appearing in the effective equations governing the leading order behavior of bifurcating eigenstates.

We now state two technical lemmas concerning the operator $\widehat{\mathcal{L}}_{0,\epsilon}[\theta]$. Lemma~\ref{lem:homogeneous} is proved in~\cite[Lemma~4.1]{DVW:14b}. Lemma~\ref{lem:technical}, which concerns solvability of the inhomogeneous equation~\eqref{eq:near_general} below, has the same conclusion as Lemma~\cite[Lemma~4.4]{DVW:14b} but is stated with one more condition,~\eqref{assumptionsR-theta}, on $R_\epsilon$.
The arguments presented in~\cite{DVW:14b} are easily adapted to yield Lemma~\ref{lem:technical}. 
\begin{Lemma} \label{lem:homogeneous}
Fix constants $A>0$, $B>0$ and $\beta>0$. Define, for $\theta^2>0$, the linear operator
\begin{equation}\label{eq:homogeneous}
\widehat{f}(\xi)\ \mapsto\ \widehat{\mathcal{L}}_{0,\epsilon}[\theta]\widehat{f}(\xi) \equiv \left(4\pi^2 A \xi^2 + {\theta}^2\right) \widehat{f}(\xi) - B\ \chi\left(|\xi|<\epsilon^{-\beta}\right) \int_{\RR} \chi\left(|\eta|<\epsilon^{-\beta}\right)\ \widehat{f}(\eta)\ d\eta .
\end{equation}
Note that $\widehat{\mathcal{L}}_{0,\epsilon}[\theta]:L^1_{\rm loc}(\RR)\to L^{1}_{\rm loc}(\RR)$. There exists a unique $\theta_{0,\epsilon}^2>0$ such that:
\begin{enumerate}
\item $\widehat{\mathcal{L}}_{0,\epsilon}[\theta_{0,\epsilon}]$
has a non-trivial kernel.
\item The ``eigenvalue'' $\theta_{0,\epsilon}^2$ is the unique positive solution of
\begin{equation} \label{eq:theta0} 
 1 \ - \ B\int_{\RR}\frac{\chi\left(|\xi|<\epsilon^{-\beta}\right)}{4\pi^2A\xi^2 + \theta_{0,\epsilon}^2}d\xi \ = \ 0 \ .
\end{equation}
\item The kernel of $\widehat{\mathcal{L}} _{0,\epsilon}[\theta_{0,\epsilon}]$ is given by:
\begin{equation} \label{eq:f0}
{\rm kernel}\left(\widehat{\mathcal{L}} _{0,\epsilon}[\theta_{0,\epsilon}]\right)\ =\ {\rm span}\left\{ \widehat{f}_{0,\epsilon}(\xi) \right\},\ \ {\rm where}\ \ 
 \widehat{f}_{0,\epsilon}(\xi) \equiv\ \frac{\chi\left(|\xi|<\epsilon^{-\beta}\right)}{4\pi^2A\xi^2 + \theta_{0,\epsilon}^2} \ .
\end{equation}
\item
 $\theta_{0,\epsilon}$ can be approximated as follows:
\begin{equation} \label{eq:esttheta0}
\left| \frac{1}{\theta_{0,\epsilon}} \ - \ \frac{2\sqrt{A}}{B} \right| \ \le\frac{1}{\pi^2 \sqrt{A}}\ \epsilon^{\beta} .
\end{equation}
\item One has
\begin{equation} \label{eq:asymptotic-f0}
\sup_{x\in\RR} \left| \ \mathcal{F}^{-1}\left\{ \widehat{f}_{0,\epsilon} \right\}(x) \ - \ \frac{2}{B} \exp{\left(-\frac{B}{2A} |x|\right)}\ \right| \ \leq C(A,B) \epsilon^{\beta}.
\end{equation}
\end{enumerate}
\end{Lemma} 

The following result concerns solutions to perturbations of $\widehat{\mathcal{L}}_{0,\epsilon}$. 
Let $\mathcal{Z}_1$ and $\mathcal{Z}_2$ denote Banach spaces with 
 $\mathcal{Z}_1, \mathcal{Z}_2\subset L^1_{\rm loc}$. Assume that 
 for any $(f,g)\in \mathcal{Z}_1\times \mathcal{Z}_2$, 
\begin{equation}\label{norms}
\left| \left\langle f, g \right\rangle_{L^2} \right| \lesssim \big\Vert f \big\Vert_{\mathcal{Z}_2} \big\Vert g \big\Vert_{\mathcal{Z}_1} , \quad \quad \big\Vert f g \big\Vert_{\mathcal{Z}_2} \lesssim \big\Vert f \big\Vert_{\mathcal{Z}_2} \big\Vert g \big\Vert_{L^\infty}, \quad \text{ and } \quad \big\Vert (1+\xi^2)^{-1}f \big\Vert_{\mathcal{Z}_2} \lesssim \big\Vert f \big\Vert_{\mathcal{Z}_1}.
\end{equation}
We seek a solution of the equation:
\begin{equation} \label{eq:near_general}
\widehat{\mathcal{L}} _{0,\epsilon}[\theta]\widehat f\ =\  R_\epsilon[\theta]\widehat f\ ,
\end{equation}
where $\widehat{\mathcal{L}}_{0,\epsilon}(\theta)$ is the operator defined in~\eqref{eq:homogeneous} and the mapping 
 $\widehat f\mapsto  R_\epsilon[\theta]\widehat f$ is linear and satisfies the following properties:

\noindent{\bf Assumptions on $R_\epsilon$:} There exist constants $\alpha,\beta,t_-,t_+,C_{R_\epsilon}>0$ such that for $\epsilon$ sufficiently small
\begin{itemize}
\item for any $\widehat f\in \mathcal{Z}_2$, and $0<t_-<\theta^2<t_+<\infty$,
\begin{equation}
 \chi\left(|\xi|<\epsilon^{-\beta}\right)  \big(R_\epsilon[\theta]\widehat f\,\big)(\xi) \ = \  \big(R_\epsilon[\theta]\widehat f\,\big)(\xi)\ , \quad \text{ and } \quad 
 \big\Vert R_\epsilon[\theta]\widehat f\, \big\Vert_{\mathcal{Z}_1} \ \leq\ C_{R_\epsilon} \epsilon^{\alpha} \big\Vert\widehat f\, \big\Vert_{\mathcal{Z}_2}\ .
\label{assumptionsR}
\end{equation}
\item for any $\widehat f\in \mathcal{Z}_2$, and $0<t_-<\theta_1^2,\theta_2^2<t_+<\infty$,
\begin{equation}
 \big\Vert  R_\epsilon[\theta_1]\widehat f -R_\epsilon[\theta_2]\widehat f\, \big\Vert_{\mathcal{Z}_1} \ \leq\ C_{R_\epsilon} \epsilon^{\alpha} |\theta_1^2-\theta_2^2|\big\Vert\widehat f\, \big\Vert_{\mathcal{Z}_2}\ .
\label{assumptionsR-theta}
\end{equation}
\end{itemize}

In the above setting we have the following 

\begin{Lemma} \label{lem:technical}
Let $(\theta_{0,\epsilon}^2,\widehat{f}_{0,\epsilon}(\xi))$ be the solution of $\widehat{\mathcal{L}} _{0,\epsilon}[\theta]\widehat{f}=0$, as defined in Lemma~\ref{lem:homogeneous}, where $A,B$ and $\beta>0$ are fixed. Let $R_\epsilon[\theta]: \mathcal{Z}_2\to\mathcal{Z}_1$ be a linear mapping satisfying the assumptions displayed in~\eqref{assumptionsR}-\eqref{assumptionsR-theta}, where $\mathcal{Z}_1, \mathcal{Z}_2$ satisfy~\eqref{norms} and $\widehat{f}_{0,\epsilon}\in \mathcal{Z}_1\cap \mathcal{Z}_2$. Then there exists $\epsilon_0>0$ such that for any $0<\epsilon<\epsilon_0$, the following hold:
\begin{enumerate}
\item There exists a unique solution $\left(\theta_\epsilon,\widehat{f}_\epsilon(\xi) \right)\in\RR^+\times \mathcal{Z}_2$ of the equation 
~\eqref{eq:near_general}, such that 
 \[\big\Vert \widehat{f}_\epsilon - \widehat{f}_{0,\epsilon} \big\Vert_{\mathcal{Z}_2} \leq C \epsilon^{\alpha} \ , \quad \text{ and }\quad
 \int_{-\infty}^\infty \widehat{f}_\epsilon(\xi)-\widehat{f}_{0,\epsilon}(\xi)\ d\xi \ = \ 0, \]
with $C=C(A,B,C_{R_\epsilon},\beta)$, independent of $\epsilon$. 
\item Moreover, one has
\[
\widehat{f}_\epsilon(\xi) \ = \ \chi\left(|\xi|<\epsilon^{-\beta}\right) \widehat{f}_\epsilon(\xi)\ , \quad \text{ and }\quad \left| \theta_\epsilon^2 - \theta_{0,\epsilon}^2 \right| \leq C\epsilon^{\alpha} \ .
\]
\end{enumerate}
\end{Lemma} 
%%%

\begin{Remark}
To prove Theorem~\ref{thm:Qzero}, respectively Theorem~\ref{thm:per_result}, we shall apply Lemma~\ref{lem:technical} with $\left(\mathcal{Z}_1,\mathcal{Z}_2\right) = \left(L^{\infty},L^1\right)$, respectively $\left(\mathcal{Z}_1,\mathcal{Z}_2\right) = \left(L^{2,-1},L^{2,1} \right)$, where $L^{2,s}$ is the space of locally integrable functions such that
\[ \big\Vert F \big\Vert_{L^{2,s}} \ \equiv \ \big\Vert (1+|\xi|^2)^{s/2} F \big\Vert_{L^2(\RR_\xi)} \ < \ \infty.\]
It is straightforward to check that such pair of spaces satisfy~\eqref{norms}, and $\widehat{f}_{0,\epsilon}\in \mathcal{Z}_1\cap \mathcal{Z}_2$.
\end{Remark}

\section{Strategy}\label{sec:strategy}

The strategy we take in Sections~\ref{sec:Q=0} and~\ref{sec:Qnot0} is to reduce the eigenvalue problem 
\begin{equation}\label{eq:orig-EVP}
H_{Q+q_\epsilon} \psi = E \psi
\end{equation}
to a homogenized and band-limited Schr\"odinger equation of the form~\eqref{eq:near_general}. We assume $(E,\psi)$ solves the eigenvalue problem~\eqref{eq:orig-EVP} and show by a long, formal, and reversible calculation that the rescaled near energy components of $\psi(x)$, $\widehat{\Phi}(\xi)$, and rescaled energies of $E$, $\theta^2$, satisfy an equation of the form~\eqref{eq:near_general}, namely
\begin{equation}\label{eq:near_general_2}
\mathcal{L}_{0,\epsilon}[\theta]\widehat{\Phi} = R_\epsilon[\theta]\widehat{\Phi}.
\end{equation}
We then apply Lemma~\ref{lem:technical} to construct solutions $(\theta^2_\epsilon, \widehat{\Phi}_\epsilon)$ to~\eqref{eq:near_general_2}. 

The reduction of~\eqref{eq:orig-EVP} to~\eqref{eq:near_general_2} for the case $Q \equiv 0$ is achieved in Proposition~\ref{summary-Q0}, and that for $Q \not \equiv 0$  is achieved in Proposition~\ref{prop:Qn0-leading-order}. In Sections~\ref{subsec:Q=0-conclusion} and~\ref{subsec:Qnot0-conclusion} the solution of the original eigenvalue problem,~\eqref{eq:orig-EVP}, is reconstructed from the solutions to~\eqref{eq:near_general_2}. 

In particular, we find  that the eigenvalue problems $H_{Q+q_\epsilon} \psi = E \psi$ with $Q\equiv 0$ and $Q \not \equiv 0$ have a bifurcating branch of eigenstates such that, for $\sigma > 0$,
\begin{align*}
\epsilon & \mapsto \epsilon^4 E_2 + \mathcal{O}(\epsilon^{4+\sigma}), \quad \text{for } Q \equiv 0, \text{ and} \\
\epsilon & \mapsto E_{b_*} + \epsilon^4 E_2 + \mathcal{O}(\epsilon^{4+\sigma}), \quad \text{for } Q \not \equiv 0,
\end{align*}
where $E_2 < 0$ and $E_{b_*}$ is the lower edge of the $(b_*)^{\rm th}$ spectral band of the eigenvalue problem $H_Q u = E u$. 

\section{Proof of Theorem~\ref{thm:Qzero}; Edge bifurcations for $-\partial_x^2 + q_\epsilon(x)$}\label{sec:Q=0}

In this section we study the bifurcation of solutions to the eigenvalue problem 
\begin{equation}\label{eq:Q=0}
\left( -\partial_x^2 + q_{\epsilon}(x) \right)\psi(x) = E \psi(x), \ \qquad \psi \in L^2(\RR),
\end{equation}
into the interval $(-\infty,0)$, the semi-infinite spectral gap of $H_0 \equiv -\partial_x^2$, for $q_\epsilon$ localized at high frequencies and decaying as $|x|\to\infty$. 

We prove Theorem~\ref{thm:Qzero}, which may be seen as a particular case of our main result, Theorem~\ref{thm:per_result}. In this case $Q\equiv0$ and thus the Floquet-Bloch eigenfunctions are explicit exponentials, making calculations more straightforward and error bounds on the approximations sharper. Section~\ref{sec:Qnot0} will present a more general argument for the $Q\neq0$ case. 

We will begin by transforming equation~\eqref{eq:Q=0} into frequency space in Section~\ref{subsec:Q=0-nearandfar}, which we will divide into a coupled system of equations, one pertaining to energies near the expected bifurcation point, and the other to energies far from the bifurcating points. Then, in Sections~\ref{subsec:Q=0-far} and~\ref{subsec:Q=0-near} we will study each part of the system in detail to finally complete the proof of Theorem~\ref{thm:Qzero} in Section~\ref{subsec:Q=0-conclusion}. 

\subsection{Near and far energy components}\label{subsec:Q=0-nearandfar}

Anticipating that the bifurcating eigenvalue, $E$, will be real, negative and of size $\approx \epsilon^4$ (\cite{DVW:14a})
we set
\begin{equation}\label{eq:def-theta}
E \equiv -\epsilon^4 \theta^2, \quad 0 < t_- \leq \theta^2 \leq t_+ < \infty,
\end{equation}
where $t_-$ and $t_+$ are independent of $\epsilon$. We expect, and eventually prove, $\theta \to \theta_{\eff}$ as $\epsilon \rightarrow 0$, with $0<\theta_{\eff}<\infty$.

Taking the Fourier transform of~\eqref{eq:Q=0} yields
\begin{equation}\label{eq:FT-of-Q=0}
(4\pi^2 \xi^2 + \epsilon^4 \theta^2)\widehat{\psi}(\xi) + \int_{\zeta} \widehat{q_\epsilon}(\xi - \zeta) \widehat{\psi}(\zeta) d\zeta = 0.
\end{equation}
We wish to study~\eqref{eq:FT-of-Q=0} as a coupled system of equations via the 
\begin{align*}
\text{\em near energy component:} & \ \{ \widehat{\psi}(\xi) \: : \: |\xi| < \epsilon^r \} \text{ and} \\
\text{\em far energy component:} & \ \{ \widehat{\psi}(\xi) \: : \: |\xi| \geq \epsilon^r \} \text{ of } \widehat{\psi}. 
\end{align*}

Let $r$ be a positive parameter, $r>0$, to be specified.  We denote $\chi$ the cut-off function: 
\begin{align*}
\chi(\xi) \equiv 1, \ |\xi| < 1 \text{ and } \chi(\xi) \equiv 0, \ |\xi| \geq 1. 
\end{align*} 
We also set 
\begin{align*}
\overline{\chi}(\xi) \equiv 1 - \chi(\xi) \text{ and } \chi_{_{\epsilon^r}}(\xi) \equiv \chi(\epsilon^{-r} \xi). 
\end{align*}

Introduce notation for near and far energy components of $\widehat{\psi}$:
\begin{equation}
\widehat{\psi}_{\nr} (\xi) \equiv \chi_{_{\epsilon^r}}(\xi) \widehat{\psi}(\xi) \ \text{ and } \ \widehat{\psi}_{\fr} (\xi) \equiv \overline{\chi}_{\epsilon^r}(\xi) \widehat{\psi}(\xi).
\end{equation}
The eigenvalue equation~\eqref{eq:FT-of-Q=0} is equivalent to the following coupled system of equations for the near and far energy components:
\begin{align}
\left( 4\pi^2 \xi^2 + \epsilon^4 \theta^2 \right) \widehat{\psi}_\nr (\xi) + \chi_{_{\epsilon^r}}(\xi) \int_{\zeta} \widehat{q_\epsilon}(\xi - \zeta) \left( \widehat{\psi}_\nr (\zeta) + \widehat{\psi}_\fr (\zeta) \right) d\zeta & = 0, \label{eq:Q=0-near}\\
\left( 4\pi^2 \xi^2 + \epsilon^4 \theta^2 \right) \widehat{\psi}_\fr (\xi) + \overline{\chi}_{\epsilon^r}(\xi) \int_{\zeta} \widehat{q_\epsilon}(\xi - \zeta) \left( \widehat{\psi}_\nr (\zeta) + \widehat{\psi}_\fr (\zeta) \right) d\zeta & = 0. \label{eq:Q=0-far}
\end{align}

The analysis of the far energy equation~\eqref{eq:Q=0-far} and near energy equation~\eqref{eq:Q=0-near} relies heavily on some smallness induced by the assumption that $\widehat{q_\epsilon}$ is localized at high frequencies, and that we  encapsulate in the following Lemma.
\begin{Lemma}\label{lem:smallness}
For every $\epsilon>0$, let $f_\epsilon,g_\epsilon\in L^1(\RR) \cap L^\infty(|\xi| \geq  \frac{1}{4\epsilon})$.
Then, for $\widehat{q_\epsilon}\in L^1(\RR)$, one has
\begin{align} \sup_{|\xi|\leq \frac{1}{4\epsilon}} \left| \int_{\RR} g_\epsilon(\zeta) \widehat{q_\epsilon}(\xi - \zeta) d\zeta \right| & \leq \sup_{|\xi| \leq  \frac{1}{2\epsilon}} |\widehat{q}_\epsilon(\xi)|\big\Vert g_\epsilon \big\Vert_{L^1(\RR)}  + \big\Vert \widehat{q}_\epsilon \big\Vert_{L^1(\RR)} \ \sup_{|\zeta|\geq  \frac{1}{4\epsilon}} |g_\epsilon (\zeta)|, \label{small-1} \\ 
\left| \iint_{\RR^2} f_\epsilon(\xi)g_\epsilon(\zeta) \widehat{q_\epsilon}(\xi - \zeta) d\zeta d\xi \right| & \leq  
\sup_{|\xi| \leq  \frac{1}{2\epsilon}} |\widehat{q}_\epsilon (\xi) |\big\Vert f_\epsilon \big\Vert_{L^1(\RR)} \big\Vert g_\epsilon \big\Vert_{L^1(\RR)}  \nn \\
& \quad + \big\Vert \widehat{q}_\epsilon \big\Vert_{L^1(\RR)} \left( \sup_{|\zeta| \geq  \frac{1}{4\epsilon}} |f_\epsilon(\zeta)| \big\Vert g_\epsilon \big\Vert_{L^1(\RR)} + \big\Vert f_\epsilon \big\Vert_{L^1(\RR)} \sup_{|\zeta| \geq  \frac{1}{4\epsilon}} |g_\epsilon(\zeta)| \right). \label{small-2}
\end{align}
\end{Lemma}
\begin{proof} We start with the proof of estimate~\eqref{small-1}. Assume $|\xi|\leq  \frac{1}{4\epsilon}$. We decompose the integration domain into $|\zeta|\leq \frac{1}{4\epsilon}$ and $|\zeta|\geq \frac{1}{4\epsilon}$. For $|\zeta| \leq  \frac{1}{4\epsilon}$ and $|\xi|\leq  \frac{1}{4\epsilon}$, we have $|\xi-\zeta|\leq |\xi|+|\zeta|\leq \frac{1}{2\epsilon}$, and therefore
\begin{equation}\label{eq:a}
 \sup_{|\xi|\leq \frac{1}{4\epsilon}} \int_{|\zeta|\leq \frac1{4\epsilon}} |g_\epsilon(\zeta) \widehat{q_\epsilon}(\xi - \zeta)|d\zeta\leq  \sup_{|\xi - \zeta| \leq  \frac{1}{2\epsilon}} |\widehat{q}_\epsilon(\xi-\zeta)| \big\Vert g_\epsilon \big\Vert_{L^1}.
\end{equation}
The integral over $|\zeta| \geq  \frac{1}{4\epsilon}$ is estimated as follows,
\begin{equation}\label{eq:b}
\int_{|\zeta|\geq \frac1{4\epsilon}} |g_\epsilon(\zeta) \widehat{q_\epsilon}(\xi - \zeta)|d\zeta\leq  \big\Vert \widehat{q}_\epsilon \big\Vert_{L^1} \sup_{|\zeta|\geq  \frac{1}{4\epsilon}} |g_\epsilon (\zeta)|.
\end{equation}
The bound~\eqref{small-1} now follows from~\eqref{eq:a} and~\eqref{eq:b}.

To prove estimate~\eqref{small-2}, we decompose the integration domain into
\[ D_1 \equiv \Big\{(\zeta,\xi), \ |\zeta-\xi|\leq \frac{1}{2\epsilon}\Big\}, \quad D_2 \equiv \Big\{(\zeta,\xi), \ |\zeta-\xi|\geq \frac{1}{2\epsilon}\Big\}.\]
The contribution from $D_1$ is controlled by the bound:
\begin{equation}\label{eq:c}
\left| \iint_{D_1} f_\epsilon(\xi)g_\epsilon(\zeta)\widehat{q_\epsilon}(\xi - \zeta)d\zeta d\xi \right| \leq \sup_{|\xi - \zeta| \leq  \frac{1}{2\epsilon}} |\widehat{q}_\epsilon(\xi-\zeta)| \big\Vert f_\epsilon \big\Vert_{L^1}\big\Vert g_\epsilon \big\Vert_{L^1}.
\end{equation}
For $(\zeta,\xi)\in D_2$, we have that either $|\zeta|\geq  \frac{1}{4\epsilon}$ or $ |\xi|\geq  \frac{1}{4\epsilon}$. Assume $|\xi| \geq  \frac{1}{4\epsilon}$; the case $|\zeta| \geq  \frac{1}{4\epsilon}$ is treated symmetrically. One has
\begin{equation}\label{eq:d}
\left| \int_\RR d\zeta g_\epsilon(\zeta)\int_{|\xi|\geq \frac{1}{4\epsilon}} d\xi f_\epsilon(\xi)\widehat{q_\epsilon}(\xi - \zeta) \right| \lesssim \big\Vert g_\epsilon \big\Vert_{L^1} \sup_{|\xi| \geq  \frac{1}{4\epsilon}} |f_\epsilon (\xi)| \big\Vert \widehat{q_\epsilon} \big\Vert_{L^1}.
\end{equation}
It follows that
\[ \left| \iint_{D_2} f_\epsilon(\xi)g_\epsilon(\zeta)\widehat{q_\epsilon}(\xi - \zeta)d\zeta d\xi \right| \leq \big\Vert \widehat{q_\epsilon} \big\Vert_{L^1} \left( \big\Vert g_\epsilon \big\Vert_{L^1}\sup_{|\xi| \geq  \frac{1}{4\epsilon}} |f_\epsilon (\xi)| + \big\Vert f_\epsilon \big\Vert_{L^1} \sup_{|\zeta| \geq  \frac{1}{4\epsilon}} |f_\epsilon (\zeta)| \right).\]
The bound~\eqref{small-2}, and therefore Lemma~\ref{lem:smallness}, now follows from~\eqref{eq:c} and~\eqref{eq:d}. 
\end{proof}

\subsection{Analysis of the far energy component}\label{subsec:Q=0-far}

We view~\eqref{eq:Q=0-far} as an equation for $\widehat{\psi}_\fr$ depending on ``parameters'' $(\widehat{\psi}_\nr, \theta^2; \epsilon)$. The following proposition studies the mapping $(\widehat{\psi}_\nr, \theta^2; \epsilon) \mapsto \widehat{\psi}_\fr[\widehat{\psi}_\nr, \theta^2;\epsilon]$. 

\begin{Proposition}\label{prop:far-in-terms-near} 
Fix $ r \in (0,2)$ and $\theta\in\RR$.  Let $\widehat{\psi}_\nr \in L^1$, and $q_\epsilon$ satisfying~\eqref{as:Q=0-q-epsilon-bound} and~\eqref{as:Q=0-q-epsilon-decay} of Theorem~\ref{thm:Qzero} with $N>2r$.
There exists $\epsilon_0$ such that for $0 < \epsilon < \epsilon_0$ the following holds.

There is a unique solution $\widehat{\psi}_\fr[\widehat{\psi}_\nr,\theta^2;\epsilon]$ of the far energy equation~\eqref{eq:Q=0-far}. Moreover, for any $(\theta^2,\epsilon)\in\RR\times (0,\epsilon_0)$, the mapping 
\begin{align*}
\widehat{\psi}_\nr \mapsto \widehat{\psi}_\fr[\widehat{\psi}_\nr,\theta^2;\epsilon]
\end{align*}
is a linear mapping from $L^1(\RR) $ to $L^1(\RR)$ and satisfies the bound
\begin{align}\label{eq:far-bound}
\big\Vert \widehat{\psi}_\fr \big\Vert_{L^1} &\leq C(\mathcal{C}_0,\mathcal{C}_N)  \big(\epsilon^{N-2r}+\epsilon^{2-r}\big)  \big\Vert \widehat{\psi}_\nr \big\Vert_{L^1}.
\end{align}
\end{Proposition}
\begin{proof}
We seek to solve~\eqref{eq:Q=0-far} for $\widehat{\psi}_\fr$ as a functional of $\widehat{\psi}_\nr$.
First note that since $\theta\in\RR$, one has for $|\xi| \geq \epsilon^r$, $|4\pi^2 \xi^2 + \epsilon^4 \theta^2| \geq 4\pi^2 \epsilon^{2r}$ is bounded away from zero for any fixed $\epsilon>0$. Dividing~\eqref{eq:Q=0-far} by $ 4\pi^2 \xi^2 + \epsilon^4 \theta^2$ and rearranging terms we obtain
\[
\widehat{\psi}_\fr(\xi)  = - \frac{\overline{\chi}_{\epsilon^r}(\xi)}{4\pi^2 \xi^2 + \epsilon^4 \theta^2} \int_{\zeta} \widehat{q_\epsilon} (\xi - \zeta) \left( \widehat{\psi}_\nr(\zeta) + \widehat{\psi}_\fr(\zeta) \right) d\zeta .\]

Iterating the equation, we have
\begin{multline*}
\widehat{\psi}_\fr(\xi)  = - \frac{\overline{\chi}_{\epsilon^r}(\xi)}{4\pi^2 \xi^2 + \epsilon^4 \theta^2} \int_{\RR}d\zeta  \widehat{q_\epsilon} (\xi - \zeta) \biggl( \widehat{\psi}_\nr(\zeta) \\
- \frac{\overline{\chi}_{\epsilon^r}(\zeta)}{4\pi^2 \zeta^2 + \epsilon^4 \theta^2}  \int_{\RR} d\eta\widehat{q_\epsilon} (\zeta - \eta) \left( \widehat{\psi}_\nr(\eta) + \widehat{\psi}_\fr(\eta) \right) \biggr) ,
\end{multline*}
which we can write as
\begin{equation}\label{eq-far-T}
\left(I-\widehat{\mathcal{T}}_\epsilon\circ \widehat{\mathcal{T}}_\epsilon\right) \widehat{\psi}_\fr=- \widehat{\mathcal{T}}_\epsilon\widehat{\psi}_\nr+\widehat{\mathcal{T}}_\epsilon\circ \widehat{\mathcal{T}}_\epsilon \widehat{\psi}_\nr.
 \end{equation}
Here $\widehat{\mathcal{T}}_\epsilon$ is the integral operator defined by
\begin{align*}
\left( \widehat{\mathcal{T}}_\epsilon\widehat{f} \right)(\xi) \equiv \int_{\zeta} \mathcal{K}_{\epsilon}(\xi,\zeta) \widehat{f}(\zeta) d\zeta \quad \text{ and } \quad \mathcal{K}_{\epsilon}(\xi,\zeta) \equiv \frac{\overline{\chi}_{\epsilon^r}(\xi)}{4\pi^2\xi^2 + \epsilon^4\theta^2} \widehat{q_\epsilon}(\xi - \zeta).
\end{align*}

We will show that the operator $(I-\widehat{\mathcal{T}}_\epsilon\circ \widehat{\mathcal{T}}_\epsilon)$ is invertible as an operator from 
 $L^1$ to itself, using that $\big\Vert\widehat{\mathcal{T}}_\epsilon\circ \widehat{\mathcal{T}}_\epsilon\big\Vert_{L^1\to L^1}$ is small when $\epsilon$ is small. Indeed, one has for $\widehat{h} \in L^1$,
\begin{align*}
\big\Vert(\widehat{\mathcal{T}}_\epsilon\circ \widehat{\mathcal{T}}_\epsilon )\widehat h \big\Vert_{L^1}&\leq \int_{\RR} d\xi    \frac{\overline{\chi}_{\epsilon^r}(\xi)}{4\pi^2\xi^2 + \epsilon^4\theta^2}\int_{\RR} d\zeta |\widehat{q_\epsilon}(\xi - \zeta)| \frac{\overline{\chi}_{\epsilon^r}(\zeta) }{4\pi^2 \zeta^2 + \epsilon^4 \theta^2}\int_{\RR} d\eta   |\widehat{q_\epsilon}(\zeta - \eta)| |\widehat{h}(\eta)|\\
&= \int_{\RR} d\eta  |\widehat{h}(\eta)| \iint_{\RR^2} d\xi d\zeta  \frac{\overline{\chi}_{\epsilon^r}(\xi)}{4\pi^2\xi^2 + \epsilon^4\theta^2} \frac{|\widehat{q_\epsilon}(\zeta - \eta)|\overline{\chi}_{\epsilon^r}(\zeta) }{4\pi^2 \zeta^2 + \epsilon^4 \theta^2} |\widehat{q_\epsilon}(\xi - \zeta)| d\zeta.
\end{align*}
Defining $f_\epsilon \equiv \frac{\overline{\chi}_{\epsilon^r}(\xi)}{4\pi^2\xi^2 + \epsilon^4\theta^2}$ and $g_\epsilon \equiv \frac{|\widehat{q_\epsilon}(\zeta - \eta)| \overline{\chi}_{\epsilon^r}(\zeta)}{4\pi^2 \zeta^2 + \epsilon^4 \theta^2}$, we can apply estimate~\eqref{small-2} from Lemma~\ref{lem:smallness}, and hypothesis (H1b), \textit{i.e.} bound~\eqref{as:Q=0-q-epsilon-decay} on $\widehat{q}_\epsilon$, to conclude 
\begin{align}
\big\Vert(\widehat{\mathcal{T}}_\epsilon\circ \widehat{\mathcal{T}}_\epsilon )\widehat h \big\Vert_{L^1} & \leq \big\Vert \widehat{h} \big\Vert_{L^1} \left[ \sup_{|\xi| \leq  \frac{1}{2\epsilon}} |\widehat{q}_\epsilon (\xi) |  \big\Vert f_\epsilon \big\Vert_{L^1(\RR)}\big\Vert g_\epsilon \big\Vert_{L^1(\RR)} \right. \nn \\
& \qquad + \left. \big\Vert \widehat{q}_\epsilon \big\Vert_{L^1(\RR)} \left( \sup_{|\zeta| \geq  \frac{1}{4\epsilon}} |f_\epsilon(\zeta)| \big\Vert g_\epsilon \big\Vert_{L^1(\RR)} + \big\Vert f_\epsilon \big\Vert_{L^1(\RR)} \sup_{|\zeta| \geq  \frac{1}{4\epsilon}} |g_\epsilon(\zeta)| \right) \right] \nn \\
& \leq C(\mathcal{C}_0, \mathcal{C}_N) (\epsilon^{N-2r} + \epsilon^{2-r}) \big\Vert \widehat{h} \big\Vert_{L^1}. \label{bnd:TT}
\end{align}
The final inequality above comes from noting 
\begin{align*}
\sup_{|\zeta| \geq  \frac{1}{4\epsilon}} |f_\epsilon(\zeta)| \leq C \epsilon^2, & \qquad \big\Vert f_\epsilon \big\Vert_{L^1} \leq C \epsilon^{-r}, \\
\sup_{|\zeta| \geq  \frac{1}{4\epsilon}} |g_\epsilon(\zeta)| \leq C(\big\Vert \widehat{q}_\epsilon \big\Vert_{L^\infty}) \epsilon^2, & \qquad \big\Vert g_\epsilon \big\Vert_{L^1} \leq C(\big\Vert \widehat{q}_\epsilon \big\Vert_{L^\infty})  \epsilon^{-r}.
\end{align*}
 
 It follows that if $r\in(0,2)$ and $N>2r$, there exists $\epsilon_0>0$ such that if $\epsilon<\epsilon_0$, then  one has $\big\Vert\widehat{\mathcal{T}}_\epsilon\circ \widehat{\mathcal{T}}_\epsilon\big\Vert_{L^1\to L^1}\leq \frac12$ and thus $(I-\widehat{\mathcal{T}}_\epsilon\circ \widehat{\mathcal{T}}_\epsilon)$ is invertible as an operator from 
  $L^1$ to $L^1$, with bound:
  \[ \big\Vert (I-\widehat{\mathcal{T}}_\epsilon\circ \widehat{\mathcal{T}}_\epsilon)^{-1}\big\Vert_{L^1\to L^1}\leq 2.\]
  
  We now estimate the right-hand side of~\eqref{eq-far-T} in $L^1$, which concludes the proof of Proposition~\ref{prop:far-in-terms-near}.
  
  First, one has immediately from~\eqref{bnd:TT} that 
  \[ \big\Vert(\widehat{\mathcal{T}}_\epsilon\circ \widehat{\mathcal{T}}_\epsilon) \widehat{\psi}_\nr \Vert_{L^1}\leq C(\mathcal{C}_0,\mathcal{C}_N)  \big(\epsilon^{N-2r}+\epsilon^{2-r}\big) \big\Vert \widehat{\psi}_\nr\big\Vert_{L^1}.\]
  Then, since $\widehat{\psi}_\nr(\zeta)=\chi_{\epsilon^r}(\zeta)\widehat{\psi}_\nr(\zeta)$, one has
  \begin{align*}
  \big\Vert\widehat{\mathcal{T}}_\epsilon\widehat{\psi}_\nr \big\Vert_{L^1}&\leq \int_{\RR}d\xi \frac{\overline{\chi}_{\epsilon^r}(\xi)}{4\pi^2 \xi^2 + \epsilon^4 \theta^2} \int_{\RR}d\zeta  |\widehat{q_\epsilon} (\xi - \zeta)| \chi_{\epsilon^r}(\zeta)| \widehat{\psi}_\nr(\zeta) |\\
  & = \int_\RR d\zeta  \chi_{\epsilon^r}(\zeta)| \widehat{\psi}_\nr(\zeta)| \int_\RR d\xi \frac{\overline{\chi}_{\epsilon^r}(\xi)}{4\pi^2 \xi^2 + \epsilon^4 \theta^2} |\widehat{q_\epsilon} (\xi - \zeta)|.
   \end{align*}
Defining $g_\epsilon=  \frac{\overline{\chi}_{\epsilon^r}(\xi)}{4\pi^2\xi^2 + \epsilon^4\theta^2}$, we can apply estimate~\eqref{small-1} from Lemma~\ref{lem:smallness}, hypothesis (H1b), \textit{i.e.} bound~\eqref{as:Q=0-q-epsilon-decay} on $\widehat{q}_\epsilon$, and using that $\epsilon^r \leq  \frac{1}{4\epsilon}$ for $\epsilon$ sufficiently small, we conclude 
\begin{align*}
\big\Vert\widehat{\mathcal{T}}_\epsilon\widehat{\psi}_\nr \big\Vert_{L^1} & \leq \big\Vert \widehat{\psi}_\nr \big\Vert_{L^1} \left[ \sup_{|\xi| \leq  \frac{1}{2\epsilon}} |\widehat{q}_\epsilon(\xi)| \big\Vert g_\epsilon \big\Vert_{L^1(\RR)} + \big\Vert \widehat{q}_\epsilon \big\Vert_{L^1(\RR)} \ \sup_{|\zeta|\geq  \frac{1}{4\epsilon}} |g_\epsilon (\zeta)| \right] \\
& \leq C( \mathcal{C}_0, \mathcal{C}_N) (\epsilon^{N-r} + \epsilon^2) \big\Vert \widehat{\psi}_\nr \big\Vert_{L^1}.
\end{align*}
The final inequality above comes from noting 
\[
\sup_{|\zeta| \geq  \frac{1}{4\epsilon}} |g_\epsilon(\zeta)| \leq C \epsilon^2, \qquad \big\Vert g_\epsilon \big\Vert_{L^1} \leq C \epsilon^{-r}.
\]

 Altogether, we proved
 \begin{align*} \big\Vert  \widehat{\psi}_\fr \big\Vert_{L^1}&\leq  \big\Vert (I-\widehat{\mathcal{T}}_\epsilon\circ \widehat{\mathcal{T}}_\epsilon)^{-1}\big\Vert_{L^1\to L^1}\Big(  \big\Vert\widehat{\mathcal{T}}_\epsilon\widehat{\psi}_\nr \big\Vert_{L^1}+\big\Vert(\widehat{\mathcal{T}}_\epsilon\circ \widehat{\mathcal{T}}_\epsilon) \widehat{\psi}_\nr \Vert_{L^1}\Big)\\
 &\leq C(\mathcal{C}_0,\mathcal{C}_N)  \big(\epsilon^{N-2r}+\epsilon^{2-r}\big) \big\Vert \widehat{\psi}_\nr\big\Vert_{L^1}.
 \end{align*}
 This completes the proof of Proposition~\ref{prop:far-in-terms-near}.
 \end{proof}

\subsection{Analysis of the near energy component}\label{subsec:Q=0-near}

By Proposition~\ref{prop:far-in-terms-near}, we have $\widehat{\psi}_\fr=\widehat{\psi}_\fr[\widehat{\psi}_\nr,\theta^2;\epsilon]$. Substitution into the near energy equation~\eqref{eq:Q=0-near}, we obtain a closed equation for $\widehat{\psi}_\nr (\xi)$:
\begin{equation}
\left( 4\pi^2 \xi^2 + \epsilon^4 \theta^2 \right) \widehat{\psi}_\nr (\xi) + \chi_{_{\epsilon^r}}(\xi) \int_{\zeta} \widehat{q_\epsilon}(\xi - \zeta) \left( \widehat{\psi}_\nr (\zeta) + \widehat{\psi}_\fr[\widehat{\psi}_\nr,\theta^2;\epsilon](\zeta) \right) d\zeta = 0. \label{eq:Q=0-near-2}
\end{equation}

The following Proposition reveals the leading order terms in~\eqref{eq:Q=0-near-2}.

\begin{Proposition}\label{prop:lead-ord-and-rem} Set $r\in(0,2)$. Assume that $q_\epsilon$ satisfies~\eqref{as:Q=0-q-epsilon-bound}, \eqref{as:Q=0-q-epsilon-decay}  and \eqref{as:Q=0-q-epsilon-effective}  of Theorem~\ref{thm:Qzero} with $N\geq2,N>2r$, and $0<t_-\leq \theta^2 \leq t_+ < \infty$.
Then there exists $\epsilon_0>0$ such that for any $0<\epsilon<\epsilon_0$ one can write~\eqref{eq:Q=0-near-2} as 
\begin{equation}
\left( 4\pi^2 \xi^2 + \epsilon^4 \theta^2 \right) \widehat{\psi}_\nr (\xi) - \chi_{_{\epsilon^r}}(\xi) \epsilon^2 B_{\eff} \int_{\RR}  \widehat{\psi}_\nr (\eta) d\eta = - \chi_{_{\epsilon^r}}(\xi) \big(\mathcal{R}[\theta] \widehat{\psi}_\nr\big) (\xi),
\label{eq:almost-eff-Q=0}
\end{equation}
where $\mathcal{R}[\theta]:L^1\to L^\infty$ is a linear mapping satisfying
\begin{equation}\label{eq:remainder-bound}
\big\Vert \mathcal{R}\left[ \theta \right]\widehat{\psi}_\nr \big\Vert_{L^\infty} \leq C(\mathcal{C}_0, \mathcal{C}_N,\mathcal{C}_{\eff},t_-,t_+)  \big\Vert \widehat{\psi}_\nr \big\Vert_{L^1} \times K_\epsilon,
\end{equation}
with $K_\epsilon=\left(\epsilon^{N-2}+\epsilon^2\right) \left(\epsilon^{N-2r}+\epsilon^{2-r}\right)+\epsilon^r\times \left(\epsilon^{N-2}+\epsilon^2\right)+ \epsilon^{2N}+\epsilon^{2N-2}+\epsilon^{2+\sigma_\eff}$.
\end{Proposition}
\begin{proof}
Using equations~\eqref{eq:Q=0-near} and~\eqref{eq:Q=0-far} to iterate once the near energy equation~\eqref{eq:Q=0-near-2} and interchanging the order of integration, we obtain
\begin{align}
\left( 4\pi^2 \xi^2 + \epsilon^4 \theta^2 \right) \widehat{\psi}_\nr (\xi) & = - \chi_{_{\epsilon^r}}(\xi) \int_{\zeta} \widehat{q_\epsilon}(\xi - \zeta) \left( \widehat{\psi}_\nr (\zeta) + \widehat{\psi}_\fr(\zeta) \right) d\zeta \nn \\
& =  \chi_{_{\epsilon^r}}(\xi) \int_{\zeta} \widehat{q_\epsilon}(\xi - \zeta) \left[  \frac{ \chi_{_{\epsilon^r}}(\zeta)}{4\pi^2 \zeta^2 + \epsilon^4 \theta^2 } \int_{\eta_1} \widehat{q_\epsilon} (\zeta - \eta_1) \left( \widehat{\psi}_\nr (\eta_1) + \widehat{\psi}_\fr(\eta_1) \right) d\eta_1 \right. \nn \\
& \qquad \qquad \qquad \qquad \qquad + \left. \frac{ \overline{\chi}_{\epsilon^r}(\zeta)}{4\pi^2 \zeta^2 + \epsilon^4 \theta^2 } \int_{\eta_2} \widehat{q_\epsilon} (\zeta - \eta_2) \left( \widehat{\psi}_\nr (\eta_2) + \widehat{\psi}_\fr(\eta_2) \right) d\eta_2 \right] d\zeta \nn \\
& = \chi_{_{\epsilon^r}}(\xi) \left[ \int_{\eta} \widehat{\psi}_\nr(\eta) \ \int_{\zeta} \frac{1}{4\pi^2 \zeta^2 + \epsilon^4 \theta^2} \widehat{q_\epsilon}(\xi - \zeta) \widehat{q_\epsilon}(\zeta - \eta) d\zeta d\eta \right. \nn \\
& \qquad \qquad \qquad \qquad \qquad + \left. \int_{\eta} \widehat{\psi}_\fr(\eta)\ \int_{\zeta} \frac{1}{4\pi^2 \zeta^2 + \epsilon^4 \theta^2} \widehat{q_\epsilon}(\xi - \zeta) \widehat{q_\epsilon}(\zeta - \eta) d\zeta d\eta \right] .\nn
\end{align}
 We rewrite this equation as
 \begin{equation}
 (4\pi^2 \xi^2 + \epsilon^4 \theta^2)\widehat{\psi}_\nr (\xi) =  \big(\mathcal{Q}[\theta]\widehat \psi_\nr\big)(\xi)+  \big(\mathcal{Q}[\theta]\widehat\psi_\fr\big)(\xi) , \label{eq:near-4}
 \end{equation}
where we recall the mapping $(\widehat\psi_\nr,\theta^2;\epsilon)\mapsto\widehat\psi_\fr[\widehat\psi_\nr,\theta^2;\epsilon]$, and denote
\[
 \big(\mathcal{Q}[\theta]\widehat\psi\big)(\xi) \equiv \chi_{_{\epsilon^r}}(\xi) \int_{\eta} \widehat{\psi}(\eta) \int_{\zeta} \frac{1}{4\pi^2 \zeta^2 + \epsilon^4 \theta^2} \widehat{q_\epsilon}(\xi - \zeta) \widehat{q_\epsilon}(\zeta - \eta) d\zeta d\eta.
\]

In what follows, we first show that the contribution of $\mathcal{Q}[\theta]\widehat\psi_\fr$ is small, and then extract the leading order term from $\mathcal{Q}[\theta]\widehat\psi_\nr$.
\medskip

\noindent {\em $L^\infty$ bound of $\mathcal{Q}[\theta]\widehat\psi_\fr$.}
First note 
\[
\big\Vert \mathcal{Q}[\theta]\widehat\psi_\fr\big\Vert_{L^\infty}\leq \sup_{|\xi| \leq \epsilon^r,\eta \in \RR} \int_{\zeta} \frac{1}{4\pi^2 \zeta^2 + \epsilon^4 \theta^2} | \widehat{q_\epsilon}(\xi - \zeta) ||\widehat{q_\epsilon}(\zeta - \eta) |d\zeta  \ \left\Vert \widehat{\psi}_\fr[\widehat\psi_\nr,\theta^2;\epsilon] \right\Vert_{L^1}.
\]
Using $\epsilon^r \leq  \frac{1}{4\epsilon}$ for $\epsilon$ sufficiently small, we can bound the factor multiplying $\left\Vert \widehat{\psi}_\fr[\widehat\psi_\nr,\theta^2;\epsilon] \right\Vert_{L^1}$ using estimate~\eqref{small-1} of Lemma~\ref{lem:smallness} with the choice $g_\epsilon(\zeta) = \frac{\sup_{\eta \in \RR} |\widehat{q}_\epsilon(\zeta - \eta)|}{4\pi^2 \zeta^2 + \epsilon^4 \theta^2}$ and applying hypothesis (H1b), \textit{i.e.} bound~\eqref{as:Q=0-q-epsilon-decay} of $\widehat{q}_\epsilon$. Noting that 
\[
\sup_{|\zeta| \geq  \frac{1}{4\epsilon}} |g_\epsilon(\zeta)| \leq C(\left\Vert \widehat{q}_\epsilon \right\Vert_{L^\infty}) \epsilon^2, \qquad \left\Vert g_\epsilon \right\Vert_{L^1} \leq C(\left\Vert \widehat{q}_\epsilon \right\Vert_{L^\infty})  \epsilon^{-2}|\theta|^{-1},
\]
one has the bound 
\begin{align}
\big\Vert \mathcal{Q}[\theta]\widehat\psi_\fr\big\Vert_{L^\infty}   &\leq C(\mathcal{C}_0,\mathcal{C}_N) \left(|\theta|^{-1}\epsilon^{N-2}+\epsilon^2\right) \big\Vert \widehat{\psi}_\fr[\widehat\psi_\nr,\theta^2;\epsilon]\big\Vert_{L^1} \nn \\
& \leq C(\mathcal{C}_0,\mathcal{C}_N,t_-) \left(\epsilon^{N-2}+\epsilon^2\right)\left(\epsilon^{N-2r}+\epsilon^{2-r}\right)\big\Vert \widehat{\psi}_\nr \big\Vert_{L^1},
\label{eq:est-Qfar}\end{align}
where the last estimate follows from Proposition~\ref{prop:far-in-terms-near}, and $0<t_-<\theta^2<t_+<\infty$.
\medskip

\noindent {\em Leading order expansion of $\mathcal{Q}[\theta]\widehat\psi_\nr$.} Let us first recall that $\widehat{\psi}_\nr(\eta)=\chi_{_{\epsilon^r}}(\eta)\widehat{\psi}_\nr(\eta)$, and consequently rewrite
\begin{equation}\label{eq:Q-near-1}
\big(\mathcal{Q}[\theta]\widehat\psi_\nr\big)(\xi)=\int_{\RR}\widehat{\psi}_\nr(\eta) \times \chi_{_{\epsilon^r}}(\xi)\chi_{_{\epsilon^r}}(\eta) \mathfrak{q}(\xi,\eta) d\eta ,
\end{equation}
with
\[ \mathfrak{q}(\xi,\eta)=\int_{\RR}\frac{1}{4\pi^2 \zeta^2 + \epsilon^4 \theta^2} \widehat{q_\epsilon}(\xi - \zeta) \widehat{q_\epsilon}(\zeta - \eta) d\zeta .\]

Our aim is to expand the pointwise first order term (in $\epsilon$) of $ \mathfrak{q}(\xi,\eta)$ for $(\xi,\eta)\in [-\epsilon^r,\epsilon^r]^2$. We write 
\begin{align}
\big(\mathcal{Q}[\theta]\widehat\psi_\nr\big)(\xi)& = \int_{\RR}d\eta \widehat{\psi}_\nr(\eta) \times \chi_{_{\epsilon^r}}(\xi)\chi_{_{\epsilon^r}}(\eta) \mathfrak{q}(0,0) \nn \\
& \qquad + \int_{\RR}d\eta \widehat{\psi}_\nr(\eta) \times \chi_{_{\epsilon^r}}(\xi)\chi_{_{\epsilon^r}}(\eta) \left[ \mathfrak{q}(\xi,\eta) - \mathfrak{q}(0,0) \right] \nn \\
& = \int_{\RR}d\eta \widehat{\psi}_\nr(\eta) \times \chi_{_{\epsilon^r}}(\xi)\chi_{_{\epsilon^r}}(\eta) \left[ \int_{\RR}\frac{\widehat{q_\epsilon}(- \zeta) \widehat{q_\epsilon}(\zeta) }{4\pi^2 \zeta^2 + 1} d\zeta \right] \nn \\
& \qquad + \int_{\RR}d\eta \widehat{\psi}_\nr(\eta) \times \chi_{_{\epsilon^r}}(\xi)\chi_{_{\epsilon^r}}(\eta) \left[ \int_{\RR}\frac{\widehat{q_\epsilon}(- \zeta) \widehat{q_\epsilon}(\zeta) }{4\pi^2 \zeta^2 + \epsilon^4 \theta^2}d\zeta - \int_{\RR}\frac{\widehat{q_\epsilon}(- \zeta) \widehat{q_\epsilon}(\zeta) }{4\pi^2 \zeta^2 + 1} d\zeta \right] \label{term1}\\
& \qquad + \int_{\RR}d\eta \widehat{\psi}_\nr(\eta) \times \chi_{_{\epsilon^r}}(\xi)\chi_{_{\epsilon^r}}(\eta) \left[ \mathfrak{q}(\xi,\eta) - \mathfrak{q}(0,0) \right]. \label{term2}
\end{align}
We will now bound the last two terms in the above sum. Firstly, using the Mean Value Theorem, one has
\[ \sup_{(\xi,\eta)\in [-\epsilon^r,\epsilon^r]^2} \big\vert \mathfrak{q}(\xi,\eta)-\mathfrak{q}(0,0) \big\vert\lesssim \epsilon^r \sup_{(\xi,\eta)\in [-\epsilon^r,\epsilon^r]^2}\left( \left| \frac{d}{d\xi} \mathfrak{q}(\xi,\eta) \right| + \left| \frac{d}{d\eta} \mathfrak{q}(\xi,\eta) \right| \right).\]
Using the symmetry properties of $\mathfrak{q}$, it suffices to estimate
\[ \left| \frac{d}{d\eta} \mathfrak{q}(\xi,\eta) \right| \ \leq \ \int_{\RR}\frac{1}{4\pi^2 \zeta^2 + \epsilon^4 \theta^2} | \widehat{q_\epsilon}(\xi - \zeta)| |\widehat{q_\epsilon}'(\zeta - \eta)| d\zeta .\]
Using estimate~\eqref{small-1} in Lemma~\ref{lem:smallness} with $g_\epsilon(\zeta) = \frac{|\widehat{q_\epsilon}'(\zeta - \eta)|}{4\pi^2 \zeta^2 + \epsilon^4 \theta^2} $, and hypothesis (H1b), \textit{i.e.} bound~\eqref{as:Q=0-q-epsilon-decay} on $\widehat{q}_\epsilon$, one obtains
\begin{align*}
 \left| \frac{d}{d\eta} \mathfrak{q}(\xi,\eta) \right| & \leq
\sup_{|\xi| \leq  \frac{1}{2\epsilon}} |\widehat{q}_\epsilon(\xi)| \left\Vert g_\epsilon \right\Vert_{L^1(\RR)} + \left\Vert \widehat{q}_\epsilon \right\Vert_{L^1(\RR)} \ \sup_{|\zeta|\geq  \frac{1}{4\epsilon}} |g_\epsilon (\zeta)| \\
& \leq C(\mathcal{C}_0,\mathcal{C}_N) \left(|\theta|^{-1}\epsilon^{N-2}+\epsilon^2\right),
\end{align*}
where we note that 
\[
\sup_{|\zeta|\geq  \frac{1}{4\epsilon}} |g_\epsilon (\zeta)| \leq C(\big\Vert  {\widehat{q}_\epsilon}' \big\Vert_{L^\infty(\RR)}) \epsilon^2, \qquad \big\Vert  g_\epsilon \big\Vert_{L^1(\RR)} \leq C(\big\Vert {\widehat{q}_\epsilon}' \big\Vert_{L^\infty(\RR)}) \epsilon^{-2}|\theta|^{-1} .
\]
Therefore, term~\eqref{term2} can be bounded as
\begin{equation}\label{bnd:term1}
\left| \int_{\RR}d\eta \widehat{\psi}_\nr(\eta) \times \chi_{_{\epsilon^r}}(\xi)\chi_{_{\epsilon^r}}(\eta) \left[ \mathfrak{q}(\xi,\eta) - \mathfrak{q}(0,0) \right] \right| \leq C(\mathcal{C}_0,\mathcal{C}_N) \ \epsilon^r \times \left(|\theta|^{-1}\epsilon^{N-2}+\epsilon^2\right) \left\Vert \widehat{\psi}_\nr \right\Vert_{L^1}.
\end{equation}

\medskip

As a second step, we study term~\eqref{term1}. In particular, we bound the integral 
\begin{align*}
\int_{\RR}\widehat{q_\epsilon}(- \zeta) \widehat{q_\epsilon}(\zeta) \left[ \frac{1}{4\pi^2 \zeta^2 + \epsilon^4 \theta^2} -\frac{1}{4\pi^2 \zeta^2 + 1} \right] d\zeta.
\end{align*}
To do so, we consider the above integral under two domains: $|\zeta| \leq  \frac{1}{4\epsilon}$ and $|\zeta| >  \frac{1}{4\epsilon}$. Notice that since $q_\epsilon$ satisfies hypothesis (H1b), \textit{i.e.} bound~\eqref{as:Q=0-q-epsilon-decay}, one has
 \begin{align*}
 \int_{|\zeta|\leq 1/(4\epsilon)}\frac{1}{4\pi^2\zeta^2 +1} \widehat{q_\epsilon}(- \zeta) \widehat{q_\epsilon}(\zeta) d\zeta &\leq \frac12\mathcal{C}_N^2 \epsilon^{2N}\\
  \int_{|\zeta|\leq 1/(4\epsilon)}\frac{1}{4\pi^2\zeta^2 +\epsilon^4 \theta^2} \widehat{q_\epsilon}(- \zeta) \widehat{q_\epsilon}(\zeta) d\zeta &\leq \frac12\mathcal{C}_N^2 |\theta|^{-1} \epsilon^{2N-2}.
  \end{align*}
Furthermore,
  \[ \int_{|\zeta|\geq 1/(4\epsilon)}\left|\frac{1}{4\pi^2\zeta^2 +1} -\frac{1}{4\pi^2\zeta^2 +\epsilon^4 \theta^2} \right||\widehat{q_\epsilon}(- \zeta) \widehat{q_\epsilon}(\zeta) |d\zeta\leq C(\epsilon^4\theta^2) \epsilon^4 \big\Vert\widehat{q_\epsilon}\big\Vert_{L^1}\big\Vert\widehat{q_\epsilon}\big\Vert_{L^\infty} , \]
and we conclude
\begin{align}
& \left| \int_{\RR}d\eta \widehat{\psi}_\nr(\eta) \times \chi_{_{\epsilon^r}}(\xi)\chi_{_{\epsilon^r}}(\eta) \left[ \int_{\RR}\frac{\widehat{q_\epsilon}(- \zeta) \widehat{q_\epsilon}(\zeta) }{4\pi^2 \zeta^2 + \epsilon^4 \theta^2}d\zeta - \int_{\RR}\frac{\widehat{q_\epsilon}(- \zeta) \widehat{q_\epsilon}(\zeta) }{4\pi^2 \zeta^2 + 1} d\zeta \right] \right| \nn \\
& \qquad \qquad \qquad \qquad \qquad \qquad \qquad \leq C(\mathcal{C}_0,\mathcal{C}_N)\big( \epsilon^{2N}+|\theta|^{-1} \epsilon^{2N-2}+C(\epsilon^4\theta^2)\epsilon^4\big) \left\Vert \widehat{\psi}_\nr \right\Vert_{L^1}. \label{eq:est-Qnear2} 
\end{align}

 Altogether, plugging estimates~\eqref{bnd:term1} and~\eqref{eq:est-Qnear2} into $\mathcal{Q}[\theta]\psi_\nr$ as defined in~\eqref{eq:Q-near-1}, yields 
\begin{equation}\label{eq:Q-near-2}
\big(\mathcal{Q}[\theta]\widehat\psi_\nr\big)(\xi) = \chi_{_{\epsilon^r}}(\xi) \left(\int_{\RR}\frac{\widehat{q_\epsilon}(- \zeta) \widehat{q_\epsilon}(\zeta) }{4\pi^2 \zeta^2 + 1} d\zeta\right) \times \int_{\RR}  \widehat{\psi}_\nr (\eta) d\eta + \chi_{_{\epsilon^r}}(\xi) \big(\mathcal{R}_1[\theta]\widehat{\psi}_\nr\big)  (\xi),
\end{equation}
 where the remainder $\mathcal{R}_1$ satisfies the bound
 \[
 \left\Vert \mathcal{R}_1[\theta] \widehat{\psi}_\nr \right\Vert_{L^\infty} \leq C(\mathcal{C}_0, \mathcal{C}_N) \left( \epsilon^r \times \left(|\theta|^{-1}\epsilon^{N-2}+\epsilon^2\right) + \epsilon^{2N}+|\theta|^{-1} \epsilon^{2N-2}+C(\epsilon^4\theta^2)\epsilon^4 \right) \left\Vert \widehat{\psi}_\nr \right\Vert_{L^1}.
\]
 
Furthermore, by Hypothesis (H2), expression~\eqref{as:Q=0-q-epsilon-effective}, we can write
\[
\left| \int_{\RR}\frac{\widehat{q_\epsilon}(- \zeta) \widehat{q_\epsilon}(\zeta) }{4\pi^2 \zeta^2 + 1} d\zeta \ - \  \epsilon^2 B_{\eff} \right|\ \leq \ \mathcal{C}_{\eff}\epsilon^{2+\sigma_\eff}. 
\]
Therefore, we can rewrite~\eqref{eq:Q-near-2} as
\begin{equation}\label{eq:Q-near-3}
\big(\mathcal{Q}[\theta]\widehat\psi_\nr\big)(\xi) = \chi_{_{\epsilon^r}}(\xi) \epsilon^2 B_{\eff} \int_{\RR}  \widehat{\psi}_\nr (\eta) d\eta + \chi_{_{\epsilon^r}}(\xi) \big(\mathcal{R}_2[\theta]\widehat{\psi}_\nr\big) (\xi),
\end{equation}
where the remainder $\mathcal{R}_2$ now satisfies the bound
\begin{equation}\label{eq:R2-bound}
\left\Vert \mathcal{R}_2[\theta] \widehat{\psi}_\nr  \right\Vert_{L^\infty} \leq C(\mathcal{C}_0, \mathcal{C}_N,\mathcal{C}_{\eff},t_-,t_+) \left( \epsilon^r \times \left(\epsilon^{N-2}+\epsilon^2\right) + \epsilon^{2N}+\epsilon^{2N-2}+\epsilon^{2+\sigma_\eff} \right) \left\Vert \widehat{\psi}_\nr \right\Vert_{L^1},
\end{equation}
where we used $0<t_-<\theta^2<t_+<\infty$.
 
We conclude the proof of Proposition~\ref{prop:lead-ord-and-rem} by plugging in expression~\eqref{eq:Q-near-3} and estimates~\eqref{eq:est-Qfar} and~\eqref{eq:R2-bound} into~\eqref{eq:near-4}.  
\end{proof}

\noindent{\bf Rescaling the equation.} We now proceed with the analysis of the near equation with the rescaling $\xi$ and $\widehat{\psi}_\nr$ in such a way as to balance both terms on the left hand side of~\eqref{eq:almost-eff-Q=0}. Thus we define
\begin{align*}
\widehat{\psi}_\nr(\xi) = \frac{1}{\epsilon^2}\widehat{\Phi} \left( \frac{\xi}{\epsilon^2} \right) = \frac{1}{\epsilon^2}\widehat{\Phi}(\xi'),\ \   \xi = \epsilon^2 \xi'\ .
\end{align*}
Note that 
\begin{align*}
\big\Vert \widehat{\psi}_\nr \big\Vert_{L^1} = \big\Vert \widehat{\Phi} \big\Vert_{L^1}.
\end{align*}
Equation~\eqref{eq:almost-eff-Q=0} then becomes, after dividing out by $\epsilon^2$, 
\begin{align*}
 \left( 4\pi^2 \xi'^2 + \theta^2 \right) \widehat{\Phi} (\xi')  - \chi_{_{\epsilon^{r-2}}}(\xi')  B_{\eff}  \int_{\zeta'} \widehat{\Phi} (\zeta') d\zeta' = -\epsilon^{-2} \chi_{_{\epsilon^{r-2}}}(\xi') \left(\mathcal{R}[\theta]\left\{\frac1{\epsilon^2} \widehat{\Phi}(\frac{\cdot}{\epsilon^2})\right\}\right) (\epsilon^2 \xi').
\end{align*}
By estimate~\eqref{eq:remainder-bound} and choosing carefully the parameters $r$ and $ N$, we can ensure that the right hand side is small. The following Proposition summarizes our result, with $N=4$ and $r=1$.
\begin{Proposition}\label{summary-Q0} Assume that the assumptions of Proposition~\ref{prop:lead-ord-and-rem} hold with $r=1$ and $N=4$.
Then one has
\begin{equation}\label{eq:near-leading-order}
\left( 4\pi^2 \xi'^2 + \theta^2 \right) \widehat{\Phi}(\xi')  - \chi_{_{\epsilon^{-1}}}(\xi') B_{\eff} \int_{\zeta} \widehat{\Phi}(\zeta') d\zeta' =  -\chi_{_{\epsilon^{-1}}}(\xi') \big(\widetilde{\mathcal{R}}[\theta]  \widehat{\Phi} \big)( \xi'),
\end{equation}
where $\widetilde{\mathcal{R}}[\theta] \widehat{\Phi}:\xi'\mapsto \epsilon^{-2} \chi_{_{\epsilon^{-1}}}(\xi') \Big(\mathcal{R}[\theta]\big\{ \frac1{\epsilon^2} \widehat{\Phi}(\frac{\cdot}{\epsilon^2}) \big\}\Big) (\epsilon^2\xi')$ 
satisfies the bound, for $\sigma = \min \{ 1, \sigma_\eff \}$,
\begin{equation}\label{bound:requirement}
\big\Vert  \chi_{_{\epsilon^{-1}}}(\xi) \big(\widetilde{\mathcal{R}}[\theta]  \widehat{\Phi} \big) (\xi)\big\Vert_{L^\infty(\RR_\xi)} \ \leq\ \epsilon^\sigma \  C(\mathcal{C}_0,\mathcal{C}_N,\mathcal{C}_{\eff},t_-,t_+) \big\Vert \widehat{\Phi} \big\Vert_{L^1}.
\end{equation}
\end{Proposition}

\subsection{Conclusion of proof of Theorem~\ref{thm:Qzero}}\label{subsec:Q=0-conclusion}

Proposition~\ref{summary-Q0} is a formal reduction of the eigenvalue problem 
\begin{equation}\label{eq:Q=0-orig-evp-0}
(-\partial_x^2 + q_\epsilon) \psi^\epsilon = E^\epsilon \psi^\epsilon, \quad \psi^\epsilon \in L^2(\RR),
\end{equation}
for $(E^\epsilon, \psi^\epsilon)$ to an equation for $(\theta_\epsilon^2, \widehat{\Phi}_\epsilon)$ of the form:
\begin{equation}\label{eq:Q=0-leading-orderA}
\widehat{\mathcal{L}} _{0,\epsilon}[\theta_\epsilon]\widehat{\Phi}_\epsilon = (4\pi^2 \xi'^2 + \theta^2) \widehat{\Phi}_\epsilon(\xi')- \chi\left(|\xi'|<\epsilon^{-1}\right) B_{\eff} \int_{\RR}\widehat{\Phi}_\epsilon(\zeta')d\zeta' = - \chi\left(|\xi'|<\epsilon^{-1}\right)\big(\widetilde{\mathcal{R}}[\theta_\epsilon]\widehat{\Phi}_\epsilon\, \big) ( \xi');
\end{equation}
(see~\eqref{eq:near-leading-order}) where $\Phi_\epsilon$ is the rescaled near-energy component of $\psi_\epsilon$. We now apply Lemma~\ref{lem:technical} to obtain a solution of~\eqref{eq:Q=0-leading-orderA}. We then construct the solution $(E^\epsilon, \psi^\epsilon)$ of the full eigenvalue problem~\eqref{eq:Q=0-orig-evp-0}. This will conclude the proof of Theorem~\ref{thm:Qzero}.

\medskip

We apply Lemma~\ref{lem:technical} to equation~\eqref{eq:Q=0-leading-orderA} with $A = 1$ and $B = B_{\eff}>0$, and $R_\epsilon=\widetilde{\mathcal{R}}$. By Proposition~\ref{summary-Q0}, $R_\epsilon$ satisfies assumption~\eqref{assumptionsR} with $\beta=1$ and $\alpha=\sigma=\min \{1, \sigma_\eff \}$. Following the steps of its proof, and using
\[ \left| \frac1{4\pi^2 \zeta^2+\epsilon^4 \theta_1^2}-\frac1{4\pi^2 \zeta^2+\epsilon^4 \theta_2^2}\right|=\frac{\epsilon^4 |\theta_1^2-\theta_2^2|}{(4\pi^2 \zeta^2+\epsilon^4 \theta_1^2)(4\pi^2 \zeta^2+\epsilon^4 \theta_2^2)}\leq \frac{|\theta_1^2-\theta_2^2|}{\theta_2^2}\frac1{4\pi^2 \zeta^2+\epsilon^4 \theta_1^2},\]
one easily checks that assumption~\eqref{assumptionsR-theta} also holds.

By Lemma~\ref{lem:technical} there exists a solution $(\theta_\epsilon^2, \widehat{\Phi}_\epsilon)$ of~\eqref{eq:Q=0-leading-orderA}, satisfying 
\begin{equation}\label{eq:bounds_rescaled_near}
\big\Vert\widehat{\Phi}_{\epsilon} - \widehat{f}_{0,\epsilon}\big\Vert_{L^1}\ \lesssim\ \epsilon^\sigma \qquad \text{ and } \qquad |\theta_\epsilon^2 - \theta_{0,\epsilon}^2|\ \lesssim\ \epsilon^\sigma\ . 
\end{equation} 
Here $(\theta_{0,\epsilon}^2, \widehat{f}_{0,\epsilon})$ is the solution of the homogeneous equation 
\[
\widehat{\mathcal{L}} _{0,\epsilon}[\theta_{0,\epsilon}]\widehat{f}_{0,\epsilon} = (4\pi^2 \xi'^2 + \theta_{0,\epsilon}^2) \widehat{f}_{0,\epsilon}(\xi') - \chi\left(|\xi'|<\epsilon^{-1}\right) B_{\eff} \int_{\RR}\chi\left(|\zeta'|<\epsilon^{-1}\right)\widehat{f}_{0,\epsilon}(\zeta')d\zeta' = 0,
\]
as described in Lemma~\ref{lem:homogeneous}. Specifically,
\begin{equation}\label{eq:f-theta-def}
\widehat{f}_{0,\epsilon}(\xi) = \frac{\chi(|\xi|<\epsilon^{-1})}{4\pi^2\xi^2 + \theta_{0,\epsilon}^2}, \quad \text{ and } \quad \left|\theta_{0,\epsilon}^2 \ - \  \frac{B_\eff^2}{4} \right| \ \lesssim \ \epsilon.
\end{equation}

We next construct the eigenpair solution $(E^\epsilon, \psi^\epsilon)$ of the Schr\"odinger equation~\eqref{eq:Q=0-start}. Define, using Proposition~\ref{prop:far-in-terms-near}, 
\begin{align*}
\psi^\epsilon & = \psi_\nr^\epsilon + \psi_\fr^\epsilon, \qquad E^\epsilon = - \epsilon^4 \theta_\epsilon^2, \\
{\rm where } \ \ \widehat{\psi}_\nr^\epsilon(\xi) & = \frac{1}{\epsilon^2} \widehat{\Phi}_{\epsilon}\left(\frac{\xi}{\epsilon^2}\right), \quad \text{ and } \quad \widehat{\psi}_\fr^\epsilon(\xi) = \widehat{\psi}_\fr [\widehat{\psi}_\nr^\epsilon, E^\epsilon; \epsilon](\xi). 
\end{align*}
Then $(E^\epsilon, \psi^\epsilon)$ is a solution of the eigenvalue problem~\eqref{eq:Q=0-orig-evp-0}. Indeed, the steps proceeding from~\eqref{eq:Q=0-orig-evp-0} to~\eqref{eq:Q=0-leading-orderA} are reversible solutions of $\psi^\epsilon \in \mathcal Z=\{ f\in C(\RR),\ \widehat f\in L^1(\RR)\}$ of~\eqref{eq:Q=0-orig-evp-0}, respectively, solutions $\Phi_\epsilon \in \mathcal Z$ of~\eqref{eq:Q=0-leading-orderA}. 

We now prove the estimates~\eqref{eq:e-value-approx} and~\eqref{eq:e-fct-approx}. Estimate~\eqref{eq:e-value-approx}, the small $\epsilon$ expansion of the eigenvalue $E^\epsilon$, follows from~\eqref{eq:bounds_rescaled_near},~\eqref{eq:f-theta-def} and the triangle inequality. Specifically, since we defined $E^\epsilon = - \epsilon^4 \theta_\epsilon^2$, we have
\[
\left| E^\epsilon + \epsilon^4 \frac{B_\eff^2}{4} \right| = \epsilon^4 \left| \theta_\epsilon^2 - \frac{B_\eff^2}{4} \right| \leq \epsilon^4 \left(\left| \theta_\epsilon^2 - \theta_{0,\epsilon}^2 \right| + \left| \theta_{0,\epsilon}^2 - \frac{B_\eff^2}{4} \right|\right) \lesssim  \epsilon^{4+\sigma}.
\]

The approximation,~\eqref{eq:e-fct-approx}, of the corresponding eigenstate, $\psi^\epsilon = \psi_\nr + \psi_\fr$, is obtained as follows. One has, by triangular inequality,
\begin{equation}
\sup_{x\in\RR}\left| \psi^\epsilon(x) - \frac{2}{B_\eff} \exp \left(- \epsilon^2 \frac{B_\eff}{2}|x| \right) \right| \leq \sup_{x\in\RR}\left| \psi_\nr^\epsilon(x) - \frac{2}{B_\eff} \exp \left(- \epsilon^2 \frac{B_\eff}{2}|x| \right) \right| + \Big\Vert \psi_\fr^\epsilon \Big\Vert_{L^\infty}. \label{eq:final0}
\end{equation}
We will look at the bounds in~\eqref{eq:final0} separately. 

Recall,
\begin{align*}
\psi_\nr(x) & = \int_\RR \widehat{\psi}_\nr^\epsilon(\xi) e^{2\pi ix\xi} \ d\xi = \int_\RR \frac{1}{\epsilon^2} \widehat{\Phi}_\epsilon \left( \frac{\xi}{\epsilon^2} \right) e^{2\pi ix\xi} \ d\xi \\
& =  \int_\RR \widehat{\Phi}_\epsilon (\eta ) e^{2\pi i\epsilon^2 x \eta} \ d\eta \\
& = \int_\RR \widehat{f}_{0,\epsilon} (\eta ) e^{2\pi i\epsilon^2 x \eta} \ d\eta + \int_\RR \left( \widehat{\Phi}_\epsilon (\eta ) - \widehat{f}_{0,\epsilon} (\eta ) \right) e^{2\pi i\epsilon^2 x \eta} \ d\eta \\
& = I_1(x) + I_2(x).
\end{align*}
For $A = 1$ and $B = B_\eff>0$, one has from estimate~\eqref{eq:asymptotic-f0} in Lemma~\ref{lem:technical}, 
\begin{equation}\label{eq:Q=0-I1}
\left| I_1(x) \ -\  \frac{2}{B_\eff} \exp \left(- \epsilon^2 \frac{B_\eff}{2}|x| \right) \right|\lesssim \epsilon.
\end{equation}
Using the first bound in~\eqref{eq:bounds_rescaled_near}, one has
\begin{equation}
\big\Vert I_2 \big\Vert_{L^\infty}  = \sup_{x\in\RR} \left| \int_\RR \left( \widehat{\Phi}_\epsilon (\eta ) - \widehat{f}_{0,\epsilon} (\eta ) \right) e^{2\pi i\epsilon^2 x \eta} \ d\eta \right| \leq   \big\Vert \widehat{\Phi}_\epsilon - \widehat{f}_{0,\epsilon} \big\Vert_{L^1} \lesssim \epsilon^{\sigma}. \label{eq:Q=0-I2}
\end{equation}
From estimate~\eqref{eq:Q=0-I1}-\eqref{eq:Q=0-I2}, we can write
\begin{equation} \sup_{x\in\RR}\left| \psi_\nr(x) - \frac{2}{B_\eff} \exp \left(- \epsilon^2 \frac{B_\eff}{2}|x| \right) \right|\lesssim \epsilon^\sigma. \label{bnd:final0-1}
\end{equation}

To bound the second norm in~\eqref{eq:final0}, we note that from Proposition~\ref{prop:far-in-terms-near} with $N=4$ and $r=1$, one has
 \begin{equation}\label{bnd:final0-2}
\Big\Vert \psi_\fr^\epsilon \Big\Vert_{L^\infty}\ \leq \ \Big\Vert \widehat{\psi}_\fr^\epsilon[\widehat{\psi}_\nr^\epsilon, E^\epsilon;\epsilon] \Big\Vert_{L^1}\ \lesssim \ \epsilon\ \Big\Vert \widehat{\psi}_\nr^\epsilon \Big\Vert_{L^1} \lesssim \epsilon, 
 \end{equation}
since $\big\Vert \widehat{\psi}_\nr^\epsilon \big\Vert_{L^1}= \big\Vert\widehat{\Phi}_{\epsilon} \big\Vert_{L^1} \to\big\Vert \widehat{f}_{0,\epsilon} \big\Vert_{L^1} \lesssim 1$ (as $\epsilon\to0$). 

Since $\psi^\epsilon$ is a unique solution of~\eqref{eq:Q=0-orig-evp-0} up to a multiplicative constant, we can conclude from~\eqref{eq:final0} and the estimates~\eqref{bnd:final0-1}-\eqref{bnd:final0-2}, that 
\begin{align*}
\sup_{x\in\RR}\left|  \psi^\epsilon(x) - \exp \left( -\epsilon^2 \frac{B_\eff}{2} |x| \right) \right| \lesssim \epsilon^{\sigma}, \qquad \sigma = \min \{1, \sigma_\eff \}.
\end{align*}
This completes the proof of Theorem~\ref{thm:Qzero}.

\begin{Remark}
Note that above we conclude that $\psi^\epsilon,\Phi_\epsilon  \in \mathcal Z=\{f\in  C(\RR), \widehat f\in L^1(\RR)\}$ while we are in fact studying the eigenvalue problem~\eqref{eq:Q=0-orig-evp-0} with $\psi^\epsilon \in L^2(\RR)$. Notice that, by definition, $\widehat{\psi}^\epsilon$ is solution to
\[ (4\pi^2 \xi^2 -E^\epsilon)\widehat{\psi}^\epsilon(\xi) + \int_{\zeta} \widehat{q_\epsilon}(\xi - \zeta) \widehat{\psi}^\epsilon(\zeta) d\zeta = 0, \quad  \widehat{\psi}^\epsilon\in L^1(\RR)\]
and therefore satisfies the following inequality (recall $E^\epsilon<0$):
\[ |\widehat{\psi}^\epsilon(\xi) |\leq \frac1{4\pi^2 \xi^2 -E^\epsilon}\big\Vert \widehat{q_\epsilon}\big\Vert_{L^\infty}\big\Vert \widehat{\psi}^\epsilon\big\Vert_{L^1}.\]
One deduces immediately $\psi^\epsilon\in L^2(\RR)$. 
\end{Remark}

\section{Proof of Theorem~\ref{thm:per_result};\\ Edge bifurcations for $-\partial_x^2 + Q(x) + q_\epsilon(x)$} \label{sec:Qnot0}

We now prove Theorem~\ref{thm:per_result} concerning solutions of the eigenvalue problem
\begin{align}\label{eq:orig_eqn_per}
\left( - \partial_x^2 + Q(x) \right) \psi(x) + q_{\epsilon}(x)\psi(x) = E \psi(x), \quad \psi \in L^2(\RR).
\end{align}
Here $Q$ is one-periodic and satisfies Hypothesis~(HQ), \textit{i.e.} assumption~\eqref{as:Qn0-Q-pb}; and $q_\epsilon$ is localized at high frequencies, and decaying as $|x|\to \infty$ in the sense of Hypothesis~(H1'a-b), \textit{i.e.} assumptions~\eqref{as:Qn0-q-epsilon-bound} and~\eqref{as:Qn0-q-epsilon-decay}, and satisfies additionally Hypothesis~(H2'), \textit{i.e.} assumption~\eqref{as:Qn0-q-epsilon-effective}. Without loss of generality, we assume thereafter $C_0\leq \cdots\leq C_6$ and $\mathcal C_1\leq \cdots\leq \mathcal C_6$. 

Following the analysis of Section~\ref{sec:Q=0}, we divide the problem into a coupled system for a ``far-energy'' component and a ``near-energy'' component (here, ``near'' refers to $E$ being close to $E_{b_*}(k_*)$ a lowermost endpoint of a spectral band of $-\partial_x^2+Q(x)$ bordering a gap; see Section~\ref{subsec:FB-theory}. See our discussion of the strategy in Section~\ref{sec:strategy}. 

In order to spectrally localize we use the Gelfand-Bloch transform, introduced in Section~\ref{subsec:GB-transform}. For fixed $k_*\in \{0,1/2\}$ and $b_*\in\NN$, we define
\begin{equation}\label{eq:psi_decomp}
\psi = \psi_\nr + \psi_\fr = \mathcal{T}^{-1}\left\{\widetilde{\psi}_\nr(k)p_{b_*}(x;k)\right\} + \mathcal{T}^{-1}\left\{\sum_{b=0}^\infty\widetilde{\psi}_{\fr,b}(k)p_b(x;k)\right\} ,
\end{equation}
with
 \begin{align*} 
 \widetilde{\psi}_\nr(k) &\equiv \chi\left(|k-k_*|<\epsilon^r\right) \mathcal{T}_{b_*}\{\psi\}(k)= \chi\left(|k-k_*|<\epsilon^r\right) \left\langle p_{b_*}(x,k), \widetilde{\psi}(x,k) \right\rangle_{L^2_{\rm per}([0,1]_x)} ,\\
 \widetilde{\psi}_{\fr,b}(k) &\equiv \chi\left(|k-k_*| \geq \epsilon^r\delta_{b_*,b}\right) \mathcal{T}_{b_*}\{\psi\}(k)=\chi\left(|k-k_*| \geq \epsilon^r\delta_{b_*,b}\right) \left\langle p_{b}(x,k), \widetilde{\psi}(x,k) \right\rangle_{L^2_{\rm per}([0,1]_x)},
\end{align*}
and where $\delta_{i,j}$ denotes Kronecker's delta function.
 Equivalently, one has
 \[ \psi(x) \ = \ \int_{-1/2}^{1/2} \left( \widetilde{\psi}_\nr(k)u_{b_*}(x;k)+\sum_{b=0}^\infty \widetilde{\psi}_{\fr,b}(k)u_b(x;k) \right) \ dk.\]

 In Section~\ref{subsec:Qnot0-nearandfar} we introduce the coupled system of equations, equivalent to~\eqref{eq:orig_eqn_per}, in terms of $\psi_\fr$ and $\psi_\nr$. In Sections~\ref{subsec:Qnot0-far} and~\ref{subsec:Qnot0-near} we analyse the far and near energy components, respectively, in more detail. Finally, in Section~\ref{subsec:Qnot0-conclusion} we complete the proof of Theorem~\ref{thm:per_result}. 
\medskip

\noindent {\em For clarity of presentation and without any loss of generality, we assume henceforth that we are localizing near the lowermost endpoint of the $(b_*)^{th}$ band and that $k_*=0$. Thus, by Lemma~\ref{lem:band-edge},
\[ b_*\ \ \text{is even},\ \ \text{thus} \ \ k_* = 0,\ \ \text{and}\ \ E_{b_*}(0) \equiv E_\star.\]
N.B. For $k_*=0$, note that $p_b(x;k_*)=u_b(x;k_*)$ and we use these expressions interchangeably. For $k_*=1/2$ one has to distinguish between $p_b(x;k_*)$ and $u_b(x;k_*)$.}

\subsection{Near and far energy components}\label{subsec:Qnot0-nearandfar}

We first take the Gelfand-Bloch transform of~\eqref{eq:orig_eqn_per}. By~\eqref{T-derivative}, we obtain
\begin{equation}\label{eq:GB_of_orig}
-\left(\partial_x + 2\pi ik\right)^2 \widetilde{\psi}(x;k) + Q(x) \widetilde{\psi}(x;k) + \left(q_\epsilon \psi \right)^{\sim}(x;k) = E \widetilde{\psi}(x;k).
\end{equation}
Recall that $\{p_b(x;k)\}_{b \geq 0}$  form a complete orthonormal set in $L_{\rm per}^2([0,1]_x)$, and satisfy 
\begin{align}
&\left( -\left(\partial_x + 2\pi ik\right)^2 + Q(x) \right) p_b(x;k) = E_b(k)p_b(x;k), \quad p_b(x+1;k)\ =\ p_b(x;k), \ x\in\RR \ .\label{bloch-transf-evp}
\end{align}
Taking the inner product of~\eqref{eq:GB_of_orig} with $p_b(x;k)$, and using self-adjointness of $-\left(\partial_x + 2\pi ik\right)^2 + Q$ and~\eqref{bloch-transf-evp}, it follows $\psi\in L^2(\RR)$ satisfies~\eqref{eq:orig_eqn_per} if and only if

\begin{equation}\label{eq:kpsi}
 \left(E_b(k) - E\right) \big\langle p_{b}(x,k), \widetilde{\psi}(x,k) \big\rangle_{L^2_{\rm per}([0,1]_x)} + \big\langle p_{b}(x,k), \left(q_\epsilon \psi \right)^{\sim}(x,k) \big\rangle_{L^2_{\rm per}([0,1]_x)} = 0
\end{equation}
for all $ b\in\NN$ and $k\in (-1/2,1/2]$. Equivalently, using notation~\eqref{def:Tb},
\begin{equation}\label{eq:kpsi_prime}
\left(E_b(k) - E \right) \mathcal{T}_b\left\{\psi\right\}(k) +  \mathcal{T}_b\left\{q_\epsilon \psi \right\}(k) = 0, \quad \forall b\in\NN, k\in (-1/2,1/2].
\end{equation}

We now decompose~\eqref{eq:kpsi_prime} into near- and far-energy equations relative to the band edge $E_{b_*}(k_*)$. In the notation introduced in~\eqref{eq:psi_decomp}:
\begin{align}\label{eq:Qn0-near}
\left( E_{b_*}(k) - E \right) \widetilde{\psi}_\nr(k) +  \chi\left(|k|<\epsilon^r\right)\left(\mathcal{T}_{b_*} \left\{q_\epsilon \psi_\nr\right\}(k) + \mathcal{T}_{b_*}\left\{q_\epsilon \psi_\fr\right\}(k) \right) &= 0,\\
\label{eq:Qn0-far}
\left( E_b(k) - E \right) \widetilde{\psi}_{\fr,b}(k) + \chi\left(|k|\geq\epsilon^r\delta_{b_*,b}\right) \left(\mathcal{T}_b\left\{q_\epsilon\psi_\nr\right\}(k) + \mathcal{T}_b\left\{q_\epsilon\psi_\fr\right\}(k) \right) &= 0.
\end{align}
Equations~\eqref{eq:Qn0-near} and~\eqref{eq:Qn0-far} are, for the case of non-trivial periodic potentials, $Q(x)$, the analogue of~\eqref{eq:Q=0-near}-\eqref{eq:Q=0-far}.

\subsection{Analysis of the far energy components}\label{subsec:Qnot0-far}

We view the system of equations~\eqref{eq:Qn0-far} for $\{ \wt{\psi}_{\fr,b}(k) \}_{b\geq 0}$ as depending on ``parameters'' $(\wt{\psi}_\nr, E; \epsilon)$ and construct the mapping $(\wt{\psi}_\nr, E; \epsilon) \mapsto \psi_\fr [\psi_\nr, E; \epsilon]$ in the following proposition.

\begin{Proposition}\label{prop:Qn0-far-intermsof-near}
Assume $b_*$ is even and $E_{\star} = E_{b_*}(0)$ the lowermost edge of the $(b_*)^{th}$ band is at the boundary of a spectral gap. Let $E<E_\star$ vary over a subset of the gap which is uniformly bounded away from the $(b_* - 1)^{\rm st}$ band (note: $E$ may be arbitrarily close to $E_\star$). Assume $q_\epsilon\in L^2\cap L^\infty$ is bounded and localized at high frequencies in the sense of~\eqref{as:Qn0-q-epsilon-bound},\eqref{as:Qn0-q-epsilon-decay} with $\beta\geq 2$. Let $ \psi_{\nr}=\mathcal{T}^{-1}\left\{\widetilde{\psi}_\nr(k)p_{b_*}(x;k)\right\} $ with $\widetilde{\psi}_\nr\in L^2_\nr$, where
\begin{equation} \label{eq:psi_localized}
L^2_\nr\equiv \left\{ f\in L^2\big((-1/2,1/2]\big) \ : \  f(k)=\chi\left(|k|<\epsilon^r\right) f(k) \right\}.
\end{equation}
 Then for any $0<r<1/2$,
% %
% %
 there exists $\epsilon_0>0$, such that for $0<\epsilon<\epsilon_0$, the following holds. 
 
There is a unique solution $\left\{\widetilde{\psi}_{\fr,b}(k)\right\}_{b\geq 0}$, and $\psi_{\fr}=\mathcal{T}^{-1}\left\{\sum_{b\geq 0}\widetilde{\psi}_{\fr,b}(k)p_{b}(x;k)\right\}\in L^2(\RR)$ of the far-frequency system~\eqref{eq:Qn0-far} .
For any $E,\epsilon$ as above, the mapping $ \wt{\psi}_\nr \mapsto \psi_\fr\left[ \wt{\psi}_\nr, E; \epsilon \right] $ is a linear mapping from $L^2_\nr$ to $L^{2}(\RR)$ , and satisfies the bound 
\begin{equation}\label{bnd:Qn0-far-in-terms-near}
\big\Vert \psi_\fr \left[ \wt{\psi}_\nr, E; \epsilon \right] \big\Vert_{L^2} \leq C(\mathcal{C}_{0},\mathcal{C}_{\beta},b_*) \epsilon^{2-2r}   \big\Vert \widetilde{\psi}_\nr \big\Vert_{L^2}.  
\end{equation}
Moreover, for any $s\in (\frac12,\frac32)$, and for $\epsilon $ sufficiently small, one has $\psi_\fr\big[\psi_\nr, E; \epsilon\big]\in H^s(\RR)$ and
\begin{equation}\label{bnd:Qn0-far-in-terms-near-Hs}
\big\Vert \psi_\fr \left[ \wt{\psi}_\nr, E; \epsilon \right] \big\Vert_{H^s} \leq C(\mathcal{C}_{0},\mathcal{C}_{\beta},b_*,s) \epsilon^{2-\max\{2r,s\}}   \big\Vert \widetilde{\psi}_\nr \big\Vert_{L^2}.  
\end{equation}
\end{Proposition}
\begin{proof}
We begin by showing that 
there is a constant $0<C_1<\infty$, independent of $\epsilon$, such that 
\begin{align}
\left| E_{b_*}(k) - E \right| & \geq C_1 \epsilon^{2r}, \ \ \epsilon^r \leq |k| \leq 1/2, \label{eq:E_b*-E*} \\
\left| E_{b}(k) - E\right| & \geq C_1, \qquad b \neq b_*, \ \ |k| \leq 1/2. \label{eq:E_b-E*}
\end{align}
Note first that~\eqref{eq:E_b-E*} is an immediate consequence of the assumption on $E$. To prove~\eqref{eq:E_b*-E*} recall, by Lemma~\ref{lem:band-edge} that $E_\star = E_{b_*}(0)$, an eigenvalue at the edge of a spectral gap, is simple, and $k \rightarrow E_{b_*}(k) -E_\star$ is continuous. Therefore, for any $k_1$, such that $0<k_1\leq 1/2$,
\begin{align}\label{eq:min-diff}
\min_{k_1 \leq |k| \leq 1/2} \left| E_{b_*}(k) - E_\star \right| \geq C(k_1) > 0. 
\end{align}
For $|k|\leq k_1$, we approximate $E_{b_*}(k)$ by Taylor expansion. In particular, since $E_{b_*}(k)$ is analytic for $k$ near $k_* = 0$, $\partial_k E_{b_*}(0) = 0$ and $\partial_k^2 E_{b_*}(0) \neq 0$, we have $|E_{b_*}(k) - E_{b_*}(0) - \frac{1}{2} \partial_k^2 E_{b_*}(0) k^2 |\leq  \mathcal{C}|k|^3$. Therefore, we can choose $0<k_1\leq \frac{1}{6\mathcal{C}} |\partial_k^2 E_{b_*}(0)|$,  so that for all $\epsilon^r \leq k_1$ we have 
\begin{align}\label{eq:diff}
\min_{\epsilon^r \leq |k| \leq k_1}\left| E_{b_*}(k) - E_\star \right| \geq \frac{1}{3} \left| \partial_k^2 E_{b_*}(0) \right| \epsilon^{2r}. 
\end{align} 
Finally, notice that since $E<E_\star$, and $E_\star$ is the lowermost edge of the $(b_*)^{th}$ band,  we have $|E_{b}(k) - E| \geq |E_{b}(k) - E_\star| $, and therefore~\eqref{eq:E_b-E*} follows from~\eqref{eq:min-diff} and~\eqref{eq:diff}.
\medskip

Thanks to the above, we can rewrite the far-frequency system,~\eqref{eq:Qn0-far}, as 
\begin{align}\label{eq:Qn0-far-rewrite}
\wt{\psi}_{\fr,b} (k) + \frac{\chi\left(|k|\geq \epsilon^r \delta_{b_*,b}\right)}{E_{b}(k) - E} \mathcal{T}_b \left\{ q_\epsilon \left(\psi_\nr + \psi_\fr \right) \right\}(k)=0, \quad b \geq 0,k\in(-1/2,1/2]. 
\end{align}

 Multiplying~\eqref{eq:Qn0-far-rewrite} by $u_b(x;k)=p_b(x;k)e^{2\pi ikx}$, summing over $b\geq0$ and integrating with respect to $k\in(-1/2, 1/2]$ yields (by~\eqref{completeness})
\begin{equation}\label{eq:Qn0-far-K}
\left( Id + \mathcal{K}_\epsilon \right) \psi_\fr(x) = - \left( \mathcal{K}_\epsilon \psi_\nr \right) (x), 
\end{equation}
where we define 
\begin{equation}
\left( \mathcal{K}_\epsilon g \right) (x) \equiv \int_{-1/2}^{1/2} \sum_{b\geq0} \frac{\chi\left(|k|\geq \epsilon^r \delta_{b_*,b}\right)}{E_{b}(k) - E} \mathcal{T}_b \left\{ q_\epsilon g \right\}(k)p_b(x;k) e^{2\pi ikx} \ dk.
\end{equation}
Thus we need to solve equation~\eqref{eq:Qn0-far-K}. As in Proposition~\ref{prop:far-in-terms-near}, it is not clear that $\left( Id + \mathcal{K}_\epsilon \right)$ is invertible. However, by bound~\eqref{bnd:Qn0-KKg-Hs-to-Hs} (with $s=0$) of Lemma~\ref{lem:Qn0-KK-bounds}, stated and proved just below, one has that for $0<r<1/2$, one can chose $\epsilon$ small enough so that the operator norm $\big\Vert\mathcal{K}_\epsilon \circ \mathcal{K}_\epsilon\big\Vert_{L^2\to L^2}\leq 1/2$. Therefore, $\left( Id - \mathcal{K}_\epsilon \circ \mathcal{K}_\epsilon \right)$ (as an operator from $L^2$ to $L^2$) is invertible. 

The solution to~\eqref{eq:Qn0-far-K} is therefore uniquely defined as
\begin{equation}\label{eq:Qn0-far-KK}
\psi_\fr(x) = - \left( Id - \mathcal{K}_\epsilon \circ \mathcal{K}_\epsilon \right)^{-1} \left( Id - \mathcal{K}_\epsilon \right) \left( \mathcal{K}_\epsilon \psi_\nr \right) (x).
\end{equation}
Indeed, it is clear that, if it exists, $\psi_\fr$ satisfying~\eqref{eq:Qn0-far-K} is uniquely defined by~\eqref{eq:Qn0-far-KK} (after multiplying the equation by $\left( Id - \mathcal{K}_\epsilon \right)$). Conversely, when multiplying~\eqref{eq:Qn0-far-KK} by $\left( Id + \mathcal{K}_\epsilon \right)$, and since $\left( Id + \mathcal{K}_\epsilon \right)$ and $\left( Id - \mathcal{K}_\epsilon \circ \mathcal{K}_\epsilon \right)^{-1}$ commute, then $\psi_\fr$, as defined by~\eqref{eq:Qn0-far-KK}, solves~\eqref{eq:Qn0-far-K}. 

Thus $\psi_\fr$ is uniquely defined from $\psi_\nr\in L^2$, and therefore from $\wt{\psi}_\nr(k)\in L^2_\nr$ (when $E$ and $\epsilon$ sufficiently small are given). $\wt{\psi}_{\fr,b}=\mathcal{T}_{b}\{\psi_{\fr}\}$ is then easily seen to satisfy~\eqref{eq:Qn0-far}, by~\eqref{TbK}.
This concludes the first part of the proposition.

We now turn to estimates~\eqref{bnd:Qn0-far-in-terms-near}-\eqref{bnd:Qn0-far-in-terms-near-Hs}. By bound~\eqref{bnd:Qn0-KKg-Hs-to-Hs} of Lemma~\ref{lem:Qn0-KK-bounds}, for any $s\in\{0\}\cup (\frac12, \frac32)$, one can choose $0<\epsilon<\epsilon_0$ small enough so that $\big\Vert \mathcal{K}_\epsilon \circ \mathcal{K}_\epsilon \big\Vert_{H^s \rightarrow H^s} \leq 1/2$, and therefore
\begin{equation}\label{KK0}
\big\Vert \left( Id - \mathcal{K}_\epsilon \circ \mathcal{K}_\epsilon \right)^{-1} \big\Vert_{H^s \rightarrow H^s} \leq 2. 
\end{equation}
Moreover, by estimates~\eqref{bnd:Qn0-Kg-tight} and~\eqref{bnd:Qn0-KKg-general} of Lemma~\ref{lem:Qn0-KK-bounds}, one has for any $0\leq s<\frac32$,
\begin{align}
\big\Vert \mathcal{K}_\epsilon \psi_\nr \big\Vert_{H^s}^2 & \leq \mathcal{C} \ \left( \epsilon^{4-2s} \big\Vert \psi_\nr \big\Vert_{H^0}^2 + \epsilon^{4-4r} \big\Vert \psi_\nr \big\Vert_{H^2}^2 \right)  \label{KK1},\\
\big\Vert (\mathcal{K}_\epsilon \circ \mathcal{K}_\epsilon) \psi_\nr \big\Vert_{H^s}^2 & \leq 
\mathcal{C} \ \left( \epsilon^{4-4r} \big\Vert \psi_\nr \big\Vert_{H^0}^2 +  \epsilon^{8-2s-4r} \big\Vert \psi_\nr \big\Vert_{H^2}^2 \right). \label{KK2} 
\end{align} 
Finally, we remark that by definition~\eqref{eq:psi_localized} and Proposition~\ref{prop:norm_equivalence}, $\wt{\psi}_\nr\in L^2_\nr$ implies
\begin{equation}\label{psinear-H2}
\forall s\geq0, \quad \big\Vert\psi_\nr\big\Vert_{H^{s}}^2 \approx \int_{-1/2}^{1/2} (1+|b_*|^2)^{s} |\wt{\psi}_\nr(k)|^2\ dk \approx \big\Vert \wt{\psi}_\nr \big\Vert_{L^2}^2 .
\end{equation}
It is now straightforward to obtain~\eqref{bnd:Qn0-far-in-terms-near}-\eqref{bnd:Qn0-far-in-terms-near-Hs}, applying the estimates~\eqref{KK0}--\eqref{KK2} and~\eqref{psinear-H2} to~\eqref{eq:Qn0-far-KK}. This completes the proof of Proposition~\ref{prop:Qn0-far-intermsof-near}.
\end{proof}

To complete this argument we now prove:
\begin{Lemma}\label{lem:Qn0-KK-bounds}
Let $Q,q_\epsilon,r,\epsilon$ and $E$ be as in Proposition~\ref{prop:Qn0-far-intermsof-near}. Then, for $0 \leq s < \frac32$, the operator $\mathcal{K}_\epsilon  : H^s(\RR)  \rightarrow  H^s(\RR)$, defined by 
\begin{equation*}
\left( \mathcal{K}_\epsilon g \right) (x) \equiv \int_{-1/2}^{1/2} \sum_{b\geq0} \frac{\chi\left(|k|\geq \epsilon^r \delta_{b_*,b}\right)}{E_{b}(k) - E} \mathcal{T}_b \left\{ q_\epsilon g \right\}(k)p_b(x;k) e^{2\pi ikx} \ dk
\end{equation*}
satisfies the bounds,
\begin{align}
\big\Vert \mathcal{K}_\epsilon g \big\Vert_{H^s}^2 & \leq \mathcal{C} \ \left( \epsilon^{4-2s} \big\Vert g \big\Vert_{H^0}^2 + \epsilon^{4-4r} \big\Vert g \big\Vert_{H^2}^2 \right) \quad& 0\leq s<\frac32 \label{bnd:Qn0-Kg-tight} \\
\big\Vert (\mathcal{K}_\epsilon \circ \mathcal{K}_\epsilon) g \big\Vert_{H^s}^2 & \leq \mathcal{C} \ \epsilon^{4-8r}\big\Vert g \big\Vert_{H^s}^2  \quad &s=0\text{ or }\frac12<s<\frac32  \label{bnd:Qn0-KKg-Hs-to-Hs}\\
\big\Vert (\mathcal{K}_\epsilon \circ \mathcal{K}_\epsilon) g \big\Vert_{H^s}^2 & \leq \mathcal{C} \ \left(  \epsilon^{4-4r} \big\Vert g \big\Vert_{H^0}^2 +\epsilon^{8-2s-4r} \big\Vert g \big\Vert_{H^2}^2 \right)  \quad& 0\leq s<\frac32 \label{bnd:Qn0-KKg-general} 
\end{align}
with $\mathcal{C}= C(\mathcal{C}_0,\mathcal{C}_{\beta},b_*)$ a constant. 
\end{Lemma}
\begin{proof}
In this proof, we will make repeated use of Proposition~\ref{prop:norm_equivalence}. We first note that we can write, by~\eqref{def:Tb} and since $\{p_b(x;k)\}_{b\geq 0}$ is orthonormal in $L^2([0,1])$ for each fixed $k\in (-1/2,1/2]$,
\begin{equation}\label{TbK}
\mathcal{T}_b \left( \mathcal{K}_\epsilon g \right) (k) = \frac{\chi\left(|k|\geq \epsilon^r \delta_{b_*,b}\right)}{E_{b}(k) - E} \mathcal{T}_b \left\{ q_\epsilon g \right\}(k).
\end{equation}
Therefore, 
\begin{align*}
\big\Vert \wt{\mathcal{K}_\epsilon g} \big\Vert_{\mathcal{X}^s}^2 & = \int_{-1/2}^{1/2} \sum_{b\geq0} (1+b^2)^s \left| \mathcal{T}_b \left\{ \mathcal{K}_\epsilon g \right\}(k) \right|^2 \ dk \\
& = \int_{-1/2}^{1/2} \sum_{b\geq0} (1+b^2)^s  \frac{\chi(|k|\geq \epsilon^r \delta_{b_*,b})}{|E_b(k) - E|^2} \left| \mathcal{T}_b \left\{ q_\epsilon g \right\}(k) \right|^2 \ dk.
\end{align*}
Using bounds~\eqref{eq:E_b*-E*} and~\eqref{eq:E_b-E*}, as well as Weyl's asymptotics (Lemma~\ref{lem:Weyl}), we have
\begin{equation}
\big\Vert \wt{\mathcal{K}_\epsilon g} \big\Vert_{\mathcal{X}^s}^2 \lesssim \sum_{b\geq 0} \int_{-1/2}^{1/2} \frac1{(1+b^2)^{2-s}}\left| \mathcal{T}_b \left\{ q_\epsilon g \right\}(k) \right|^2 \ dk + \epsilon^{-4r}  \int_{-1/2}^{1/2} \left| \mathcal{T}_{b_*} \left\{ q_\epsilon g \right\}(k) \right|^2 \ dk, \label{bnd:Kg-Hs} 
\end{equation}
which will be used several times in the proof.

First we estimate the right hand side of~\eqref{bnd:Kg-Hs} as
\begin{equation}\label{bnd:Kg-Hs-1}
\big\Vert \wt{\mathcal{K}_\epsilon g} \big\Vert_{\mathcal{X}^s}^2 \lesssim \epsilon^{-4r} \big\Vert \wt{q_\epsilon g} \big\Vert_{\mathcal{X}^{-(2-s)}}^2,
\end{equation}
and use~\eqref{eq:qpsi-bnd-genA} in Lemma~\ref{lem:qpsi-bound} with $\delta=2-s>\frac12$ to obtain
\begin{equation}\label{bnd:Kg-Hs-to-Hs}
\big\Vert \mathcal{K}_\epsilon g \big\Vert_{H^s} \leq \mathcal{C} \ \epsilon^{2-s-2r} \big\Vert g \big\Vert_{H^{2-s}}, \quad 0 \leq s < \frac32.
\end{equation}
In order to obtain~\eqref{bnd:Qn0-KKg-Hs-to-Hs}, we iterate~\eqref{bnd:Kg-Hs-to-Hs}, and deduce
\begin{equation}\label{bnd:KKg-Hs-to-Hs}
\big\Vert \mathcal{K}_\epsilon \circ \mathcal{K}_\epsilon  g \big\Vert_{H^s} \leq \mathcal{C} \ \epsilon^{2-s-2r} \big\Vert \mathcal{K}_\epsilon g \big\Vert_{H^{2-s}}\leq \mathcal{C} \ \epsilon^{2-4r} \big\Vert  g \big\Vert_{H^{s}}, \quad \frac12 < s < \frac32.
\end{equation}
As for the case $s=0$, we use first~\eqref{bnd:Kg-Hs-to-Hs}, then~\eqref{bnd:Kg-Hs-1} 
\begin{equation}\label{bnd:KKg-Hs-crude}
\big\Vert \mathcal{K}_\epsilon \circ \mathcal{K}_\epsilon  g \big\Vert_{L^2} \leq \mathcal{C} \ \epsilon^{2-2r} \big\Vert \mathcal{K}_\epsilon g \big\Vert_{H^{2}}\leq \mathcal{C} \ \epsilon^{2-4r} \big\Vert \wt{q_\epsilon g} \big\Vert_{\mathcal{X}^{0}}\leq \mathcal{C} \ \epsilon^{2-4r} \big\Vert g\big\Vert_{L^2},
\end{equation}
since $ \big\Vert \wt{q_\epsilon g} \big\Vert_{\mathcal{X}^{0}}\approx \big\Vert q_\epsilon g\big\Vert_{L^2}\leq \big\Vert q_\epsilon\big\Vert_{L^\infty}\big\Vert g\big\Vert_{L^2}$.

Let us now turn to estimates~\eqref{bnd:Qn0-Kg-tight} and~\eqref{bnd:Qn0-KKg-general}. First we estimate the right hand side of~\eqref{bnd:Kg-Hs} as
\begin{equation}\label{bnd:Kg-Hs-2}
\big\Vert \wt{\mathcal{K}_\epsilon g} \big\Vert_{\mathcal{X}^s}^2 \lesssim  \big\Vert \wt{q_\epsilon g} \big\Vert_{\mathcal{X}^{-(2-s)}}^2+\epsilon^{-4r} \big\Vert \wt{q_\epsilon g} \big\Vert_{\mathcal{X}^{-2}}^2 .
\end{equation}
Using~\eqref{eq:qpsi-bnd-genB} in Lemma~\ref{lem:qpsi-bound} on both terms of the right-hand side with (respectively) $\delta=2-s>\frac12$ and $\delta=2$ yields
\begin{equation}\label{bnd:Kg-Hs-middle}
\big\Vert \mathcal{K}_\epsilon g \big\Vert_{H^s}^2 \leq \mathcal{C} \ \left(\epsilon^{2(2-s)}\big\Vert g \big\Vert_{L^2}^2 + \epsilon^{4-4r} \big\Vert g \big\Vert_{H^2}^2\right) , \quad 0 \leq s < \frac32.
\end{equation}
This completes the proof of bound~\eqref{bnd:Qn0-Kg-tight}. Note that in order to estimate~\eqref{bnd:Kg-Hs-2} when $s=2$, we can use $ \big\Vert \wt{q_\epsilon g} \big\Vert_{\mathcal{X}^{0}}\approx \big\Vert q_\epsilon g\big\Vert_{L^2}\leq \big\Vert q_\epsilon\big\Vert_{L^\infty}\big\Vert g\big\Vert_{L^2}$ and~\eqref{eq:qpsi-bnd-genA} in Lemma~\ref{lem:qpsi-bound} with $\delta=2$ to deduce
\begin{equation}\label{bnd:Kg-H2-middle}
\big\Vert \mathcal{K}_\epsilon g \big\Vert_{H^2}^2 \leq \mathcal{C} \ \left( \big\Vert g \big\Vert_{L^2}^2 + \epsilon^{4-4r} \big\Vert g \big\Vert_{H^2}^2\right) .
\end{equation}

We now turn to the operator $\mathcal{K}_\epsilon \circ \mathcal{K}_\epsilon$. We apply the now proven bound~\eqref{bnd:Qn0-Kg-tight} to get
\begin{equation}\label{bnd:KKg-intermediate}
\big\Vert (\mathcal{K}_\epsilon \circ \mathcal{K}_\epsilon) g \big\Vert_{H^s}^2 \leq \mathcal{C} \ \left( \epsilon^{2(2-s)} \big\Vert \mathcal{K}_\epsilon g \big\Vert_{H^0}^2 + \epsilon^{4-4r} \big\Vert \mathcal{K}_\epsilon g \big\Vert_{H^2}^2 \right) , \quad 0 \leq s < \frac32.
\end{equation}
Using again~\eqref{bnd:Kg-Hs-middle} with $s=0$ to bound $\big\Vert \mathcal{K}_\epsilon g \big\Vert_{H^0}$, and~\eqref{bnd:Kg-H2-middle} to bound $\big\Vert \mathcal{K}_\epsilon g \big\Vert_{H^2}$, we conclude
\begin{equation}\label{bnd:KKg-Hs-gen}
\big\Vert (\mathcal{K}_\epsilon \circ \mathcal{K}_\epsilon) g \big\Vert_{H^s}^2 \leq \mathcal{C} \ \left( \big( \epsilon^{8-2s} + \epsilon^{4-4r} \big) \big\Vert g \big\Vert_{H^0}^2 +  \epsilon^{8-2s-4r}  \big\Vert g \big\Vert_{H^2}^2 \right).
\end{equation}
This proves bound~\eqref{bnd:Qn0-KKg-general}, and completes the proof of Lemma~\ref{lem:Qn0-KK-bounds}.
\end{proof}

\subsection{Analysis of the near energy component}\label{subsec:Qnot0-near}

In this section we study the near equation~\eqref{eq:Qn0-near}, that we recall:
\[\left( E_{b_*}(k) - E \right) \widetilde{\psi}_\nr(k) +  \chi\left(|k|<\epsilon^r\right)\left(\mathcal{T}_{b_*} \left\{q_\epsilon \psi_\nr\right\}(k) + \mathcal{T}_{b_*}\left\{q_\epsilon \psi_\fr\right\}(k) \right) = 0,\]
 with the aim of extracting its leading order expression. We also make the following ansatz: 
 \begin{equation}\label{eq:def-theta-2}
 E = E_\star - \epsilon^4 \theta^2, \quad 0 < t_- \leq \theta^2 \leq t_+ < \infty, \quad E_\star \equiv E_{b_*}(0),
 \end{equation}
 where $t_-$ and $t_+$ are independent of $\epsilon$. Recall also that, by definition, 
 \begin{equation}\label{eq:Qn0-orig-rewrite-1}
 \psi_\nr+\psi_\fr=\psi \quad \text{ with } \quad \forall b\in\NN, \quad \left(E_b(k) - E \right) \mathcal{T}_b\left\{\psi \right\}(k) +  \mathcal{T}_b\left\{q_\epsilon \psi \right\}(k) = 0.
 \end{equation}
Therefore iterating~\eqref{eq:Qn0-orig-rewrite-1} using~\eqref{eq:bloch-convolution} in Proposition~\ref{prop:bloch-convolution}, we can write
\begin{align*}
\mathcal{T}_{b_*}\{q_\epsilon \psi\}(k) & = \int_0^1 \overline{p_{b_*}(x;k)} \wt{q_\epsilon \psi}(x;k) \ dx \\
& = \int_0^1 \overline{p_{b_*}(x;k)}  \int_{-1/2}^{1/2} \wt{q_\epsilon}(x;k-l) \wt{\psi}(x;l) \ dl \ dx \\
& = \int_0^1 \overline{p_{b_*}(x;k)}  \int_{-1/2}^{1/2} \wt{q_\epsilon}(x;k-l) \sum_{a\geq 0} p_{a}(x;l) \mathcal{T}_a\{\psi\}(l) \ dl \ dx\\
& = -\int_0^1 \overline{p_{b_*}(x;k)}  \int_{-1/2}^{1/2} \wt{q_\epsilon}(x;k-l) \sum_{a\geq 0} p_{a}(x;l) \frac1{E_a(l)-E_\star + \epsilon^4 \theta^2} \mathcal{T}_a\{q_\epsilon \psi\}(l)\ dl \ dx.
\end{align*}
We then use Fubini's theorem
and rewrite the near equation~\eqref{eq:Qn0-near} as 
\begin{equation}\label{eq:near-I}
(E_{b_*}(k) - E_\star + \epsilon^4 \theta^2) \wt{\psi}_\nr(k) + \chi(|k|<\epsilon^r) \left( \mathfrak{I}[\theta]\psi_\nr  \ + \  \mathfrak{I}[\theta]\psi_\fr \right)(k) \ = \ 0,
\end{equation}
with the notation
\begin{equation}\label{def-mfI}
 \big(\mathfrak{I}[\theta]\psi\big)(k)  \ = \ - \int_{-1/2}^{1/2} \sum_{a\geq 0}  \mathcal{T}_a\{q_\epsilon \psi\}(l)\ \frac1{E_a(l)-E_\star + \epsilon^4 \theta^2} \ I_{b_*,a}[q_\epsilon](k;l) \ dl,
\end{equation}
where we define for $a \geq 0$, $b \geq 0$:
\begin{equation}\label{def:Iab-2}
 I_{b,a}[q_\epsilon](k;l) \ \equiv \  \int_0^1 \overline{p_{b}(x;k)} \wt{q_\epsilon}(x;k-l)  p_a(x;l)\ dx.
\end{equation} 
Note that~\eqref{eq:near-I} is the analogue of~\eqref{eq:near-4} for the case $Q \equiv 0$. 

\begin{Proposition}\label{prop:Qn0-near-rewrite1}
Let $q_\epsilon$ be such that $\widehat{q_\epsilon} \in L^1\cap L^\infty$ and assume $q_\epsilon$ is concentrated at high frequencies in the sense of~\eqref{as:Qn0-q-epsilon-decay}. Assume $Q\in W^{N,\infty}$ with $N$ sufficiently large, so that~\eqref{as:Qn0-Q-pb} applies. Then for $\epsilon$ sufficiently small we can write the near energy equation~\eqref{eq:near-I} as
\begin{multline}\label{eq:almost-eff-Qn0}
(E_{b_*}(k) - E_\star + \epsilon^4 \theta^2) \wt{\psi}_\nr(k) - \chi(|k|<\epsilon^r) \epsilon^2 B_{b_*, {\eff}} \times \int_{-1/2}^{1/2} \wt{\psi}_\nr(s) \ ds\\
= \chi(|k|<\epsilon^r)\big( R[\theta]\wt{\psi}_\nr\big)(k),
\end{multline}
where $B_{b_*, {\eff}}$ is as defined in Hypothesis (H2'), equation~\eqref{as:Qn0-q-epsilon-effective}, and $R[\theta]:L^2_\nr\to L^\infty$ (recall definition~\eqref{eq:psi_localized}) is a linear mapping satisfying the bound 
\begin{equation}\label{eq:Qn0-remainder-bound1}
 \big\Vert R[\theta]\wt{\psi}_\nr \big\Vert_{L^\infty} \leq \mathcal{C} \epsilon^{2+\sigma_\nr} \big\Vert \wt{\psi}_\nr \big\Vert_{L^1} + \mathcal{C} \epsilon^{3 + \sigma_\fr} \big\Vert \wt{\psi}_\nr \big\Vert_{L^2}.
\end{equation}
Here, $\sigma_{\nr} = \min\{1,r, \sigma_\eff \}$, with $\sigma_\eff$ defined in~\eqref{as:Qn0-q-epsilon-effective}, $\sigma_\fr = 1/2-2r$, and $\mathcal{C}=C(C_6,\mathcal{C}_6,\mathcal{C}_0,b_*)$ is a constant.
\end{Proposition}
\begin{proof}
Adding and subtracting the anticipated dominant contribution to $\mathfrak{I}[\theta]\psi_\nr$, we may rewrite the near-energy equation~\eqref{eq:near-I} as 
\begin{align}
& (E_{b_*}(k) -  E_\star + \epsilon^4 \theta^2) \wt{\psi}_\nr(k) - \chi(|k|<\epsilon^r) \epsilon^2  B_{b_*,\eff} \times \int_{-1/2}^{1/2} \wt{\psi}_\nr (s) \ ds  \nn \\
& \quad \quad  =- \chi(|k|<\epsilon^r) \left[ \epsilon^2  B_{b_*,\eff} \times \int_{-1/2}^{1/2} \wt{\psi}_\nr(s) \ ds + \big(\mathfrak{I}[\theta]\psi_\nr\big)(k) \right] - \chi(|k|<\epsilon^r) \big(\mathfrak{I}[\theta]\psi_\fr\big)(k) \nn \\
& \quad \quad \equiv \chi(|k|<\epsilon^r)\big( R_1[\theta]\wt{\psi}_\nr\big)(k) + \chi(|k|<\epsilon^r) \big(R_2[\theta]\wt{\psi}_\nr\big)(k) \equiv \chi(|k|<\epsilon^r) \big(R[\theta]\wt{\psi}_\nr\big)(k). \label{eq:rem-def-near-rewrite}
\end{align}
The proof of the bound~\eqref{eq:Qn0-remainder-bound1} follows from Lemmata~\ref{lem:Qn0-I-far-bound},~\ref{lem:Qn0-near-diff-bound}, and~\ref{lem:LOT-expanded} below, and triangular inequality.
\end{proof}

\begin{Lemma}\label{lem:Qn0-I-far-bound}
Under the assumptions of Proposition~\ref{prop:Qn0-near-rewrite1}, there exists $\epsilon_0>0$ such that for $\epsilon\in(0,\epsilon_0)$, one has
\[ \big\Vert  \chi(|k|<\epsilon^r) \big(R_2[\theta]\wt{\psi}_\nr]\big)(k) \big\Vert_{L^\infty(\RR_k)} = \big\Vert \chi(|k|<\epsilon^r)  \big(\mathfrak{I}[\theta]\psi_\fr \big)(k)\big\Vert_{L^\infty(\RR_k)} \ \leq \ \mathcal{C} \epsilon^{2+3/2-2r} \big\Vert \wt{\psi}_\nr\big\Vert_{L^2} ,\]
with $\mathcal{C}=C(\mathcal{C}_{6},C_{6}, \mathcal{C}_0 ,b_*)$.
\end{Lemma}

\begin{Lemma}\label{lem:Qn0-near-diff-bound}
Under the assumptions of Proposition~\ref{prop:Qn0-near-rewrite1}, one has
\begin{multline}
\sup_{|k|<\epsilon^r} \left| \big(\mathfrak{I}[\theta]\psi_\nr \big)(k) + \int_{-1/2}^{1/2} ds \ \wt{\psi}_\nr(s) \int_{-1/2}^{1/2} dl \sum_{a \geq 0} I_{b_*,a}[Q_\epsilon](k;l) I_{a,b_*}[q_\epsilon](l;s)  \right| \\ \lesssim  \mathcal{C} \ \epsilon^3 \big\Vert \wt{\psi}_\nr \big\Vert_{L^1}, \label{est:near-diff-bnd}
\end{multline}
where 
\begin{equation}\label{eq:Qeps-def1}
{Q_\epsilon}(\xi)\equiv \frac{\widehat{q_\epsilon}(\xi)}{1+4\pi^2 |\xi|^2 }.
\end{equation}
Here $\mathcal{C}=C(C_4,\mathcal{C}_4,\mathcal{C}_0,b_*)$ is a constant.
\end{Lemma}

\begin{Lemma}\label{lem:LOT-expanded}
Under the assumptions of Proposition~\ref{prop:Qn0-near-rewrite1}, one has
\begin{multline}
\sup_{|k|<\epsilon^r,|s|<\epsilon^r}\left| \int_{-1/2}^{1/2} dl \sum_{a \geq 0} I_{b_*,a}[Q_\epsilon](k;l) I_{a,b_*}[q_\epsilon](l;s) - \epsilon^2 B_{b_*, {\eff}} \right| \\
\leq C( \mathcal{C}_2, \mathcal{C}_0, b_*) \ (\epsilon^{2+r} + \epsilon^{2 + \sigma_{\eff}}).
\end{multline}
Here, $B_{b_*, \eff}$ is defined in~\eqref{as:Qn0-q-epsilon-effective} and $C( \mathcal{C}_2, \mathcal{C}_0, b_*)$ is a constant.
\end{Lemma}

The proofs of Lemmata~\ref{lem:Qn0-I-far-bound},~\ref{lem:Qn0-near-diff-bound} and~\ref{lem:LOT-expanded} appear at the end of this section.
\medskip

\noindent{\bf Rescaling the equation.} The next step consists in rescaling the equation so as to balance terms on the left hand side of~\eqref{eq:almost-eff-Qn0}. We therefore define
\begin{equation}\label{eq:rescaling}
k=\epsilon^2 \kappa, \quad \wt{\psi}_\nr(k) = \frac{1}{\epsilon^2} \widehat{\Phi}\left(\frac{k}{\epsilon^2}\right) = \frac{1}{\epsilon^2} \chi(|\kappa|<\epsilon^{r-2}) \widehat{\Phi}(\kappa).
\end{equation}
% %
% %
Note also that one has the following estimates:
\begin{equation}\label{eq:psi-to-Phi}
\big\Vert \wt{\psi}_\nr \big\Vert_{L^1} \lesssim \big\Vert \widehat{\Phi} \big\Vert_{L^{2,1}}\qquad \text{ and } \qquad \big\Vert \wt{\psi}_\nr \big\Vert_{L^2} \lesssim \epsilon^{-1} \big\Vert \widehat{\Phi} \big\Vert_{L^{2,1}}.
\end{equation}
These follow from the definition $\big\Vert \widehat{\Phi} \big\Vert_{L^{2,1}}^2 \equiv \int_{-\infty}^\infty (1+|\kappa|^2)| \widehat{\Phi}(\kappa)|^2\ d\kappa$, and the following bounds:
\begin{align*}
&\big\Vert \wt{\psi}_\nr \big\Vert_{L^1} = \int_{-\infty}^{\infty} \left| \chi(|k|<\epsilon^r) \wt{\psi}_\nr(k) \right| \ dk = \int_{-\infty}^{\infty} \left| \chi(|\kappa|<\epsilon^{r-2}) \widehat{\Phi}(\kappa) \right| \ d\kappa \lesssim \big\Vert \widehat{\Phi} \big\Vert_{L^{2,1}}\\
& \big\Vert \wt{\psi}_\nr \big\Vert_{L^2}^2 = \int_{-\infty}^{\infty} \left| \chi(|k|<\epsilon^r) \wt{\psi}_\nr(k) \right|^2 \ dk = \epsilon^{-2} \big\Vert \widehat{\Phi} \big\Vert_{L^2}^2\leq \epsilon^{-2} \big\Vert \widehat{\Phi} \big\Vert_{L^{2,1}}^2.
\end{align*}

The next proposition extracts the leading order terms in~\eqref{eq:almost-eff-Qn0}, in terms of the variable $\kappa$ and unknown $\widehat{\Phi}$.

\begin{Proposition}\label{prop:Qn0-leading-order}
Assume the hypotheses of Proposition~\ref{prop:Qn0-near-rewrite1} hold. Then, the rescaled near-energy component solves the equation: 
\begin{equation}
\left( \frac12\partial_k^2 E_{b_*}(0) \kappa^2 + \theta^2 \right) \widehat{\Phi}(\kappa) - \chi(|\kappa|<\epsilon^{r-2}) B_{b_*,\eff} \times 
\int_{-\infty}^\infty \widehat{\Phi}(\xi) \ d\xi =  \big(R_\flat [\theta]\widehat{\Phi}\big) ( \kappa), \label{eq:Qn0-leading-order} 
\end{equation}where  $R_\flat [\theta]:L^{2,1}\to L^{2,-1}$ is a linear mapping satisfying the estimate
\begin{equation}\label{eq:Qn0-leading-order-bnd}
 \big\Vert R_\flat [\theta]\widehat{\Phi} \big\Vert_{L^{2,-1}} \leq \mathcal{C} \left( \epsilon^{2r} + \epsilon^{\sigma_\nr} + \epsilon^{\sigma_\fr } \right) \big\Vert \widehat{\Phi} \big\Vert_{L^{2,1}}.
\end{equation}
Here  $\mathcal{C}=C(\mathcal{C}_{6},C_{6},\mathcal{C}_0,b_*, \sup_{|k'|<\epsilon^r} | E_{b_*}^{(4)}(k') |)$ is a constant, and we recall that $\sigma_{\nr} = \min\{1/2,r, \sigma_\eff \}$ and $\sigma_\fr = 1/2-2r$ (see Proposition~\ref{prop:Qn0-near-rewrite1}).
\end{Proposition}
\begin{proof}
Substituting the rescalings~\eqref{eq:rescaling} into~\eqref{eq:almost-eff-Qn0} and dividing by $\epsilon^2$ yields:
\begin{multline}\label{eq:almost-eff-Qn0-rescaled}
\epsilon^{-4}(E_{b_*}(\epsilon^2 \kappa) -  E_\star + \epsilon^4 \theta^2) \widehat{\Phi}(\kappa)- \chi(|\kappa|<\epsilon^{r-2}) B_{b_*, \eff} \times \int_{-\infty}^{\infty}\widehat{\Phi}(\xi) d\xi\\
= \chi(|\kappa|<\epsilon^{r-2}) \epsilon^{-2} \big(R[\theta]\wt{\psi}_\nr\big)(\epsilon^2 \kappa).
\end{multline}
The estimate on $R[\theta]\wt{\psi}_\nr$ in~\eqref{eq:Qn0-remainder-bound1},
together with~\eqref{eq:psi-to-Phi}, yields immediately
\begin{align}
\left\Vert  \epsilon^{-2} \big(R[\theta]\wt{\psi}_\nr\big)(\epsilon^2 \kappa) \right\Vert_{L^{2,-1}(\RR_\kappa)} &\leq C \big\Vert \epsilon^{-2} R[\theta]\wt{\psi}_\nr \big\Vert_{L^\infty} \\
&\leq \mathcal{C} \epsilon^{\sigma_{\nr}} \big\Vert \wt{\psi}_\nr \big\Vert_{L^1} + \mathcal{C} \epsilon^{ 1+\sigma_\fr} \big\Vert \psi_\nr \big\Vert_{L^2} \nn\\
&\leq \mathcal{C}\left(\epsilon^{\sigma_{\nr}} + \epsilon^{\sigma_\fr } \right) \big\Vert \widehat{\Phi}\big\Vert_{L^{2,1}}.\label{eq:R-I}
\end{align}

There remains to expand $\epsilon^{-4} ( E_{b_*}(\epsilon^2 \kappa) -  E_\star + \epsilon^4 \theta^2 )$. Since $ E_{b_*}^{(j)}(0)=0$ for $j = 1, \ 3$ (by Lemma~\ref{lem:band-edge}), Taylor expansion of $E_{b_*}(\epsilon^2 \kappa)$ about $\kappa=0$ to fourth order yields
\begin{equation}\label{eq:E_b*-E}
E_{b_*}(\epsilon^2 \kappa) -  E_\star + \epsilon^4 \theta^2 = \frac{\epsilon^4\kappa^2}{2} E_{b_*}^{(2)}(0) + \epsilon^4\theta^2 + \frac{\epsilon^8\kappa^4}{24} E_{b_*}^{(4)}(k') ,
\end{equation}
where $k'$ is such that $|k'|<|\epsilon^2\kappa|\leq \epsilon^{r}$. 
Therefore, provided $\sup_{|k'|<\epsilon^r} \left| E_{b_*}^{(4)}(k') \right| < \infty$, one has
\begin{multline}
\left\Vert  \left(\epsilon^{-4}(E_{b_*}(\epsilon^2 \kappa) -  E_\star + \epsilon^4 \theta^2) -\left( \frac{\kappa^2}{2} E_{b_*}^{(2)}(0) + \theta^2 \right) \right) \widehat{\Phi}(\kappa) \right\Vert_{L^{2,-1}(\RR_\kappa)} \\ \lesssim \epsilon^{4}   \big\Vert \widehat{\Phi}\big\Vert_{L^{2,1}} \left( \sup_{\kappa\in [-\epsilon^{r-2},\epsilon^{r-2}]} \frac{\kappa^8}{(1+|\kappa|^2)^2}\right)^{1/2} 
\lesssim \epsilon^{2r}  \big\Vert \widehat{\Phi}\big\Vert_{L^{2,1}}. \label{eq:R-III}
\end{multline}

Plugging estimates~\eqref{eq:R-I} and~\eqref{eq:R-III} into~\eqref{eq:almost-eff-Qn0-rescaled} immediately yields~\eqref{eq:Qn0-leading-order} with bound~\eqref{eq:Qn0-leading-order-bnd} on $R_\flat [\theta]\widehat{\Phi}$. This completes the proof of Proposition~\ref{prop:Qn0-leading-order}.
\end{proof}

\medskip

We now give the proofs of Lemmata~\ref{lem:Qn0-I-far-bound},~\ref{lem:Qn0-near-diff-bound} and~\ref{lem:LOT-expanded}. 

\begin{proof}[Proof of Lemma~\ref{lem:Qn0-I-far-bound}]
Using Cauchy-Schwarz inequality in~\eqref{def-mfI}, one has
\begin{equation}\label{eq:cs-I-far}
\big |\mathfrak{I}[\theta]\psi_\fr  \big|(k) \leq \left( \int_{-1/2}^{1/2} \sum_{a\geq0}  \left|\frac{I_{b_*,a}[q_\epsilon](k;l)}{E_a(l) - E_\star + \epsilon^4 \theta^2}   \right|^2 dl \right)^{1/2} 
\times
\left( \int_{-1/2}^{1/2} \sum_{a\geq0}\big| \mathcal{T}_a\{q_\epsilon\ \psi_\fr \}(l)\big|^2 \ dl \right)^{1/2}.
\end{equation}
The second factor of~\eqref{eq:cs-I-far} is estimated as follows, using Proposition~\ref{prop:Qn0-far-intermsof-near} and Hypothesis (H1'a), estimate~\eqref{as:Qn0-q-epsilon-bound},
\[ \left( \int_{-1/2}^{1/2} \sum_{a\geq0}\big| \mathcal{T}_a\{q_\epsilon\ \psi_\fr \}(l)\big|^2 \ dl \right)^{1/2}\approx \big\Vert q_\epsilon \psi_\fr\big\Vert_{L^2} \leq \big\Vert q_\epsilon\big\Vert_{L^\infty} \big\Vert  \psi_\fr\big\Vert_{L^2} \lesssim \ \mathcal{C}_0 \epsilon^{2-2r} \big\Vert \wt{\psi}_\nr\big\Vert_{L^2} .\]
As for the first term of~\eqref{eq:cs-I-far}, we treat differently the cases $a=b_*$, $a\neq b_*$ and 
$a \leq a_*^\epsilon$ and $a > a_*^\epsilon$, where $a_*^\epsilon = \max\{a\geq 0 \text{ such that } \sqrt{E_a(k)}<\pi/(3\epsilon)\}$. By Weyl's asymptotics (Lemma~\ref{lem:Weyl}), one has $\sqrt{E_a(k)}\approx a$ and therefore $a_*^\epsilon\approx 1/\epsilon$.

\medskip

{\noindent \textit{Case $a=b_*$.}} By~\eqref{eq:def-theta-2}, $|E_{b_*}(k) - E_\star + \epsilon^4 \theta^2| \geq t_- \epsilon^4$. Together with estimate~\eqref{est-I-ab-bounded} of Lemma~\ref{lem:I-bounds} with $a = b_*$, we can bound
\[
\int_{-1/2}^{1/2} \left|\frac{I_{b_*,b_*}[q_\epsilon](k;l)}{E_a(l) - E_\star + \epsilon^4 \theta^2}   \right|^2 dl \leq C \epsilon^{-8} \int_{-1/2}^{1/2} \left| I_{b_*,b_*}[q_\epsilon](k;l) \right|^2 dl \leq C (1+|b_*|^N)^2 \epsilon^{2N-8}.
\]

\medskip

{\noindent \textit{Case $0\leq a \leq a_*^\epsilon$, $a\neq b_*$.}} By Weyl's asymptotics (Lemma~\ref{lem:Weyl}), one has $\frac{1}{E_a(l) - E}\lesssim 1/(a^2+1)$ for $a\neq b_*$. Therefore, applying estimate~\eqref{est-I-ab-small} of Lemma~\ref{lem:I-bounds}, one has the bound
\begin{align*}
\sum_{0\leq a\leq a_*^\epsilon,a\neq b_*} \int_{-1/2}^{1/2} \left|\frac{I_{b_*,a}[q_\epsilon](k;l)}{E_a(l) - E_\star + \epsilon^4 \theta^2}   \right|^2 dl & \lesssim\sum_{0 \leq a\leq a_*^\epsilon,a\neq b_*} \frac{1}{(a^2 + 1)^2} \int_{-1/2}^{1/2} \left| I_{b_*,a}[q_\epsilon](k;l) \right|^2 dl \\
& \lesssim \epsilon^{3} (1 + |b_*|^2)^2.
\end{align*}

\medskip

{\noindent \textit{Case $a > a_*^\epsilon$, $a\neq b_*$.}} 
In this case, one has $\frac{1}{E_a(l) - E}\lesssim \epsilon^2$, and therefore
\[
\sum_{a> a_*^\epsilon} \int_{-1/2}^{1/2} \left|\frac{I_{b_*,a}[q_\epsilon](k;l)}{E_a(l) - E_\star + \epsilon^4 \theta^2}  \right|^2 dl  \lesssim \epsilon^{4} \sum_{a>a_*^\epsilon} \int_{-1/2}^{1/2} \left|I_{b_*,a}[q_\epsilon](k;l) \right|^2 dl \lesssim \epsilon^4,
\]
where we used estimate~\eqref{eq:I-sum-b-bnd} of Lemma~\ref{lem:I-summable}. 

\medskip

Thus 
\[
\big\Vert \chi(|k|<\epsilon^r)  \big(\mathfrak{I}[\theta]\psi_\fr \big)(k)\big\Vert_{L^\infty(\RR_k)} \lesssim \epsilon^{2-2r} \left( \epsilon^{2N - 8} + \epsilon^3 \right)^{1/2} \big\Vert \wt{\psi}_\nr \big\Vert_{L^2}.
\]
Choosing $N = 11/2$ and defining $\sigma_\fr \equiv 1/2 - 2r$, we obtain 
\[
\left|  \chi(|k|<\epsilon^r)\big( R_2[\theta]\wt{\psi}_\nr\big)(k) \right| = \left|  \chi(|k|<\epsilon^r) \big( \mathfrak{I}[\theta]\psi_\fr\big)(k) \right| \leq \mathcal{C} \epsilon^{3+\sigma_\fr} \big\Vert \wt{\psi}_\nr \big\Vert_{L^2},
\]
with $\mathcal{C}=C(\mathcal{C}_{6},C_{6},\mathcal{C}_0,b_*)$ which completes the proof of Lemma~\ref{lem:Qn0-I-far-bound}.
\end{proof}

\begin{proof}[Proof of Lemma~\ref{lem:Qn0-near-diff-bound}]
Let us first manipulate $ \mathfrak{I}[\theta]\psi_\nr $. Using~Proposition~\ref{prop:bloch-convolution} and definition $\psi_\nr(x) = \int_{-1/2}^{1/2} e^{2\pi iys} \wt{\psi}_\nr(s) p_{b_*}(x;s) ds$, one has
\begin{align}
-\big(\mathfrak{I}[\theta]\psi_\nr \big)(k) & = \int_{-1/2}^{1/2} dl \sum_{a\geq0} \mathcal{T}_a\{q_\epsilon\ \psi_\nr \}(l) \frac{1}{E_a(l) -  E_\star + \epsilon^4 \theta^2} I_{b_*,a}[q_\epsilon](k,l) \nn \\
& = \int_{-1/2}^{1/2} dl \sum_{a\geq0} \left( \int_0^1 dy \ \overline{p_a(y;l)} \int_{-1/2}^{1/2} ds \ \wt{q_\epsilon}(y;l-s) p_{b_*}(y;s) \wt{\psi}_\nr(s)  \right) \frac{I_{b_*,a}[q_\epsilon](k,l)}{E_a(l) -  E_\star + \epsilon^4 \theta^2}  \nn \\
& = \int_{-1/2}^{1/2} ds\ \wt{\psi}_\nr(s) \int_{-1/2}^{1/2} dl \sum_{a\geq 0} \frac{I_{b_*,a}[q_\epsilon](k,l)}{E_a(l) -  E_\star + \epsilon^4 \theta^2} \left( \int_0^1 dy \ \overline{p_a(y;l)} \wt{q_\epsilon}(y;l-s) p_{b_*}(y;s) \right)  \nn \\
& = \int_{-1/2}^{1/2} ds\ \wt{\psi}_\nr(s) \int_{-1/2}^{1/2} dl \sum_{a\geq 0}  \frac{I_{b_*,a}[q_\epsilon](k;l) I_{a,b_*}[q_\epsilon](l;s)}{E_a(l) -  E_\star + \epsilon^4 \theta^2}. \label{eq:l-o-t}
\end{align}
Note that~\eqref{eq:l-o-t} is the analogue of~\eqref{eq:Q-near-1} in the case $Q \equiv 0$.

Our aim is now to prove that, to leading order, as $\epsilon \rightarrow 0$:
\[ \frac{I_{b_*,a}[q_\epsilon](k;l) I_{a,b_*}[q_\epsilon](l;s)}{E_a(l) -   E_\star + \epsilon^4 \theta^2}  \approx I_{b_*,a}[Q_\epsilon](k;l) I_{a,b_*}[q_\epsilon](l;s), \quad \widehat{Q_\epsilon}(\xi)\equiv \frac{\widehat{q_\epsilon}(\xi)}{1+4\pi^2 |\xi|^2 }.\]
To this end, we proceed in a manner similar to the proof of Lemma~\ref{lem:Qn0-I-far-bound}. Decompose the sum over $a$ into the cases: $a=b_*$, $a\neq b_*$ and $a\leq a_*^\epsilon$, and $a>a_*^\epsilon$, where 
\[a_*^\epsilon \equiv\max\{a\geq 0 \text{ such that } \sqrt{E_a(k)}<\pi/(3\epsilon)\}.\]
By Weyl's asymptotics (Lemma~\ref{lem:Weyl}), one has $\sqrt{E_a(k)}\approx a$ and therefore $a_*^\epsilon\approx 1/\epsilon$.

Let us first notice that $Q_\epsilon\in L^2$ and clearly satisfies~\eqref{as:Qn0-q-epsilon-decay}. Therefore, the bounds of Lemma~\ref{lem:I-bounds} apply with $q_\epsilon$ replaced by $Q_\epsilon$.
Moreover, one has $\big\Vert Q_\epsilon\big\Vert_{H^2}\lesssim \big\Vert q_\epsilon\big\Vert_{L^2}$, thus~\eqref{est-I-ab-small-bis}, in particular, applies. 

\medskip

\noindent {\em Case $a=b_*$.} We use that $\frac{1}{|E_{b_*}(l)-  E_\star + \epsilon^4 \theta^2|} \leq t_-^{-1} \epsilon^{-4}$. By the Cauchy-Schwarz inequality, the triangle inequality, and~\eqref{est-I-ab-bounded} (noting that $I_{b_*,b_*}[q_\epsilon](l;k) = \overline{I_{b_*,b_*}[q_\epsilon](k;l)}$), one has
\begin{equation}\label{est-I-ab-main-a=b*} 
\int_{-1/2}^{1/2} \ dl  \left| \left[ \frac{I_{b_*,b_*}[q_\epsilon](k;l)}{E_{b_*}(l)-  E_\star + \epsilon^4 \theta^2}  -I_{b_*,b_*}[Q_\epsilon](k;l)\right] I_{b_*,b_*}[q_\epsilon](l;s) \right| \leq \mathcal{C} \epsilon^{2N-4},
\end{equation}
with $\mathcal{C}=C(C_{N+1/2},\mathcal{C}_{N+1/2},\big\Vert q_\epsilon\big\Vert_{L^2},b_*)$, uniformly with $k,s\in (-1/2,1/2]$.

\medskip

\noindent {\em Case $a\leq a_*^\epsilon$, $a\neq b_*$.} We now use estimate~\eqref{est-I-ab-small} for the contribution of $I_{b_*,a}[q_\epsilon](k;l) $, and estimate~\eqref{est-I-ab-small-bis} for the contribution of $I_{b_*,a}[Q_\epsilon](k;l)$:
It follows
\begin{align*} \int_{-1/2}^{1/2} \big| I_{b_*,a}[q_\epsilon](k;l) \big|^2\ dl &\leq C(C_2,\mathcal{C}_2,\big\Vert q_\epsilon\big\Vert_{L^2},b_*) \epsilon^{3} \\
\int_{-1/2}^{1/2} \big| I_{b_*,a}[Q_\epsilon](k;l)  \big|^2\ dl &\leq C(C_4,\mathcal{C}_4,\big\Vert q_\epsilon\big\Vert_{L^2},b_*)\epsilon^{7}.
\end{align*}
Similar estimates apply of course to $I_{a,b_*}[q_\epsilon](l;s) = \overline{I_{b_*,a}[q_\epsilon](s;l)}$. By Weyl's asymptotics (Lemma~\ref{lem:Weyl}), one has $\frac{1}{|E_{a}(l)-E|}\lesssim \frac1{1+|a|^2}$ and $a_*^\epsilon\approx 1/\epsilon$. Using the triangle inequality and the Cauchy-Schwarz inequality, it follows
\begin{equation}\label{est-I-ab-main-a<a*} \sum_{a\neq b_*,a\leq a_*^\epsilon} \int_{-1/2}^{1/2} \ dl  \left|\left[ \frac{I_{b_*,a}[q_\epsilon](k;l)}{E_{a}(l) -   E_\star + \epsilon^4 \theta^2}  -I_{b_*,a}[Q_\epsilon](k;l)  \right] I_{a,b_*}[q_\epsilon](l;s) \right|  \leq \mathcal{C} \epsilon^{3},
\end{equation}
with $\mathcal{C}=C(C_4,\mathcal{C}_4,\big\Vert q_\epsilon\big\Vert_{L^2},b_*)$, uniformly with $k,s\in (-1/2,1/2]$.
\medskip

\noindent {\em Case $a> a_*^\epsilon$.} Let us study in detail
\[I_{b_*,a}[Q_\epsilon](k;l)=\int_0^1 \overline{p_{b_*}(x;k)} \left(\sum_{n\in\ZZ} e^{2\pi i nx}\frac{\widehat{q_\epsilon}(k-l+n)}{1+4\pi^2 (k-l+n)^2}  \right)p_a(x;l)\ dx.\]
By Lemma~\ref{lem:asymptotics-in-b}, there exists $B^+_{a,b_*}(x;k,l)$ and $B^-_{a,b_*}(x;k,l)$ such that
\[  \overline{p_{b_*}(x;k)}p_a(x;l)e^{2\pi i (l-k)x} = B^+_{a,b_*}(x;k,l)e^{ix\sqrt{E_a(l)} }+B^-_{a,b_*}(x;k,l)e^{- ix\sqrt{E_a(l)} },\]
and $B^\pm_{a,b_*}(x;k,l)$ satisfies
\[ \big\Vert B^\pm_{a,b_*}(\cdot;k,l)\big\Vert_{W^{2,\infty}_{\rm per}}\leq C(\big\Vert Q\big\Vert_{W^{1,\infty}_{\rm per}},b_*) \quad \text{ and }\quad \big\Vert \partial_x B^\pm_{a,b_*}(\cdot;k,l)\big\Vert_{L^\infty}\leq \frac{C(\big\Vert Q\big\Vert_{W^{1,\infty}_{\rm per}},b_*)}{1+|a|},\]
uniformly with respect to $a,k,l$. 

After integrating twice by parts, one has (here and thereafter, we abuse notations and write $F_\pm$ for $F_++F_-$)
\[I_{b_*,a}[Q_\epsilon](k;l)=\frac{-1}{E_a(l)}\int_0^1\partial_x^2\left\{ B^\pm_{a,b_*}(x;k,l) \left(\sum_{n\in\ZZ} e^{2\pi i (n+k-l)x}\frac{\widehat{q_\epsilon}(k-l+n)}{1+4\pi^2 (k-l+n)^2}  \right)\right\}  e^{\pm ix\sqrt{E_a(l)} }\ dx.\]
We then make use of the identity
\begin{multline*}
-\partial_x^2\left\{ \frac{B^\pm_{a,b_*}(x;k,l)  e^{2\pi i (n+k-l)x}}{1+4\pi^2 (k-l+n)^2}\right\}\\= \left( B^\pm_{a,b_*}(x;k,l)-\frac{B^\pm_{a,b_*}+\partial_x^2 B^\pm_{a,b_*}(x;k,l)+4\pi i (n+k-l) \partial_xB^\pm_{a,b_*}(x;k,l) }{1+4\pi^2 (k-l+n)^2}\right) e^{2\pi i (n+k-l)x},
\end{multline*}
and deduce from the above estimates
\[  I_{b_*,a}[Q_\epsilon](k;l)=\frac{I_{b_*,a}[q_\epsilon](k;l)}{E_a(l)}+\frac{1}{E_a(l)}\int_0^1\sum_{n\in\ZZ} J_{n}(x;k,l)  \widehat{q_\epsilon}(k-l+n)\ dx,\]
with
\begin{equation}\label{est-Jn}
\big\Vert J_n(\cdot;k,l)\big\Vert_{L^\infty_{\rm per}}\leq C(\big\Vert Q\big\Vert_{W^{1,\infty}_{\rm per}},b_*) \times \Big(\frac{1}{1+4\pi^2 (k-l+n)^2}+\frac{1}{\sqrt{E_a(l)}(1+2\pi |k-l+n|)}\Big),
\end{equation}
uniformly with respect to $a,k,l$ and $n$.

In order to estimate the latter, we decompose the sum over $|n|<1/(3\epsilon)$, and $|n|\geq 1/(3\epsilon)$. For the former, we have, thanks to assumption~\eqref{as:Qn0-q-epsilon-decay} and the Cauchy-Schwarz inequality, 
\begin{multline*} \int_{-1/2}^{1/2}\left|\sum_{|n|<1/(3\epsilon)} \Big(\frac{1}{1+4\pi^2 (k-l+n)^2}+\frac{1}{\sqrt{E_a(l)}(1+2\pi |k-l+n|)}\Big) \widehat{q_\epsilon}(k-l+n)\right|^2\ dl \\
\lesssim \int_{-1/2}^{1/2}\sum_{|n|<1/(3\epsilon)} \left|\widehat{q_\epsilon}(k-l+n)\right|^2\lesssim \big(\mathcal{C}_\beta \epsilon^\beta\big)^2.
\end{multline*}
For the latter, one has
\begin{multline*} \int_{-1/2}^{1/2}\left|\sum_{|n|\geq 1/(3\epsilon)} \Big(\frac{1}{1+4\pi^2 (k-l+n)^2}+\frac{1}{\sqrt{E_a(l)}(1+2\pi |k-l+n|)}\Big) \widehat{q_\epsilon}(k-l+n)\right|^2\ dl \\
\lesssim \big\Vert q_\epsilon\big\Vert_{L^2} \int_{-1/2}^{1/2}\sum_{|n|\geq 1/(3\epsilon)} \left|\frac{1}{1+4\pi^2n^2 }+\frac{1}{\sqrt{E_a(l)}(1+2\pi|n| )}\right|^2\lesssim \big\Vert q_\epsilon\big\Vert_{L^2}\times \big( \epsilon^3+\frac{\epsilon}{E_a(l)}\big).
\end{multline*}

Altogether, we conclude that 
\[ \int_{-1/2}^{1/2} \left|\frac{I_{b_*,a}[q_\epsilon](k;l)}{E_a(l)-  E_\star + \epsilon^4 \theta^2}-I_{b_*,a}[Q_\epsilon](k;l) \right|^2\ dl\leq \mathcal{C} \frac1{E_a(l)^2} \Big(\mathcal{C}_\beta^2 \epsilon^{2\beta}+ \epsilon^{3}+\frac{\epsilon}{E_a(l)} +\frac1{E_a(l)^2}\Big), \]
with $\mathcal{C}=C(\mathcal{C}_\beta,\big\Vert Q\big\Vert_{W^{1,\infty}_{\rm per}},b_*)$.

Finally, summing over $a>a_*^\epsilon$ (and recalling that, by Weyl's asymptotics, $a_*^\epsilon\approx 1/\epsilon$ and $E_a(l)\approx |a|^2$), one has
\[\sum_{a> a_*^\epsilon} \int_{-1/2}^{1/2} \left| \frac{I_{b_*,a}[q_\epsilon](k;l)}{E_{a}(l) -   E_\star + \epsilon^4 \theta^2}  -I_{b_*,a}[Q_\epsilon](k;l)  \right|^2\ dl  \leq \mathcal{C} \epsilon^{6},
\]
with $\mathcal{C}=C(\mathcal{C}_2,\big\Vert Q\big\Vert_{W^{1,\infty}_{\rm per}},b_*)$, uniformly with $k\in [-1/2,1/2]$. Using Cauchy-Schwarz inequality and~\eqref{est-I-uniform} of Lemma~\ref{lem:I-bounds}, we have
\begin{equation}\label{est-I-ab-main-a>a*} \sum_{a> a_*^\epsilon} \int_{-1/2}^{1/2} \ dl  \left|\big(  \frac{I_{b_*,a}[q_\epsilon](k;l)}{E_{a}(l) -   E_\star + \epsilon^4 \theta^2}  -I_{b_*,a}[Q_\epsilon](k;l) \big) I_{a,b_*}[q_\epsilon](l;s)\right| \leq \mathcal{C} \epsilon^{3}.
\end{equation}

Bounds~\eqref{est-I-ab-main-a=b*} with $N = 7/2$,~\eqref{est-I-ab-main-a<a*}, and~\eqref{est-I-ab-main-a>a*} imply~\eqref{est:near-diff-bnd} and complete the proof of Lemma~\ref{lem:Qn0-near-diff-bound}. 
\end{proof}

\begin{proof}[Proof of Lemma~\ref{lem:LOT-expanded}]
Using the identity~\eqref{eq:Iab-Ta-connection} of Lemma~\ref{lem:I-summable}, one can write
\[
\int_{-1/2}^{1/2}dl \sum_{a\geq 0} I_{b_*,a}[Q_\epsilon](k;l)I_{a,b_*}[q_\epsilon](l;s)
=\int_{-1/2}^{1/2}dl \sum_{a\geq 0} \overline{\mathcal{T}_{a}\{ u_{b_*}(\cdot;k) Q_\epsilon(\cdot) \}}(l)\mathcal{T}_a\{ u_{b_*}(\cdot;s) q_\epsilon(\cdot) \}(l).
\]
Expanding the term $\overline{\mathcal{T}_{a}\{ u_{b_*}(\cdot;k) Q_\epsilon(\cdot) \}}$ via the definition~\eqref{def:Tb}, one has
\begin{multline*}
\int_{-1/2}^{1/2}dl \sum_{a\geq 0} \overline{\mathcal{T}_{a}\{ u_{b_*}(\cdot;k) Q_\epsilon(\cdot) \}}(l)\mathcal{T}_a\{ u_{b_*}(\cdot;s) q_\epsilon(\cdot) \}(l) \\ 
= \int_{-1/2}^{1/2}dl \sum_{a\geq 0} \int_{\RR}dx u_a(x;l) \overline{u_{b_*}(x;k)} \overline{Q_\epsilon(x)} \mathcal{T}_a\{ u_{b_*}(\cdot;s) q_\epsilon(\cdot) \}(l).
\end{multline*}
Finally, using the completeness of the Bloch functions,~\eqref{completeness}, one has
\begin{multline*}
\int_{-1/2}^{1/2}dl \sum_{a\geq 0} \int_{\RR}dx u_a(x;l) \overline{u_{b_*}(x;k)} \overline{Q_\epsilon(x)} \mathcal{T}_a\{ u_{b_*}(\cdot;s) q_\epsilon(\cdot) \}(l) \\
= \int_{-\infty}^{\infty}dx  \overline{u_{b_*}(x;k)}u_{b_*}(x;s) Q_\epsilon(x) q_\epsilon(x).
\end{multline*}
Since $q_\epsilon$ is real-valued, so is $Q_\epsilon$. Therefore, we can write
\begin{multline}
\int_{-1/2}^{1/2}dl \sum_{a\geq 0} I_{b_*,a}[Q_\epsilon](k;l)I_{a,b_*}[q_\epsilon](l;s) = \int_{-\infty}^{\infty}dx |u_{b_*}(x;0)|^2 Q_\epsilon(x) q_\epsilon(x) \\
 + \int_{-\infty}^{\infty}dx Q_\epsilon(x) q_\epsilon(x) \left[ \overline{u_{b_*}(x;k)} u_{b_*}(x;s) - \overline{u_{b_*}(x;0)} u_{b_*}(x;0) \right]. \label{eq:term-7}
\end{multline}

Recalling that $\wt{\psi}_\nr(s) =\wt{\psi}_\nr(s) \chi_{\epsilon^r}(s)$, we bound the term~\eqref{eq:term-7} by Taylor expanding about $s=0$ and $k=0$. From Lemma~\ref{lem:estimates-uniform-in-b}, one has for any $k,s\in[-\epsilon^r,\epsilon^r]$,
\[ \left| \int_{-\infty}^{\infty}dx   Q_\epsilon(x)q_\epsilon(x) (\overline{u_{b_*}(x;k)}u_{b_*}(x;s)- \overline{u_{b_*}(x;0)} u_{b_*}(x;0))\right|\leq \epsilon^r C(\big\Vert Q\big\Vert_{L^\infty_{\rm per}},b_*) \int_{-\infty}^{\infty}dx  | Q_\epsilon(x)q_\epsilon(x) |.\]
One checks using Hypothesis (H1'b),~\eqref{as:Qn0-q-epsilon-decay}, that
\[ \int_{-\infty}^{\infty}dx  | Q_\epsilon(x)q_\epsilon(x) |\leq \big\Vert q_\epsilon\big\Vert_{L^2}\big\Vert Q_\epsilon\big\Vert_{L^2}\leq \epsilon^2 C(\mathcal{C}_2)\big\Vert q_\epsilon\big\Vert_{L^2}^2,\]
to write
\[
\left| \int_{-1/2}^{1/2}dl \sum_{a\geq 0} I_{b_*,a}[Q_\epsilon](k;l)I_{a,b_*}[q_\epsilon](l;s) \ -\ \int_{-\infty}^{\infty}dx |u_{b_*}(x;0)|^2 Q_\epsilon(x)q_\epsilon(x) \right|\lesssim \epsilon^{2+r} .\]
 Lemma~\ref{lem:LOT-expanded} is now an immediate consequence of the definition of $B_{\eff}$ in Hypothesis (H2'),~\eqref{as:Qn0-q-epsilon-effective}.
\end{proof}

\subsection{Conclusion of the proof of Theorem~\ref{thm:per_result}}\label{subsec:Qnot0-conclusion}

Proposition~\ref{prop:Qn0-leading-order} is a formal reduction of the eigenvalue problem 
\begin{equation}\label{eq:orig-evp-final}
(-\partial_x^2 + Q(x) + q_\epsilon(x)) \psi^\epsilon = E^\epsilon \psi^\epsilon, \quad \psi^\epsilon \in H^2(\RR),
\end{equation}
for $(E^\epsilon,\psi^\epsilon)$ to 
an equation for $(\theta_\epsilon^2,\Phi_\epsilon)$ of the form:
\begin{equation}
\left( \frac12\partial_k^2 E_{b_*}(0) \kappa^2 + \theta_\epsilon^2 \right) \widehat{\Phi_\epsilon}(\kappa) - \chi(|\kappa|<\epsilon^{r-2}) B_{b_*,\eff} \times 
\int_{-\infty}^\infty \widehat{\Phi_\epsilon}(\xi) \ d\xi =  R_\flat [\theta_\epsilon]\widehat{\Phi_\epsilon} ( \kappa);
 \label{eq:Qn0-leading-orderA} 
\end{equation}
 (see~\eqref{eq:Qn0-leading-order}) where $\Phi_\epsilon$ is the rescaled near-energy component of $\psi^\epsilon$. We now apply Lemma~\ref{lem:technical} to obtain a solution of~\eqref{eq:Qn0-leading-orderA}. We then construct the solution $(E^\epsilon, \psi^\epsilon)$ of the full eigenvalue problem~\eqref{eq:orig-evp-final}. This will conclude the proof of Theorem~\ref{thm:per_result}.

\medskip

We apply Lemma~\ref{lem:technical} to equation~\eqref{eq:Qn0-leading-orderA}, with $A = \frac1{8\pi^2} \partial_k^2 E_{b_*}(k_*)$ and $B = B_{b_*,\eff}$ and $R_\epsilon=R_\flat$. By Proposition~\ref{prop:Qn0-leading-order}, $R_\epsilon$ satisfies assumption~\eqref{assumptionsR} with $\beta=2-r$, $\alpha=\sigma_1 = \min\{\sigma_\eff, r, 1/2 - 2r\}$. Following the steps of its proof, and using
\begin{align*}\left| \frac1{E_a(l)-E_\star+\epsilon^4 \theta_1^2}-\frac1{E_a(l)-E_\star+\epsilon^4 \theta_2^2}\right|&=\frac{\epsilon^4 |\theta_1^2-\theta_2^2|}{(E_a(l)-E_\star+\epsilon^4 \theta_1^2)(E_a(l)-E_\star+\epsilon^4 \theta_1^2)}\\
&\leq \frac{|\theta_1^2-\theta_2^2|}{\theta_2^2}\frac1{E_a(l)-E_\star+\epsilon^4 \theta_1^2},\end{align*}
one easily checks that assumption~\eqref{assumptionsR-theta} also holds.

Thus by Lemma~\ref{lem:technical} there exists a solution $\left( \theta_\epsilon^2, \widehat{\Phi}_{\epsilon} \right)$ of~\eqref{eq:Qn0-leading-orderA}, satisfying
\begin{equation}\label{eq:Qn0-bounds_rescaled_near}
\big\Vert\widehat{\Phi}_{\epsilon} - \widehat{f}_{0,\epsilon}\big\Vert_{L^{2,1}}\ \lesssim\ \epsilon^{\sigma_1} \qquad \text{ and } \qquad |\theta_\epsilon^2 - \theta_{0,\epsilon}^2|\ \lesssim\ \epsilon^{\sigma_1}\ . 
\end{equation} 
Here $\left( \theta_{0,\epsilon}^2, \widehat{f}_{0,\epsilon} \right)$ is the solution of the homogeneous equation 
\[
\widehat{\mathcal{L}} _{0,\epsilon}[\theta_{0,\epsilon}] \widehat{f}_{0,\epsilon} (\xi) = ( \frac1{2} \partial_k^2 E_{b_*}(k_*) \xi^2 + \theta_{0,\epsilon}^2) \widehat{f}_{0,\epsilon}(\xi) + \chi\left(|\xi|<\epsilon^{r-2}\right) B_{b_*,\eff} \int_{\RR}\chi\left(|\eta|<\epsilon^{r-2}\right)\widehat{f}_{0,\epsilon}(\eta)d\eta = 0,
\]
as described in Lemma~\ref{lem:homogeneous}. Specifically,
\begin{equation}\label{eq:f0-def2}
\widehat{f}_{0,\epsilon}(\xi) = \frac{\chi(|\xi|<\epsilon^{r-2})}{\frac1{2} \partial_k^2 E_{b_*}(k_*) \xi^2 + \theta_{0,\epsilon}^2}, \quad \text{ and } \quad \left|\theta_{0,\epsilon}^2 \ - \  \frac{B_{b_*,\eff}^2}{\frac1{2\pi^2} \partial_k^2 E_{b_*}(k_*)} \right|\ \lesssim\ \epsilon^{2-r}.
\end{equation}

We next construct the eigenpair solution $\left( E^\epsilon, \psi^\epsilon \right)$ of the Schr\"odinger equation~\eqref{eq:Qn0-start}. Define, using 
Proposition~\ref{prop:Qn0-far-intermsof-near}, 

\[
\psi^\epsilon \equiv \psi_\nr^\epsilon 
+ \psi_\fr^\epsilon,\qquad E^\epsilon\equiv E_{b_*}(k_*)-\epsilon^4\theta_\epsilon^2,\]
where
\[\wt{\psi}_\nr^\epsilon(\xi) = \frac{1}{\epsilon^2} \widehat{\Phi}_{\epsilon}\left(\frac{\xi}{\epsilon^2}\right) , \quad \psi_{\nr}^\epsilon=\mathcal{T}^{-1}\left\{\widetilde{\psi}_\nr^\epsilon(k)p_{b_*}(x;k)\right\}  \quad \text{ and } \quad \psi_\fr^\epsilon(\xi) = \psi_\fr [\wt{\psi}_\nr^\epsilon,E^\epsilon;\epsilon](\xi).
\]

Then $(E^\epsilon, \psi^\epsilon)$ is a solution of the eigenvalue problem~\eqref{eq:orig-evp-final}. Indeed, the steps proceeding from~\eqref{eq:orig-evp-final} to~\eqref{eq:Qn0-leading-orderA} are  reversible for solutions $\psi^\epsilon\in H^1(\RR)$
 of~\eqref{eq:orig-evp-final}, respectively, solutions $\Phi_\epsilon\in H^1(\RR)$ of~\eqref{eq:Qn0-leading-orderA}.

We now prove the estimates~\eqref{eq:Qnot0-eigenvalue-bound} and~\eqref{eq:Qnot0-eigenfunction-bound}. 
By~\eqref{eq:esttheta0} in Lemma~\ref{lem:homogeneous},~\eqref{eq:Qn0-bounds_rescaled_near} and recalling $E^\epsilon=E_{b_*}(k_*)-\epsilon^4\theta_\epsilon^2$, one has
\begin{align*} 
\left|E^\epsilon - \left( E_{b_*}(k_*)-\epsilon^4 \frac{B_{b_*,\eff}^2 }{\frac1{2\pi^2} \partial_k^2 E_{b_*}(k_*)} \right) \right| & =\epsilon^4 \left|\frac{B_{b_*,\eff}^2 }{\frac1{2\pi^2} \partial_k^2 E_{b_*}(k_*)} -\theta_\epsilon^2 \right| \\
 & \leq \epsilon^4 \left|\frac{B_{b_*,\eff}^2 }{\frac1{2\pi^2} \partial_k^2 E_{b_*}(k_*)} -\theta_{0,\epsilon}^2 \right| +\epsilon^4 \left|\theta_{0,\epsilon}^2-\theta_\epsilon^2 \right|  \lesssim \epsilon^{4+\sigma_1}.
\end{align*}
This shows estimate~\eqref{eq:Qnot0-eigenvalue-bound}, the small $\epsilon$ expansion of the eigenvalue $E^\epsilon$. 

The approximation,~\eqref{eq:Qnot0-eigenfunction-bound}, of the corresponding eigenstate, $\psi^\epsilon=\psi_\nr^\epsilon +\psi_\fr^\epsilon$, is obtained as follows. For $A = \frac1{8\pi^2} \partial_k^2 E_{b_*}(k_*)$ and $B = B_{b_*,\eff}$, one has
\begin{multline}
\left\Vert \psi^\epsilon(x) - u_{b_*}(x;0) \frac{2}{B} \exp \left( -\epsilon^2 \frac{B}{2A} |x| \right) \right\Vert_{L^\infty} 
\\
 \leq \left\Vert \psi_\nr^\epsilon(x) - u_{b_*}(x;0) \frac{2}{B} \exp \left( -\epsilon^2 \frac{B}{2A} |x| \right)  \right\Vert_{L^\infty} + \Big\Vert \psi_\fr^\epsilon(x) \Big\Vert_{L^\infty} .\label{bnd:final1}
\end{multline}
We look at each of the terms in~\eqref{bnd:final1} separately. 

Recall,
\begin{align*}
\psi_\nr^\epsilon(x) & = \int_{-1/2}^{1/2} \chi(|k|<\epsilon^r) \wt{\psi}_\nr (k) u_{b_*}(x;k)  \ dk \\
& = \int_{-1/2}^{1/2} \chi(|k|<\epsilon^r) \frac{1}{\epsilon^2} \widehat{\Phi}_\epsilon \left( \frac{k}{\epsilon^2}\right) e^{2\pi i kx} p_{b_*}(x;k)  \ dk \\
& = \int_\RR \chi(|\xi|<\epsilon^{r-2}) \widehat{\Phi}_\epsilon (\xi)  e^{2\pi i\epsilon^2\xi x} p_{b_*}(x;\epsilon^2 \xi)\ d\xi \\
& = u_{b_*}(x;0) \int_\RR \chi(|\xi|<\epsilon^{r-2}) \widehat{f}_{0,\epsilon} (\xi) e^{2\pi i\epsilon^2\xi x} \ d\xi \\
& \qquad \qquad + u_{b_*}(x;0) \int_\RR \chi(|\xi|<\epsilon^{r-2}) \left( \widehat{\Phi}_\epsilon (\xi) - \widehat{f}_{0,\epsilon} (\xi) \right) e^{2\pi i\epsilon^2\xi x} \ d\xi \\
& \qquad \qquad + \int_\RR \chi(|\xi|<\epsilon^{r-2}) \widehat{\Phi}_\epsilon (\xi) e^{2\pi i\epsilon^2\xi x}( p_{b_*}(x;\epsilon^2\xi)-p_{b_*}(x;0)) \ d\xi \\
& = I_1(x) + I_2(x) + I_3(x).
\end{align*}
We study each of these pieces in more detail. By~\eqref{eq:f0-def2}, $\chi(|\xi|<\epsilon^{r-2}) \widehat{f}_{0,\epsilon} (\xi) = \widehat{f}_{0,\epsilon} (\xi)$. Therefore, for $A = \frac1{8\pi^2} \partial_k^2 E_{b_*}(k_*)$ and $B = B_{b_*,\eff}$, one has from estimate~\eqref{eq:asymptotic-f0} in Lemma~\ref{lem:homogeneous},
\begin{equation}\label{bnd:I1}
\left| I_1(x) \ - \  u_{b_*}(x;0) \frac{2}{B} \exp \left( -\epsilon^2 \frac{B}{2A} |x| \right) \right| \ \lesssim \  \epsilon^{2-r}.
\end{equation}
Using the first bound of~\eqref{eq:Qn0-bounds_rescaled_near}, one has
\begin{align}
\big\Vert I_2 \big\Vert_{L^\infty} & = \sup_{x\in\RR} \left| u_{b_*}(x;0) \int_\RR \chi(|\xi|<\epsilon^{r-2}) \left( \widehat{\Phi}_\epsilon (\xi) - \widehat{f}_{0,\epsilon} (\xi) \right) e^{2\pi i\epsilon^2\xi x} \ d\xi \right| \nn \\
& \lesssim  \Vert u_{b_*}(x;0) \Vert_{L^\infty(\RR_x)} \big\Vert \widehat{\Phi}_\epsilon - \widehat{f}_{0,\epsilon}  \big\Vert_{L^{2,1}} \lesssim  \epsilon^{\sigma_1}. \label{bnd:I2}
\end{align}
Similarly,
\begin{align}
\big\Vert I_3 \big\Vert_{L^\infty} & = \sup_{x\in\RR} \left| \int_\RR \chi(|\xi|<\epsilon^{r-2}) \widehat{\Phi}_\epsilon (\xi) e^{2\pi i\epsilon^2\xi x} (p_{b_*}(x;\epsilon^2\xi)-p_{b_*}(x;0))\ d\xi \right| \nn \\
& \leq  \epsilon^2 \sup_{|k'| \leq \epsilon^r} \Vert \partial_k p_{b_*}(x;k') \Vert_{L^\infty(\RR_x)} \int_\RR \chi(|\xi|<\epsilon^{r-2}) |\xi| |\widehat{\Phi}_\epsilon (\xi)| \ d\xi \nn \\
& \lesssim \epsilon^2 \big\Vert \widehat{\Phi}_\epsilon \big\Vert_{L^{2,1}},\label{bnd:I3}
\end{align}
where we used that $\partial_k p_{b_*}(x;k')$ is well-defined and bounded by Lemma~\ref{lem:regularity-of-Eb}.
Finally, notice that $\big\Vert\widehat{\Phi}_{\epsilon} \big\Vert_{L^{2,1}} \to\big\Vert \widehat{f}_{0,\epsilon} \big\Vert_{L^{2,1}}$ as $\epsilon\to0$, and $\big\Vert \widehat{f}_{0,\epsilon} \big\Vert_{L^{2,1}}$ is bounded uniformly with respect to $\epsilon$. Therefore, from estimates~\eqref{bnd:I1}-\eqref{bnd:I3}, and noting that $\min\{\sigma_1, 2, 2-r\} = \sigma_1$, we can write
\begin{equation}
 \left\Vert \psi_\nr^\epsilon(x) - u_{b_*}(x;0) \exp \left( -\epsilon^2 \frac{B}{2A} |x| \right)  \right\Vert_{L^\infty}
 \leq C \ \epsilon^{\sigma_1}. \label{bnd:final1-1}
\end{equation}

The second term in~\eqref{bnd:final1} can be bound using~\eqref{bnd:Qn0-far-in-terms-near-Hs} in Proposition~\ref{prop:Qn0-far-intermsof-near} and~\eqref{eq:psi-to-Phi}: 
\begin{equation}\label{bnd:final1-2}
\Big\Vert \psi_\fr^\epsilon(x) \Big\Vert_{L^\infty} \lesssim\  \big\Vert \psi_\fr\big\Vert_{H^s}\ \lesssim\ \epsilon^{2 - s - 2r}\ \big\Vert \wt{\psi}_\nr\big\Vert_{L^2}\ \lesssim\ \epsilon^{1 -\max\{ s , 2r\}}\ \big\Vert\widehat{\Phi}_{\epsilon} \big\Vert_{L^{2,1}} \ \lesssim \ \epsilon^{1 -\max\{ s , 2r\}},
\end{equation}
for $1/2<s<3/2$, where we again note that $\big\Vert\widehat{\Phi}_{\epsilon} \big\Vert_{L^{2,1}} \to\big\Vert \widehat{f}_{0,\epsilon} \big\Vert_{L^{2,1}}$ as $\epsilon\to0$, and $\big\Vert \widehat{f}_{0,\epsilon} \big\Vert_{L^{2,1}}$ is bounded uniformly with respect to $\epsilon$.

Since $\psi^\epsilon$ is a unique solution of~\eqref{eq:orig-evp-final} up to a multiplicative constant, we can conclude from~\eqref{bnd:final1} and the estimates~\eqref{bnd:final1-1}-\eqref{bnd:final1-2}, that 
\begin{align*}
\left\Vert \psi^\epsilon(x) - u_{b_*}(x;0) \exp \left( -\epsilon^2 \frac{B}{2A} |x| \right) \right\Vert_{L^\infty} \lesssim \ \epsilon^{\sigma_2}, \qquad \sigma_2 = \min \{\sigma_\eff, r, 1-\max\{s,2r\} \}.
\end{align*}
This completes the proof of Theorem~\ref{thm:per_result}, with the choice $r=1/6$ and $s=2/3$.

\appendix

\section{Bounds used in Section~\ref{sec:Qnot0}}\label{sec:bounds-I}
To study the near- and far-energy equations~\eqref{eq:Qn0-near} and~\eqref{eq:Qn0-far}, we will make use of the following Lemmata.

\begin{Lemma}\label{lem:qpsi-bound}
Let $q_\epsilon \in L^2\cap L^\infty$ and assume $q_\epsilon$ is concentrated at high frequencies in the sense of~\eqref{as:Qn0-q-epsilon-decay}: There exists $\beta \geq 2$ and a constant $\mathcal{C}_\beta$ such that
\begin{equation}\label{eq:q-eps-osc-decay-2}
\left( \int_{- \frac{1}{2\epsilon}}^{ \frac{1}{2\epsilon}} \left| \widehat{q}_\epsilon(\xi) \right|^2 d\xi \right)^{1/2} \lesssim \mathcal{C}_\beta \epsilon^\beta, \ \text{ for } \ 0 < \epsilon \ll 1.
\end{equation}
Then, for any $\psi \in H^\delta$ with $\frac12<\delta \leq 2$, we have
\begin{equation}\label{eq:qpsi-bnd-genA}
\big\Vert \wt{q_\epsilon \psi} \big\Vert_{\mathcal{X}^{-\delta}}^2 \leq C(\big\Vert q_\epsilon \big\Vert_{L^2},\big\Vert q_\epsilon \big\Vert_{L^\infty},\mathcal{C}_{\beta}) \  \epsilon^{2\delta} \big\Vert \wt{\psi} \big\Vert_{\mathcal{X}^\delta}^2 . 
\end{equation}
If, moreover, $\psi \in H^2$, then we have
\begin{equation}\label{eq:qpsi-bnd-genB}
\big\Vert \wt{q_\epsilon \psi} \big\Vert_{\mathcal{X}^{-\delta}}^2 \leq C(\big\Vert q_\epsilon \big\Vert_{L^2},\big\Vert q_\epsilon \big\Vert_{L^\infty},\mathcal{C}_{\beta}) \ \left( \epsilon^{2\delta} \big\Vert \wt{\psi} \big\Vert_{\mathcal{X}^0}^2 + \epsilon^4 \ \big\Vert \wt{\psi}  \big\Vert_{\mathcal{X}^2}^2 \right). 
\end{equation}
\end{Lemma}
\begin{proof}
The norm $\mathcal{X}^s$ is defined in~\eqref{Xs-norm}. By Proposition~\ref{prop:norm_equivalence}, one has
\begin{align*}
\big\Vert \wt{q_\epsilon \psi} \big\Vert_{\mathcal{X}^{-\delta}}^2 &\lesssim \big\Vert q_\epsilon \psi \big\Vert_{H^{-\delta}}^2 \lesssim  \int_{\xi} \frac{1}{(1+|\xi|^2)^\delta} \left| \widehat{q_\epsilon \psi}(\xi)\right|^2 \ d\xi  .
\end{align*}
Bound the above by estimating the integral separately over the ranges: $|\xi|>\frac{1}{4\epsilon}$ and $|\xi|\leq\frac{1}{4\epsilon}$. For $|\xi|>\frac{1}{4\epsilon}$, one has
\begin{equation} \label{est-in-lem1} \int_{|\xi|>\frac{1}{4\epsilon}}\left| \frac{1}{(1+|\xi|^2)^\delta} \left| \widehat{q_\epsilon \psi}(\xi)\right|^2 \right|\ d\xi\ \lesssim \epsilon^{2\delta}\big\Vert \widehat{q_\epsilon \psi}(\xi)\big\Vert_{L^2}^2\lesssim \epsilon^{2\delta}\big\Vert q_\epsilon \psi\big\Vert_{L^2}^2\lesssim \epsilon^{2\delta}\big\Vert q_\epsilon \big\Vert_{L^\infty}^2\big\Vert\psi\big\Vert_{L^2}^2.
\end{equation}
For $|\xi|\leq\frac{1}{4\epsilon}$, we begin with a pointwise bound of $\widehat{q_\epsilon \psi}(\xi)$.
\begin{align*}
\widehat{q_\epsilon \psi}(\xi) & = \int_{\zeta} \widehat{q_\epsilon}(\zeta - \xi) \widehat{\psi}(\zeta) \ d\zeta \\
& = \int_{|\zeta|<1/(4\epsilon)} \widehat{q_\epsilon}(\zeta - \xi) \widehat{\psi}(\zeta) \ d\zeta \ + \ \int_{|\zeta|\geq 1/(4\epsilon)} \widehat{q_\epsilon}(\zeta - \xi) \widehat{\psi}(\zeta) \ d\zeta. \end{align*}
Since $q_\epsilon$ satisfies~\eqref{eq:q-eps-osc-decay-2}, one has for any $\gamma \in[0,2]$
\begin{align*}
\sup_{|\xi|<\frac{1}{4\epsilon}} \big|\widehat{q_\epsilon \psi}(\xi)\big| & \leq \epsilon^\beta \mathcal{C}_\beta \big\Vert \psi \big\Vert_{L^2} \ + \ \int_{|\zeta|\geq 1/(4\epsilon)} \frac{\widehat{q_\epsilon}(\zeta - \xi)}{(1+|\zeta|^2)^{\gamma/2}} (1+|\zeta|^2)^{\gamma/2} \widehat{\psi}(\zeta) \ d\zeta \\
& \leq   \epsilon^\beta \mathcal{C}_\beta \big\Vert \psi \big\Vert_{L^2} +\epsilon^{\gamma} \big\Vert q_\epsilon \big\Vert_{L^2} \big\Vert \psi \big\Vert_{H^\gamma}.
\end{align*}
Since $\beta \geq 2$ and $\delta>\frac12$, one deduces
\begin{equation} \label{est-in-lem2} \int_{|\xi|<\frac{1}{4\epsilon}}\left| \frac{1}{(1+|\xi|^2)^\delta} \left| \widehat{q_\epsilon \psi}(\xi)\right|^2 \right|\ d\xi\ \lesssim \sup_{|\xi|<\frac{1}{4\epsilon}}\big|\widehat{q_\epsilon \psi}(\xi)\big|^2\lesssim C(\big\Vert q_\epsilon \big\Vert_{L^2},\mathcal{C}_{\beta})\big(\epsilon^{4}\big\Vert\psi\big\Vert_{L^2}^2+\epsilon^{2\gamma}\big\Vert\psi\big\Vert_{H^\gamma}^2\big).
\end{equation}

Estimate~\eqref{eq:qpsi-bnd-genA} follows from~\eqref{est-in-lem1} and~\eqref{est-in-lem2} with $\gamma=\delta$. Estimate~\eqref{eq:qpsi-bnd-genB} follows from~\eqref{est-in-lem1} and~\eqref{est-in-lem2} with $\gamma=2$.
Lemma~\ref{lem:qpsi-bound} is proved. 
\end{proof}

We now turn to the study of 
\begin{equation}\label{def:Iab}
 I_{a,b}[q_\epsilon](k;l) \ \equiv \  \int_0^1 \overline{p_{a}(x;k)} \wt{q_\epsilon}(x;k-l)  p_b(x;l)\ dx.
\end{equation}
\begin{Lemma}\label{lem:I-summable}
Let $q_\epsilon\in L^2(\RR)$ and $Q$ be continuous. Then for any $k,l\in (-1/2,1/2]$, one has
\begin{equation}\label{eq:Iab-Ta-connection}
 I_{a,b}[q_\epsilon](k;l) =\mathcal{T}_a\{ u_b(\cdot;l) q_\epsilon(\cdot) \}(k)=\overline{\mathcal{T}_b\{ u_a(\cdot;k) q_\epsilon(\cdot) \}(l)}.
\end{equation}
For each fixed $b \geq 0$, we have the bounds:
\begin{equation}\label{eq:I-sum-a-bnd}
\sum_{a\geq0} \int_{-1/2}^{1/2} \left| I_{a,b}[q_\epsilon](k;l) \right|^2\ dk  \ \leq\  C(\big\Vert Q\big\Vert_{L^\infty_{\rm per}}) \big\Vert q_\epsilon \big\Vert_{L^2}^2 
\end{equation}
and
\begin{equation}\label{eq:I-sum-b-bnd}
\sum_{a\geq0} \int_{-1/2}^{1/2} \left| I_{b,a}[q_\epsilon](k;l) \right|^2\ dl  \ \leq\  C(\big\Vert Q\big\Vert_{L^\infty_{\rm per}})  \big\Vert q_\epsilon \big\Vert_{L^2}^2 .
\end{equation}
Note: The order of $(a,b)$ with respect to the variable of integration is important.
\end{Lemma}
\begin{proof}
Note that, by definition~\eqref{def:T-x;k}, we can write 
\[
\wt{q_\epsilon}(x;k-l) = \mathcal{T}\left\{ q_\epsilon(x) \right\}(x;k-l) = \mathcal{T} \left\{ e^{2\pi i l x} q_\epsilon (x) \right\}(x;k).
\]
Furthermore, by~\eqref{prop:T-per}, one has 
\begin{align*}
\wt{q_\epsilon}(x;k-l)  p_b(x;l) & = \mathcal{T} \left\{ e^{2\pi i l x}q_\epsilon (x) \right\}(x;k) p_b(x;l) \\
& = \mathcal{T} \left\{ e^{2\pi i l x}  p_b(x;l) q_\epsilon (x)  \right\}(x;k) \\
& = \mathcal{T} \left\{ u_b(x;l) q_\epsilon (x)  \right\}(x;k).
\end{align*}
Therefore, 
\begin{align}
I_{a,b}[q_\epsilon](k;l)  & = \int_0^1 \overline{p_a(x;k)} \wt{q_\epsilon}(x;k-l)  p_b(x;l)\ dx \nn \\
& = \int_0^1 \overline{p_a(x;k)} \mathcal{T} \left\{ u_b(x;l) q_\epsilon (x)  \right\}(x;k) \ dx \nn \\
& = \mathcal{T}_a\{ u_b(x;l) q_\epsilon (x) \}(k), \label{eq:Iab-equiv}
\end{align} 
which implies~\eqref{eq:Iab-Ta-connection}.

We now complete the proofs of the bounds~\eqref{eq:I-sum-a-bnd}-\eqref{eq:I-sum-b-bnd}. On the one hand,
\[
\left\Vert u_b(x;l) q_\epsilon (x) \right\Vert_{L^2(\RR_x)}^2 \lesssim \sup_{x,l} \left| u_b(x;l) \right|^2 \big\Vert q_\epsilon \big\Vert^2_{L^2} \lesssim  \ \leq\  C(\big\Vert Q\big\Vert_{L^\infty_{\rm per}}) \big\Vert q_\epsilon \big\Vert^2_{L^2},
\]
where we used Lemma~\ref{lem:estimates-uniform-in-b}.
On the other hand, by Proposition~\ref{prop:norm_equivalence},
\begin{align*}
\left\Vert u_b(x;l) q_\epsilon (x) \right\Vert_{L^2(\RR_x)}^2 &\approx \left\Vert \mathcal{T} \{ u_b(\cdot;l) q_\epsilon (\cdot) \} \right\Vert_{\mathcal{X}^0}^2 \\
&\equiv \sum_{a\geq0}\int_{-1/2}^{1/2}dk \left| \mathcal{T}_a\{u_b(\cdot;l) q_\epsilon (\cdot)\}(k) \right|^2 =  \sum_{a\geq0}\int_{-1/2}^{1/2}dk \left| I_{a,b}[q_\epsilon](k;l) \right|^2,
\end{align*}
where we used~\eqref{eq:Iab-equiv}. This implies~\eqref{eq:I-sum-a-bnd}. The bound~\eqref{eq:I-sum-b-bnd} follows by applying similar arguments to $I_{b,a}[q_\epsilon](k;l) = \overline{\mathcal{T}_a \{ u_b(\cdot;k) q_\epsilon (\cdot) \}(l) }$.
\end{proof}

\begin{Lemma}\label{lem:I-bounds}
Set $N\geq 0$. Assume that $q_\epsilon\in L^2$ and assume $q_\epsilon$ is concentrated at high frequencies in the sense of~\eqref{as:Qn0-q-epsilon-decay} with $\beta\geq N+1/2$. Assume $Q$ is such that~\eqref{as:Qn0-Q-pb} holds with $\alpha\in \NN,\alpha\geq N+1/2$. Let $k\in (-1/2,1/2]$, and $a,b\geq 0$. Then
\begin{enumerate}
\item
\begin{equation}\label{est-I-uniform}
\int_{-1/2}^{1/2} \left| I_{b,a}[q_\epsilon](k;l) \right|^2\ dl \leq C\big(C_0,\big\Vert q_\epsilon\big\Vert_{L^2}\big), 
\end{equation}
and
\begin{equation}\label{est-I-ab-bounded}
\int_{-1/2}^{1/2} \left| I_{b,a}[q_\epsilon](k;l) \right|^2\ dl \leq \mathcal{C} \epsilon^{2N} (1+|a|^{2\alpha})(1+| b|^{2\alpha}), 
\end{equation}
where $\mathcal{C}=C(\mathcal{C}_{\beta}, \ C_\alpha, \ \big\Vert q_\epsilon \big\Vert_{L^2})$.
% %
% %
\item
If $a$ is such that $\sqrt{E_{a}(k)}<\pi/(3\epsilon)$ ($a \lesssim \epsilon^{-1}$), then 
\begin{equation}\label{est-I-ab-small}
\int_{-1/2}^{1/2} \left| I_{b,a}[q_\epsilon](k;l) \right|^2\ dl \leq \mathcal{C} \epsilon^{3}(1+|b|^4), 
\end{equation}
where $\mathcal{C}=C(\mathcal{C}_{2}, \ C_2 , \ \big\Vert q_\epsilon \big\Vert_{L^2} )$. 

Furthermore, if $q_\epsilon \in H^2$, then  
\begin{equation}\label{est-I-ab-small-bis}
\int_{-1/2}^{1/2} \left| I_{b,a}[q_\epsilon](k;l) \right|^2\ dl  \leq \mathcal{C} \epsilon^{7}(1+|b|^4), 
\end{equation}
where $\mathcal{C}=C(\mathcal{C}_{4}, \ C_4 , \ \big\Vert q_\epsilon \big\Vert_{H^2} )$.
\end{enumerate}

All of these estimates are uniform in $a,b,k,\epsilon$ and hold, symmetrically, for $\int_{-1/2}^{1/2} \left| I_{a,b}[q_\epsilon](k;l) \right|^2\ dk$.
\end{Lemma}
\begin{proof}Estimate~\eqref{est-I-uniform} is a straightforward consequence of Lemma~\ref{lem:I-summable}. 

Let us turn to~\eqref{est-I-ab-bounded}. Fix $k \in (-1/2, 1/2]$.
We recall that $\wt{q_\epsilon}(x;k-l) =\sum_{n\in \ZZ} e^{2\pi i nx}\widehat{q_\epsilon}(k-l+n)$, therefore
\begin{align*}
 I_{b,a}[q_\epsilon](k;l) & = \int_0^1 \overline{p_{b}(x;k)} \left(\sum_{n\in \ZZ} e^{2\pi i nx}\widehat{q_\epsilon}(k-l+n)  \right)p_{a}(x;l)\ dx \\
&= \int_0^1 \overline{p_{b}(x;k)} \left(\sum_{|n|<1/(3\epsilon)} e^{2\pi i nx}\widehat{q_\epsilon}(k-l+n)  \right)p_{a}(x;l)\ dx \\
& \qquad \qquad \qquad \qquad + \int_0^1 \overline{p_{b}(x;k)} \left(\sum_{|n|\geq 1/(3\epsilon)} e^{2\pi i nx}\widehat{q_\epsilon}(k-l+n) \right) p_{a}(x;l)\ dx \\
& \equiv I_{b,a}^{(i)}[q_\epsilon](k;l) + I_{b,a}^{(ii)}[q_\epsilon](k;l).
\end{align*}

We now bound $\int_{-1/2}^{1/2} \left| I_{b,a}^{(i)}[q_\epsilon](k;l) \right|^2\ dl $.
\begin{align*}
\int_{-1/2}^{1/2} \left| I_{b,a}^{(i)}[q_\epsilon](k;l) \right|^2\ dl  & = \int_{-1/2}^{1/2} \left| \int_0^1 \overline{p_{b}(x;k)} \left(\sum_{|n|<1/(3\epsilon)} e^{2\pi i nx}\widehat{q_\epsilon}(k-l+n) \right) p_{a}(x;l)\ dx \right|^2 \ dl \\
& \leq \sup_{l,x} \left|\overline{p_{b}(x;k)}p_{a}(x;l)\right|^2 \int_{-1/2}^{1/2} \left| \sum_{|n|< 1/(3\epsilon)}  \widehat{q_\epsilon}(k-l+n) \right|^2 \ dl.
\end{align*}
By Cauchy-Schwarz inequality, one has
\[ 
\left| \sum_{|n|< 1/(3\epsilon)} 1\times \widehat{q_\epsilon}(k-l+n) \right|^2 \leq \big(\frac{2}{3\epsilon}+1\big) \sum_{|n|< 1/(3\epsilon)} |\widehat{q_\epsilon}(k-l+n)|^2.
\]
Therefore, by assumption~\eqref{as:Qn0-q-epsilon-decay}, one has
\begin{align*}
\int_{-1/2}^{1/2} \left| I_{b,a}^{(i)}[q_\epsilon](k;l) \right|^2\ dl  & \leq C(C_0) \epsilon^{-1} \sum_{|n|< 1/(3\epsilon)} \int_{-1/2}^{1/2} |\widehat{q_\epsilon}(k-l+n)|^2 \ dl \qquad [l' = n-l]\\
& \leq C(C_0) \epsilon^{-1} \int_{-1/(2\epsilon)}^{1/(2\epsilon)} |\widehat{q_\epsilon}(k+l')|^2 \ dl' \ \leq\  C(C_0,\mathcal{C}_{\beta})\epsilon^{2\beta-1} , 
\end{align*}
which concludes the first part of the estimate.

Turning to $\int_{-1/2}^{1/2} \left| I_{b,a}^{(ii)}[q_\epsilon](k;l) \right|^2\ dl$, we have for any $\alpha\in \NN$,
\begin{align*}
\int_{-1/2}^{1/2} \left| I_{b,a}^{(ii)}[q_\epsilon](k;l) \right|^2\ dl & = \int_{-1/2}^{1/2} \left| \int_0^1 \overline{p_{b}(x;k)} \left(\sum_{|n|\geq 1/(3\epsilon)} e^{2\pi i nx}\widehat{q_\epsilon}(k-l+n) \right) p_{a}(x;l)\ dx \right|^2 \ dl \\
& = \int_{-1/2}^{1/2} \left| \sum_{|n|\geq 1/(3\epsilon)} \widehat{q_\epsilon}(k-l+n) \int_0^1 \frac{1}{(2\pi in)^\alpha} \ e^{2\pi i nx} \ \partial_x^\alpha \left( \overline{p_{b}(x;k)}  p_{a}(x;l)\right) \ dx \right|^2 \ dl .
\end{align*}
By Cauchy-Schwarz inequality, one has for $\alpha\geq 1$:
\begin{align*}
\left| \sum_{|n|\geq 1/(3\epsilon)} \frac1{n^\alpha}\times \widehat{q_\epsilon}(k-l+n) \right|^2 &\leq \big(\sum_{|n|\geq 1/(3\epsilon)}\frac{1}{n^{2\alpha}}\big) \big(\sum_{|n|\geq 1/(3\epsilon)} |\widehat{q_\epsilon}(k-l+n)|^2\big)\\
&\lesssim \epsilon^{2\alpha-1}\sum_{|n|\geq 1/(3\epsilon)} |\widehat{q_\epsilon}(k-l+n)|^2.
\end{align*}
It follows, using~\eqref{as:Qn0-Q-pb}
\begin{align*}
\int_{-1/2}^{1/2} \left| I_{b,a}^{(ii)}[q_\epsilon](k;l) \right|^2\ dl & \lesssim C_\alpha^2 (1+| a|^\alpha )^2(1+ | b |^\alpha)^2 \epsilon^{2\alpha -1} \sum_{|n|\geq 1/(3\epsilon)} \int_{-1/2}^{1/2} |\widehat{q_\epsilon}(k-l+n)|^2 \ dl \qquad [l' = n-l]\\
&  \lesssim C_\alpha^2 (1+| a|^\alpha )^2(1+ | b |^\alpha)^2 \epsilon^{2\alpha -1} \big\Vert q_\epsilon \big\Vert_{L^2}^2.
\end{align*}
Estimate~\eqref{est-I-ab-bounded} follows with $\alpha,\beta\geq N+1/2$.
\medskip

In order to obtain~\eqref{est-I-ab-small}, we shall use the above estimate concerning $\int_{-1/2}^{1/2} \left| I_{b,a}^{(i)}[q_\epsilon](k;l) \right|^2\ dl$ (with $\beta=2$), and the refined analysis of Lemma~\ref{lem:asymptotics-in-b} below for $\int_{-1/2}^{1/2} \left| I_{b,a}^{(ii)}[q_\epsilon](k;l) \right|^2\ dl$. More precisely,
\begin{align*}
\int_{-1/2}^{1/2} \left| I_{b,a}^{(ii)}[q_\epsilon](k;l) \right|^2\ dl & = \int_{-1/2}^{1/2} \left| \int_0^1 \overline{p_{b}(x;k)} \left(\sum_{|n|\geq 1/(3\epsilon)} e^{2\pi i nx}\widehat{q_\epsilon}(k-l+n)  \right)p_{a}(x;l)\ dx \right|^2 \ dl \\
&=\int_{-1/2}^{1/2} \Bigg| \int_0^1    \overline{p_{b}(x;k)}  \sum_{|n|\geq 1/(3\epsilon)}  \widehat{q_\epsilon}(k-l+n)  \\
&\qquad \qquad \left(A^+_{a}(x;l)e^{ix(\sqrt{E_{a}(l)}-2\pi l + 2\pi n)} +A^-_{a}(x;l)e^{-ix(\sqrt{E_{a}(l)}+2\pi l-2\pi n)} \right)\ dx \Bigg|^2 \ dl.
\end{align*}
One deduces, after integrating by parts twice and using Lemma~\ref{lem:asymptotics-in-b} (since $Q\in W^{1,\infty}_{\rm per}(\RR)$),
\begin{align*}
\int_{-1/2}^{1/2} \left| I_{b,a}^{(ii)}[q_\epsilon](k;l) \right|^2\ dl &\lesssim C_2^2(1+| b|^2)^2  \int_{-1/2}^{1/2} \left| \sum_{n \geq 1/(3\epsilon)} \widehat{q_\epsilon}(k-l+n) \frac{1}{(2\pi n-n_0)^2}  \right|^2 \ dl,
\end{align*}
with $n_0=\sqrt{E_{a}(l)}-2\pi l<\pi (1/(3\epsilon)+1)$.
Once again, by Cauchy-Schwarz inequality, and since $|n|\geq 1/(3\epsilon)$, one deduces
\[\int_{-1/2}^{1/2} \left| I_{b,a}^{(ii)}[q_\epsilon](k;l) \right|^2\ dl \leq \epsilon^{3} C_2^2 (1+| b|^2)^2 \big\Vert q_\epsilon\big\Vert^2_{L^2},\]
and~\eqref{est-I-ab-small} is proved.

Estimate~\eqref{est-I-ab-small-bis} is proved similarly, but using that $q_\epsilon\in H^2$ implies 
 \begin{align*}&\int_{-1/2}^{1/2} \left| \sum_{n \geq 1/(3\epsilon)} \widehat{q_\epsilon}(k-l+n) \frac{1}{(2\pi n-n_0)^2}  \right|^2 \ dl\\
 &\quad =\int_{-1/2}^{1/2} \left| \sum_{n \geq 1/(3\epsilon)} (1+|k-l+n|^2)\widehat{q_\epsilon}(k-l+n) \frac{1}{(1+|k-l+n|^2)(2\pi n-n_0)^2}  \right|^2 \ dl\\
 &\quad \lesssim \epsilon^7 \big\Vert q_\epsilon \big\Vert_{H^2}^2.
 \end{align*}
Estimate~\eqref{est-I-ab-small-bis} follows as above, and Lemma~\ref{lem:I-bounds} is proved.
\end{proof}

\section{Detailed Information on the Floquet-Bloch states}\label{sec:bounds-FB}

In this section, we collect some information on the Floquet-Bloch states, as defined in Section~\ref{sec:background}, and that are used in Section~\ref{sec:Qnot0}.

These results will be based on the following identity, which follows from the variations of constants formula (see~\cite{eastham1973spectral})
\begin{multline} \label{VarConst}u_b(x;k)=u_b(0;k)\cos(\sqrt{E_b(k)} \ x)+ \partial_x u_b(0;k) \frac{\sin(\sqrt{E_b(k)} \ x)}{\sqrt{E_b(k)}} \\
+ \int_0^x \frac{\sin(\sqrt{E_b(k)} \ (x-y))}{\sqrt{E_b(k)}}Q(y)u_b(y;k)\ dy.
\end{multline}
Of course, this identity makes sense only for $E_b(k)>0$. However, since $Q \in L^\infty_{\rm per}$, one can always introduce a large enough constant $C$ such that $Q + C\geq 1$ and therefore $\spec(H_{Q+C})\subset [1,+\infty)$, and the above identity holds replacing  $\sqrt{E_b(k)}$ with $\sqrt{E_b(k)+C}$ and $Q$ with $Q+C$. {\em In this section, without loss of generality, we assume $\inf \spec (H_Q)\geq 1$. In particular, we have 
\[
E_b(k) \geq 1, \quad k \in (-1/2,1/2], \quad b \geq 0.
\]
}

\begin{Lemma}\label{lem:estimates-uniform-in-b}
Assume that $Q$ is continuous, one-periodic. Then,
\begin{align}
\sup_{k\in(-\frac12,\frac12]} \big\Vert  u_{b}(x;k)\big\Vert_{L^{\infty}(\RR_x)}\ &\leq \ C\left(\big\Vert Q\big\Vert_{L^\infty_{\rm per}}\right)\ , 
\label{pb-x-reg}\\
\sup_{k\in(-\frac12,\frac12]}\Big(\big\Vert \partial_x u_{b}(x;k)\big\Vert_{L^\infty(\RR_x)}+\big\Vert \partial_x u_{b}(x;k)\big\Vert_{L^2([0,1]_x)}\Big)\ &\leq \ C\left(\big\Vert Q\big\Vert_{L^\infty_{\rm per}}\right)\ \left(1+|b|\right)\ , \label{pb-d1x-reg}
\end{align}
 uniformly for $b\geq0$.
 
If moreover $Q\in W^{N-2,\infty}_{\rm per}(\RR)$ with $N\geq2$; then one has
\begin{equation}
\sup_{k\in(-\frac12,\frac12]}\Big(\big\Vert \partial_x^N u_{b}(x;k)\big\Vert_{L^\infty(\RR_x)}+\big\Vert \partial_x^N u_{b}(x;k)\big\Vert_{L^2([0,1]_x)}\Big)\ \leq \ C\left(\big\Vert Q\big\Vert_{W^{N-2,\infty}_{\rm per}}\right)\ \left(1+|b|^N\right)\ . \label{pb-dx-reg}
\end{equation}
\end{Lemma}
\begin{proof}
We first bound $|u_b(0;k)|$ and $|\partial_x u_b(0;k)|$. We will use these bounds to prove the estimate claimed in Lemma~\ref{lem:estimates-uniform-in-b} by applying them to $u_b(x;k)$ as expressed in~\eqref{VarConst}. 

First, integrate~\eqref{VarConst} against $\cos(\sqrt{E_b(k)} \ x)$ over the integration domain $[0,\pi/\sqrt{E_b(k)}]$.
\begin{align*}
\int_0^{\frac\pi{\sqrt{E_b(k)}}} \cos(\sqrt{E_b(k)}\ x) u_b(x;k) \ dx & = u_b(0;k) \int_0^{\frac\pi{\sqrt{E_b(k)}}} \cos^2(\sqrt{E_b(k)}\ x) \ dx \\
& \qquad + \partial_x u_b(0;k) \int_0^{\frac\pi{\sqrt{E_b(k)}}} \frac{\cos(\sqrt{E_b(k)}\ x) \sin(\sqrt{E_b(k)}\ x)}{\sqrt{E_b(k)}} \\
& \qquad + \int_0^{\frac\pi{\sqrt{E_b(k)}}} \cos(\sqrt{E_b(k)}\ x) \int_0^x \frac{\sin(\sqrt{E_b(k)}\ (x-y))}{\sqrt{E_b(k)}} Q(y) u_b(y;k) \ dydx. 
\end{align*}
Since $\pi/\sqrt{E_b(k)}\lesssim 1 $ and $|\sin(x)/x|\leq 1$ for any $x\in\RR$, we deduce
\begin{equation}
\left| u_b(0;k)  \right| \leq C(\Vert Q \Vert_{L^\infty_{\rm per}} ) \left\Vert u_b(x; k) \right\Vert_{L^2([0,1]_x)}  . \label{eq:ub-bnd}
\end{equation}
Similarly, integrate~\eqref{VarConst} against $\sin(\sqrt{E_b(k)} \ x)$ over $x\in[0,\pi/\sqrt{E_b(k)}]$, and deduce
\begin{equation}\label{eq:dx-ub-bnd}
|\partial_x u_b(0;k)| \lesssim \sqrt{E_b(k)} C(\Vert Q \Vert_{L^\infty_{\rm per}} ) \left\Vert u_b(x; k) \right\Vert_{L^2([0,1]_x)} . 
\end{equation}
Plugging~\eqref{eq:ub-bnd}-\eqref{eq:dx-ub-bnd} back into~\eqref{VarConst} yields~\eqref{pb-x-reg}.

Differentiating~\eqref{VarConst} yields
\begin{multline*} \partial_x u_b(x;k)=-\sqrt{E_b(k)} u_b(0;k)\sin(\sqrt{E_b(k)} \ x)+\partial_x u_b(0;k)\cos(\sqrt{E_b(k)} \ x)\\
+\int_0^x\cos(\sqrt{E_b(k)} \ (x-y))Q(y)u_b(y;k)\ dy.
\end{multline*}
Estimate~\eqref{pb-d1x-reg} then follows from the bounds~\eqref{eq:ub-bnd}-\eqref{eq:dx-ub-bnd}, as well as Weyl's asymptotics (see Lemma~\ref{lem:Weyl}): $\sqrt{E_b(k)}\approx b$.

Estimate~\eqref{pb-dx-reg} is then obtained by induction on $N$ and differentiating $N-2$ times the identity
\[ -\partial_x^2 u_b(x;k)+Q(x)u_b(x;k)=E_b(k)u_b(x;k).\]
This completes the proof of Lemma~\ref{lem:estimates-uniform-in-b}.
\end{proof}

We now give more precise asymptotics for $u_b(x;k)$ when $b$ is large.
\begin{Lemma}\label{lem:asymptotics-in-b}
Assume that $Q$
 is continuous, one-periodic and $Q\in W^{1,\infty}(\RR)$.
Then one can write $p_b(x;k) = e^{-2\pi ikx} u_b(x;k)$ as
\[ p_b(x;k)=A^+_b(x;k)e^{ix(\sqrt{E_b(k)}-2\pi k)}+A^-_b(x;k)e^{ix(-\sqrt{E_b(k)}-2\pi k)}\]
with
\begin{align}
\begin{split}\label{pb-asympt}
\sup_{k\in(-\frac12,\frac12]}\big\Vert A^\pm_b(x;k)\big\Vert_{W^{2,\infty}([0,1]_x)}  & \leq  C \left( \big\Vert Q\big\Vert_{W^{1,\infty}_{\rm per}}\right) , \\
\sup_{k\in(-\frac12,\frac12]}\big\Vert \partial_x  A^\pm_b(x;k)\big\Vert_{L^{\infty}([0,1]_x)}  & \leq  C \left( \big\Vert Q\big\Vert_{L^{\infty}_{\rm per}}\right) \frac{1}{1+|b|}, 
\end{split}
\end{align}
uniformly for $b\geq 0$.
\end{Lemma}
\begin{proof}
Following~\eqref{VarConst}, we set
\begin{align*} A^+_b(x;k)&=\frac12 u_b(0;k)+\frac1{2i}\frac{\partial_x u_b(0;k)}{\sqrt{E_b(k)}}
+\frac1{2i\sqrt{E_b(k)}}\int_0^xe^{-y\sqrt{E_b(k)}}Q(y)u_b(y;k)\ dy,\\
A^-_b(x;k)&=\frac12 u_b(0;k)-\frac1{2i}\frac{\partial_x u_b(0;k)}{\sqrt{E_b(k)}}
-\frac1{2i\sqrt{E_b(k)}}\int_0^x e^{y\sqrt{E_b(k)}}Q(y)u_b(y;k)\ dy.
\end{align*} 

Using Weyl's asymptotics (Lemma~\ref{lem:Weyl}), we can approximate $\sqrt{E_b(k) }\approx b$. Therefore, by bounds~\eqref{eq:ub-bnd}-\eqref{eq:dx-ub-bnd},  
\[
\sup_{k\in (-1/2,1/2]} \left\Vert A^{\pm}_b(x;k) \right\Vert_{L^\infty([0,1]_x)} \leq C.
\]
Differentiating $A^{\pm}_b(x;k)$ once with respect to $x \in [0,1]$ yields
\[
\partial_x A^{\pm}_b(x;k) = \pm \frac1{2i\sqrt{E_b(k)}} e^{\mp x\sqrt{E_b(k)}}Q(x)u_b(x;k).
\]
By Weyl's asymptotics and result~\eqref{pb-x-reg} of Lemma~\ref{lem:estimates-uniform-in-b}, one has 
\[
\sup_{k\in (-1/2,1/2]} \left\Vert\partial_x A^{\pm}_b(x;k) \right\Vert_{L^\infty([0,1]_x)} \leq C \left( \Vert Q \Vert_{L^\infty_{\rm per}} \right) \ \frac{1}{1+|b|}.
\]
Furthermore, differentiating once more the above identity yields
\[
\partial_x^2 A^{\pm}_b(x;k) = e^{\mp x\sqrt{E_b(k)}} \left( - \frac{1}{2i} Q(x) u_b(x;k) \pm \frac{Q'(x) u_b(x;k)+Q(x) \partial_x u_b(x;k)}{2i\sqrt{E_b(k)}} \right).
\]
Once more, using Weyl's asymptotics along with results~\eqref{pb-x-reg}-\eqref{pb-d1x-reg} from Lemma~\ref{lem:estimates-uniform-in-b}, one has
\[
\sup_{k\in (-1/2,1/2]} \left\Vert\partial_x^2 A^{\pm}_b(x;k) \right\Vert_{L^\infty([0,1]_x)} \leq C \left( \Vert Q \Vert_{W^{1,\infty}_{\rm per}}\right).
\]
This concludes the estimates in~\eqref{pb-asympt} and the proof of Lemma~\ref{lem:asymptotics-in-b} is now complete.
\end{proof}

\section{Proof of Theorem~\ref{thm:qeps-specific}}\label{sec:effective-potential}

In Theorems~\ref{thm:Qzero} and~\ref{thm:per_result}, we have shown that the bifurcation of localized states into the spectral gaps is determined by effective parameters, $A_\eff$ and $B_\eff$. The former depends only on the background potential, $Q$, while the latter represents a dominant (resonant / non-oscillatory) contribution from $q_\epsilon$, through Hypothesis~(H2) or~(H2'), \textit{i.e.} assumption~\eqref{as:Q=0-q-epsilon-effective} or~\eqref{as:Qn0-q-epsilon-effective}.

In this section we compute $B_\eff$ for the particular case of two-scale functions $q_\epsilon(x) = q(x, x/\epsilon)$ that is almost-periodic and mean zero in the fast variable:
\begin{equation}\label{def:qeps-specific}
q_{\epsilon}(x) = \sum_{j \neq 0} q_j(x) e^{2\pi i \lambda_j \frac{x}{\epsilon}}, \quad \text{ with } \quad \inf_{j \neq l} |\lambda_j - \lambda_l| \geq \theta > 0, \quad  \inf_{j \neq 0} |\lambda_j| \geq \theta > 0.
\end{equation}
where $\theta>0$ is a constant. This immediately yields the result claimed in Theorem~\ref{thm:qeps-specific}.

\begin{Lemma}\label{lem:Qeps}
Assume that $q_\epsilon$ is as defined in~\eqref{def:qeps-specific} with $\sup_{j\neq0} \Vert (1 + |\xi|^3) \widehat{q_j}(\xi)  \Vert_{L^\infty(\RR_\xi)} \leq C < \infty$ and $ \sum_{j\neq0} \Vert (1 + |\xi|^2) \widehat{q_j}(\xi)  \Vert_{L^2(\RR_\xi)} \leq C < \infty$, and is real-valued. 
Recall that $Q_\epsilon(x)$ is defined with $\widehat{Q_\epsilon}(\xi)=\frac{\widehat{q_\epsilon}(\xi)}{1+4\pi^2 |\xi|^2}$.
Then for any $f\in W^{1,\infty}(\RR)$, one has
\begin{equation}\label{eq:est-Q2}
 \left|\int_{-\infty}^\infty f(x) q_\epsilon(x)Q_\epsilon(x) \ dx\ - \ \epsilon^2 \sum_{m\neq 0}\frac{1}{4\pi^2m^2 } \int_{-\infty}^\infty  f(x) |q_m|^2(x) \ dx \right| \lesssim \epsilon^3.
 \end{equation}
In particular, assumption (H2) in Theorem~\ref{thm:Qzero} (resp. (H2') in Theorem~\ref{thm:per_result}) hold with $\sigma_\eff=1$ and $B_\eff= \sum_{m\neq 0}\frac{1}{4\pi^2\lambda_m^2 } \int_{-\infty}^\infty  |q_m|^2(x) \ dx$ (resp. $B_\eff= \sum_{m\neq 0}\frac{1}{4\pi^2\lambda_m^2 } \int_{-\infty}^\infty | u_{b_*}(x;k_*)|^2 |q_m|^2(x) \ dx$).
\end{Lemma}
\begin{proof} Our first aim is to prove the following estimate:
\begin{equation}\label{eq:est-Q}
\left\Vert  Q_\epsilon(x)- \sum_{j\neq 0} \frac{\epsilon^2}{4\pi^2 |\lambda_j|^2}q_j(x) e^{2i\pi \lambda_jx/\epsilon} \right\Vert_{L^2} \lesssim \epsilon^3 ,
\end{equation}
Since $q_\epsilon(x) = \sum_{j \neq 0} q_j(x) e^{2\pi i\lambda_jx/\epsilon}$, one has $\widehat{q}_\epsilon(\xi) = \sum_{j\neq0} \widehat{q}_j(\xi - \lambda_j/\epsilon)$, and therefore
\[
\widehat{ Q_\epsilon}(\xi)  = \sum_{j\neq 0} \frac{\widehat{q_j}( \xi - \lambda_j/\epsilon)}{1 +4 \pi^2 \xi^2}.
\]
Similarly, denoting $Q_\epsilon^\dag (x)= \epsilon^2 \sum_{j\neq0} \frac{q_j(x)}{4\pi^2 \lambda_j^2} e^{2\pi i\lambda_jx/\epsilon}$, one has
\[
\widehat{Q_\epsilon^\dag}(\xi) = \epsilon^2 \sum_{j\neq0} \frac{1}{4\pi^2\lambda_j^2} \widehat{q_j}(\xi - \lambda_j/\epsilon).
\]
Defining $\mathcal{R}_\epsilon(x) \equiv Q_\epsilon^\dag(x) - Q_\epsilon(x)$, one has by Parseval's identity, 
\begin{align*}
\left\Vert \mathcal{R}_\epsilon \right\Vert_{L^2}^2 = \left\Vert \widehat{\mathcal{R}_\epsilon} \right\Vert_{L^2}^2 
& = \int_{-\infty}^{\infty} \left| \sum_{j\neq0} \widehat{q_j}(\xi - \lambda_j/\epsilon) \left[ \frac{\epsilon^2}{4\pi^2 \lambda_j^2} - \frac{1}{1 +4 \pi^2 \xi^2} \right] \right|^2 d\xi \\
& = \epsilon^4 \int_{-\infty}^{\infty} \left| \sum_{j\neq0} \frac{1}{4\pi^2\lambda_j^2} \widehat{q_j}(\xi - \lambda_j/\epsilon) \left[ \frac{1 + 4\pi^2 (\xi^2 - \lambda_j^2/\epsilon^2)}{1 +4 \pi^2 \xi^2} \right] \right|^2 d\xi .
\end{align*}
We consider the above integral over two domains: $|\xi| \leq \theta/(2\epsilon)$ and $|\xi| > \theta/(2\epsilon)$. 
For $|\xi| \leq \theta/(2\epsilon)$, one has $|1 + 4\pi^2 (\xi^2 - \lambda_j^2/\epsilon^2)| \lesssim \lambda_j^2/\epsilon^2$. By assumption, $\sup_{j\neq0} \Vert (1 + |\xi|^3) \widehat{q_j}(\xi)  \Vert_{L^\infty(\RR_\xi)} \leq C < \infty$. Therefore,
\begin{align*}
& \epsilon^4 \int_{|\xi| \leq \theta/(2\epsilon)} \left| \sum_{j\neq0} \frac{1}{4\pi^2\lambda_j^2} \widehat{q_j}(\xi - \lambda_j/\epsilon) \left[ \frac{1 + 4\pi^2 (\xi^2 - \lambda_j^2/\epsilon^2)}{1 +4 \pi^2 \xi^2} \right] \right|^2 d\xi \\
& \qquad \qquad \qquad \qquad \lesssim \int_{|\xi| \leq \theta/(2\epsilon)} \left| \frac{1}{1 +4 \pi^2 \xi^2} \right|^2 \ \left| \sum_{j\neq0} \widehat{q_j}(\xi - \lambda_j/\epsilon) \right|^2 d\xi \\
& \qquad \qquad \qquad \qquad \lesssim \sup_{j\neq0} \Vert (1 + |\xi|^3) \widehat{q_j}(\xi)  \Vert_{L^\infty(\RR_\xi)}^2 \int_{|\xi| \leq \theta/(2\epsilon)} \left| \frac{1}{1 +4 \pi^2 \xi^2} \right|^2 \ \left| \sum_{j\neq0} \frac{1}{1 + |\xi - \lambda_j/\epsilon|^3} \right|^2 d\xi.
\end{align*}
Since $|\lambda_j|\geq \theta |j|$, one has $|\xi - \lambda_j/\epsilon| \geq  \theta |j|/(2\epsilon)$ for $|\xi| \leq \theta/(2\epsilon)$ and thus $\sum_{j\neq0} \frac{1}{\lambda_j^2} \frac{1}{1 + |\xi - \lambda_j/\epsilon|^3} \lesssim \epsilon^3 \sum_{j\neq0} |j|^{-3}$. It follows 
\begin{equation}\label{eq:xi-small-bnd}
\epsilon^4 \int_{|\xi| \leq 1/(2\epsilon)} \left| \sum_{j\neq0} \frac{1}{4\pi^2\lambda_j^2} \widehat{q_j}(\xi - \lambda_j/\epsilon) \left[ \frac{1 + 4\pi^2 (\xi^2 - \lambda_j^2/\epsilon^2)}{1 +4 \pi^2 \xi^2} \right] \right|^2 d\xi 
 \lesssim \epsilon^6 \ \sup_{j\neq0} \Vert (1 + |\xi|^3) \widehat{q_j}(\xi)  \Vert_{L^\infty(\RR_\xi)}^2.
\end{equation}

For $|\xi| > \theta/(2\epsilon)$, one has $|1 + 4\pi^2 \xi^2| \gtrsim \epsilon^{-2}$ and therefore,
\begin{align*}
& \epsilon^4 \int_{|\xi| > 1/(2\epsilon)} \left| \sum_{j\neq0} \frac{1}{4\pi^2\lambda_j^2} \widehat{q_j}(\xi - \lambda_j/\epsilon) \left[ \frac{1 + 4\pi^2 (\xi^2 - \lambda_j^2/\epsilon^2)}{1 +4 \pi^2 \xi^2} \right] \right|^2 d\xi \\
& \qquad \qquad \qquad \qquad \lesssim \epsilon^8 \int_{|\xi| > 1/(2\epsilon)} \left| \sum_{j\neq0} \frac{1}{4\pi^2\lambda_j^2} (1 +4\pi^2 |\xi - \lambda_j/\epsilon|^2) \widehat{q_j}(\xi - \lambda_j/\epsilon) \left[ \frac{1 + 4\pi^2 (\xi^2 - \lambda_j^2/\epsilon^2)}{1 + |\xi - \lambda_j/\epsilon|^2} \right] \right|^2 d\xi.
\end{align*}
Now, one has 
\[
\left|\frac{1 + 4\pi^2 (\xi^2 - \lambda_j^2/\epsilon^2)}{1 + 4\pi^2|\xi - \lambda_j/\epsilon|^2} \right|\lesssim\left| \frac{1 + (\xi - \lambda_j/\epsilon)(\xi + \lambda_j/\epsilon)}{1 + 4\pi^2 |\xi - \lambda_j/\epsilon|^2} \right|\leq \left| \frac{\lambda_j}{\epsilon}\right|+1.
\]
Therefore, one obtains
\begin{equation}\label{eq:xi-large-bnd}
\epsilon^4 \int_{|\xi| > 1/(2\epsilon)} \left| \sum_{j\neq0} \frac{1}{4\pi^2\lambda_j^2} \widehat{q_j}(\xi - \lambda_j/\epsilon) \left[ \frac{1 + 4\pi^2 (\xi^2 - \lambda_j^2/\epsilon^2)}{1 +4 \pi^2 \xi^2} \right] \right|^2 d\xi \lesssim \epsilon^6 \left( \sum_{j\neq0} \Vert (1 + |\xi|^2) \widehat{q_j}(\xi) \Vert_{L^2(\RR_\xi)} \right)^2.
\end{equation}
Bounds~\eqref{eq:xi-small-bnd} and~\eqref{eq:xi-large-bnd} complete the proof of estimate~\eqref{eq:est-Q}. 

Estimate~\eqref{eq:est-Q2} is deduced as follows. Notice first that by~\eqref{eq:est-Q2} and Cauchy-Schwarz inequality, one has
 \[ \left|\int_{-\infty}^\infty dx f(x) q_\epsilon(x)\big( Q_\epsilon(x) \ - \ Q_\epsilon^\dag(x) \big) \right| \lesssim \big\Vert f\big\Vert_{L^\infty} \big\Vert q_\epsilon \big\Vert_{L^2}\big\Vert Q_\epsilon-Q_\epsilon^\dag \big\Vert_{L^2} \lesssim \epsilon^3.\]

Let us now consider 
\[\int_{-\infty}^\infty f(x)q_\epsilon(x)Q_\epsilon^\dag(x)  =\sum_{j\neq 0} \sum_{l\neq 0} \frac{\epsilon^2}{4\pi^2 \lambda_l^2 } \int_{-\infty}^\infty f(x) q_j(x)\overline{q_l}(x) e^{2i\pi (\lambda_j-\lambda_l)x/\epsilon},\]
where we used that since $q_\epsilon$ is real-valued, $\lambda_{-l}=-\lambda_l$ and $q_{-l}(x)=\overline{q_l}(x)$.

If $j\neq l$, then one has
  \begin{align*}
  \left| \frac{\epsilon^2}{4\pi^2 \lambda_l^2 } \int_{-\infty}^\infty dx f(x)q_j(x)\overline{q_l}(x) e^{2i\pi (\lambda_j-\lambda_l)x/\epsilon}   \right|
  &\lesssim \frac{\epsilon^3}{8\pi^3|\lambda_l^2(\lambda_j-\lambda_l)| } \int_{-\infty}^\infty dx \partial_x \big( f(x)q_j(x)\overline{q_l}(x)\big) e^{2i\pi (\lambda_j-\lambda_l)x/\epsilon}  \\
  &\lesssim \epsilon^3 \big\Vert f\big\Vert_{W^{1,\infty}}\frac{\big\Vert q_j\big\Vert_{W^{1,2}}\big\Vert q_l\big\Vert_{W^{1,2}}}{|\lambda_l^2(\lambda_j-\lambda_l)|}.
  \end{align*}
Since $\big\Vert q_j\big\Vert_{W^{1,2}} \lesssim  \big\Vert (1+|\xi|)\widehat{q_j}(\xi) \big\Vert_{L^2(\RR_\xi)}$, and by assumption, $\sum_{j\neq 0} \big\Vert (1+|\xi|^2)\widehat{q_j}(\xi)  \big\Vert_{L^2(\RR_\xi)}<\infty$, it follows
\[ \left|\int_{-\infty}^\infty dx f(x) \left(q_\epsilon(x)Q_\epsilon^\dag(x) \ - \ \epsilon^2 \sum_{m\neq 0}\frac{1}{4\pi^2\lambda_m^2 }  |q_m|^2(x) \right)   \right|\lesssim \epsilon^3.\]
The above estimates and triangular inequality immediately yield~\eqref{eq:est-Q2}, and the proof is complete.
\end{proof}

\end{document}